\documentclass[%
 reprint,
superscriptaddress,
 amsmath,amssymb,
 aps,
]{revtex4-2}

\usepackage{graphicx}% Include figure files
\usepackage{dcolumn}% Align table columns on decimal point
\usepackage{bm}% bold math
\usepackage{tikz,physics,dsfont,amsthm,mathtools,amssymb,comment}

\usepackage{amsthm,graphicx,mathtools,tikz,cancel, bbm,float,bm,subcaption,amsmath,multirow}
\usepackage[shortlabels]{enumitem}
\usepackage{algorithm}
\usepackage{algpseudocode}
\usepackage{chngcntr}

\usepackage{optidef}
\usepackage[hidelinks]{hyperref}
\usepackage{cleveref}
\usepackage{apptools}
\usepackage[T1]{fontenc}

\AtAppendix{\counterwithin{lemma}{section}}
\newcommand{\Guus}[1]{}
\newcommand{\bd}[1]{}
\newcommand{\todo}[1]{}

\newtheorem{proposition}{Proposition}
\newtheorem{lemma}{Lemma}

\newtheorem{corollary}{Corollary}
\newtheorem{theorem}{Theorem}
\theoremstyle{definition}
\newtheorem{definition}{Definition}
\DeclareMathOperator{\id}{$I$}

\begin{document}

\preprint{APS/123-QED}

\title{On the accuracy of twirled approximations in repeater chains}% Force line breaks with \\
%\thanks{A footnote to the article title}%

\author{Bethany Davies}
 \affiliation{QuTech, Delft Univ. of Technology, Lorentzweg 1, 2628 CJ Delft, The Netherlands}%Lines break automatically or can be forced with \\
\affiliation{Quantum Computer Science, EEMCS, Delft Univ. of Technology, Mekelweg 4, 2628 CD Delft, The Netherlands}
\affiliation{Kavli Institute of Nanoscience, Delft Univ. of Technology, Lorentzweg 1, 2628 CJ Delft, The Netherlands}

\author{Guus Avis}
\affiliation{QuTech, Delft Univ. of Technology, Lorentzweg 1, 2628 CJ Delft, The Netherlands}
\affiliation{Quantum Computer Science, EEMCS, Delft Univ. of Technology, Mekelweg 4, 2628 CD Delft, The Netherlands}
\affiliation{Kavli Institute of Nanoscience, Delft Univ. of Technology, Lorentzweg 1, 2628 CJ Delft, The Netherlands}
\affiliation{University of Massachusetts Amherst, Amherst, Massachusetts, USA}

\author{Stephanie Wehner}
\affiliation{QuTech, Delft Univ. of Technology, Lorentzweg 1, 2628 CJ Delft, The Netherlands}
\affiliation{Quantum Computer Science, EEMCS, Delft Univ. of Technology, Mekelweg 4, 2628 CD Delft, The Netherlands}
\affiliation{Kavli Institute of Nanoscience, Delft Univ. of Technology, Lorentzweg 1, 2628 CJ Delft, The Netherlands}

\date{\today}

\begin{abstract}
In the performance analysis of quantum networks, it is common to approximate bipartite entangled states as either being Bell-diagonal or Werner states.
We refer to these as \textit{twirled approximations} because it is possible to bring any state to such a form with a twirling map.
Although twirled approximations can simplify calculations, they can lead to an inaccuracy in performance estimates.
The goal of this work is to quantify this inaccuracy. 
We consider repeater chains where end-to-end entanglement is achieved by performing an entanglement swap at each repeater in the chain. 
We consider two scenarios: postselected and non-postselected entanglement swapping, where postselection is performed based on the Bell-state measurement outcomes at the repeaters. 
% In postselected swapping, the end-to-end state after entanglement swapping is conditioned on the Bell-state measurement outcome obtained when the swap is performed at each repeater. 
% In non-postselected swapping, the end-to-end state is obtained by taking a weighted average of all postselected outcomes, i.e. it is the expected end-to-end state. 
We show that, for non-postselected swapping, the Bell-diagonal approximation is exact for the computation of the Bell-diagonal elements of the end-to-end state. 
We also find that the Werner approximation accurately approximates the end-to-end fidelity when the infidelity of each initial state is small with respect to the number of repeaters in the chain. 
For postselected swapping, we find bounds on the difference in end-to-end fidelity from what is obtained with the twirled approximation, for initial states with a general noisy form. Finally, for the example of performing quantum key distribution over a repeater chain, we demonstrate how our insights can be used to understand how twirled approximations affect the secret-key rate. 
\end{abstract}

%\keywords{Suggested keywords}%Use showkeys class option if keyword
                              %display desired
\maketitle

%\tableofcontents

\section{\label{sec:intro} Introduction}
A common form of quantum repeater utilises \textit{entanglement swapping} \cite{zukowski1993event,wehner2018vision,munro2015inside}. Consider a simple scenario with two end nodes, each equipped with a single qubit, and an intermediate (repeater) node equipped with two qubits. Suppose that the repeater node shares the entangled two-qubit states $\rho_1$ and $\rho_2$ with each end node, such that the total initial state of the chain is $\rho_1 \otimes \rho_2$.
%When entanglement is shared between a repeater and two end-nodes, which 
An entanglement swap transforms $\rho_1 \otimes \rho_2$ into a two-qubit state $\rho'$ shared between the end nodes.
%We will refer to the intermediate node as a \textit{repeater node}, and in our work we will consider each repeater node to have two qubits. 
%The \textit{end nodes} each contain one qubit, between which the goal is to share entanglement.
%Such an operation is critical for overcoming the exponential decrease in entanglement generation rate with distance in the repeaterless case.
% This is a protocol that enables entanglement generation between two distantly separated qubits by placing an intermediate node between them that is equipped with two memory qubits. 
%Entanglement is initially generated on each side (between $A_k$ and $B_k$ for $k=1,2$). 
An entanglement swap consists of a Bell-state measurement (BSM) at the repeater node and classical communication of the BSM outcome to the end nodes, followed by local Pauli corrections at the end nodes. 
%is performed on the qubits of the central node and the outcome is classically communicated to the end-nodes. 
%A Pauli correction is then performed at node $B_2$. 
If the level of noise in the initial states, the BSM and the Pauli corrections is low enough, then the end-to-end state $\rho'$ will be entangled (see Figure \ref{fig:twirled_approx_in_swapping}a). 
%If the level of noise in the initial states, BSM and Pauli corrections is low enough, then the end-to-end state between $B_1$ and $B_2$ will be entangled (see Figure \ref{fig:twirled_approx_in_swapping}). 

In a \textit{repeater chain}, $N-1$ repeater nodes are placed between the end nodes. Entanglement is firstly shared between adjacent nodes, in the form of $N$ entangled two-qubit states $\bigotimes_{k=1}^N\rho_k$. End-to-end entanglement is achieved by performing an entanglement swap at each repeater node.
%Such a setup can help to overcome the exponential decay in entanglement generation rate with distance in the repeaterless case.
%Entanglement swapping provides a solution to the exponential decay in entanglement generation rate with distance \cite{dur1999repeaters}. 
%By placing multiple repeater stations between the two end-nodes and performing an entanglement swap on each one, 
% transforms two shorter-distance entangled states into a longer-distance entangled state by performing a Bell-state measurement (BSM) at the central node, classical communication of the measurement outcome, and the application of Pauli corrections at one of the end-nodes. 
% A \textit{repeater chain} is when multiple such measurement stations are place between the two parties, and entanglement is generated by performing a swap at each intermediate station.

Swapping-based repeaters are the form of repeater most within experimental reach, and present-day demonstrations of quantum networks have distributed entanglement in such a way --- see e.g. \cite{pompili2021multinode, hermans2022teleportation}. 
For this reason, theoretical work on the design and performance analysis of large-scale quantum networks typically uses swapping-based repeaters as a basic assumption in the network model \cite{dur1999repeaters,jiang2007optimal,vardoyan2023qnum,vanmilligen2024utilizing, meng2024percolationnoisy,avis2023requirements,victora2023entanglement,avis2025stochastic,inesta2023continuous}. 
%In the performance analysis of large-scale quantum networks, the form of the initial states in the repeater chain is typically approximated 

By the \textit{Bell-diagonal} and \textit{Werner} approximations of the two-qubit state $\rho$, we respectively refer to the states
\begin{align}
    \mathcal{B}(\rho) &= \sum_{i,j=0}^1 \lambda_{ij}\ketbra{\Psi_{ij}}, \label{eqn:twirled_approx_BD} \\
    \mathcal{W}(\rho) &= \frac{4F-1}{3} \ket{\Psi_{00}}\bra{\Psi_{00}} + \frac{(1-F)}{3}I_4
    \label{eqn:twirled_approx_werner}
\end{align}
such that $\lambda_{ij} = \bra{\Psi_{ij}}\rho \ket{\Psi_{ij}}$, $F=\lambda_{00}$ is the fidelity with respect to $\ket{\Psi_{00}}$, and $I_4$ is the identity matrix. Both approximations are equivalent to applying the symmetrising map $\mathcal{B}$ ($\mathcal{W}$) to $\rho$ \cite{werner1989correlations,bennett1996mixed}, which is also known as \textit{twirling} the state $\rho$. We therefore refer to the Bell-diagonal and Werner approximations as \textit{twirled approximations}. 
%When modelling a repeater chain, it can be convenient to assume that the initial states are either in either \textit{Bell-diagonal} form of \textit{Werner} form. 
%Alternatively, a Bell-diagonal state is obtained when $\ket{\Psi_{00}}$ is subject to Pauli noise, and a Werner state is obtained when $\ket{\Psi_{00}}$ is subject to depolarising noise.

%e refer to the  \textit{Bell-diagonal approximation} of a general two-qubit state $\rho$ as the state obtained by removing its off-Bell-diagonal elements. Similarly, we refer to the \textit{Werner approximation} of $\rho$ as the Werner state with fidelity $F=\bra{\Psi_{00}}\rho\ket{\Psi_{00}}$ (Definition \ref{def:twirled_approximations}). In this manuscript, by \textit{fidelity} we always refer to the fidelity to the target maximally entangled state $\ket{\Psi_{00}}$.
%See Definition \ref{def:twirled_approximations} for the precise definition. 
%We denote the twirling maps as $\mathcal{B}$ ($\mathcal{W}$),  and so the Bell-diagonal (Werner) approximation of $\rho$ is written as $\mathcal{B}(\rho)$ ($\mathcal{W}(\rho)$). 

When modelling states in a repeater chain, it can be convenient to use twirled approximations for the initial states of the chain. See Figure \ref{fig:twirled_approx_in_swapping}b for an illustration of how twirled approximations are used in a repeater chain. Without the the approximation, we input a pair of two-qubit initial states $\rho_1\otimes \rho_2$, and after the entanglement swap we have end-to-end state $\rho'$. With the approximation, each initial state is twirled with the map $\mathcal{B}$ ($\mathcal{W}$), and after the entanglement swap the end-to-end state is $\rho'_{\mathcal{B}}$ ($\rho'_{\mathcal{W}}$).

Twirled approximations have a symmetrised form, which requires only a few parameters to be specified and has a direct operational interpretation. For the Bell-diagonal approximation (\ref{eqn:twirled_approx_BD}), there are three parameters $\lambda_{01}$, $\lambda_{11}$, and $\lambda_{10}$, which are interpreted as the probabilities of $X$, $Y$ and $Z$ errors when applying a Pauli channel to the state $\ket{\Psi_{00}}$ to obtain the noisy state $\mathcal{B}(\rho)$. For the Werner approximation (\ref{eqn:twirled_approx_werner}), only a single parameter is required, which is the fidelity $F$ \cite{werner1989correlations}. 
Another important property is that the symmetric form is preserved after entanglement swapping, i.e. $\rho'_{\mathcal{B}}$ is Bell-diagonal and $\rho'_{\mathcal{W}}$ is Werner.
Consequently, $\rho'_{\mathcal{B}}$ ($\rho'_{\mathcal{W}}$) is often simpler to compute than $\rho'$, which has many advantages in the performance analysis of large-scale quantum networks, potentially with complex topologies. 
In particular, twirled approximations enable the analytical study of high-level network performance metrics, because the metrics may then be more easily understood as a function of the initial states and therefore low-level properties of the hardware \cite{vardoyan2023qnum,goodenough2024swapasap,davies2024buffering,guedes2024sequentialswaps,vanmilligen2024utilizing, meng2024percolationnoisy}. 
Twirled approximations enable a more efficient numerical simulation and optimisation of large-scale quantum networks \cite{victora2023entanglement,avis2025stochastic,inesta2023continuous,haldar2024fast,addala2025faster,zwerger2017trappedions}. They are also used to avoid making overly specific assumptions when modelling quantum hardware \cite{dasilva2024requirements,abruzzo2013SKR,zwerger2017trappedions}, because any noise model can in principle be transformed into such a case with twirling.
%mention quantum purification?
%Furthermore, quantum purification protocols 

\begin{figure}[t]
\includegraphics[width=80mm]{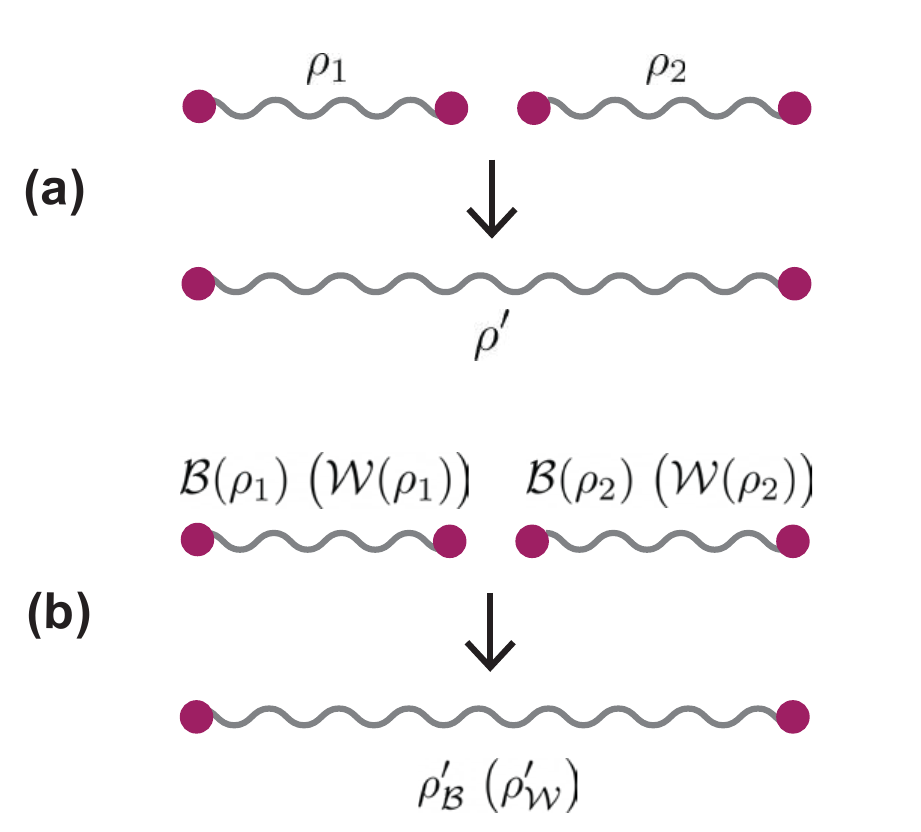}
\caption{\raggedright \label{fig:twirled_approx_in_swapping} \textbf{Twirled approximations in a repeater chain with $N=2$ initial states.} (a) The entanglement swapping of a pair of two-qubit entangled states $\rho_1\otimes\rho_2$ results in the end-to-end state $\rho'$. (b) With the Bell-diagonal (Werner) approximation $\mathcal{B}(\rho_k)$ ($\mathcal{W}(\rho_k)$) of the initial states, the end-to-end state $\rho'_{\mathcal{B}}$ ($\rho'_{\mathcal{W}}$) is simpler to compute than in case (a).}
\end{figure}

%Despite their advantages, the use of twirled approximations in a repeater chain can cause inaccuracies in performance estimates.
%, when compared to the same scenario without the approximation.
%The goal of this work is to quantify this underestimate in performance. 
% For example, let us consider swapping a pair of two-qubit states $\rho_1\otimes\rho_2$, such that $\rho_k$ has fidelity $F$, for $k=1,2$. Using the Werner approximation for both states in this scenario, we swap two states of the form (\ref{eqn:twirled_approx_werner}), and the end-to-end state $\rho'_{\mathcal{W}}$ is therefore also Werner with fidelity $F'_{\mathcal{W}}$, where \cite{munro2015inside} $$ F'_{\mathcal{W}} = \bra{\Psi_{00}} \rho'_{\mathcal{W}} \ket{\Psi_{00}}  =  F^2 + \frac{(1-F)^2}{3}.$$ However, depending on the exact form of the $\rho_k$, the true fidelity $F'$ of the end-to-end state $\rho'$ can exhibit a potentially significant variation above or below the value $F'_{\mathcal{W}}$. 
Despite their advantages, twirled approximations can cause inaccuracies in performance estimates.
For example, suppose that the initial states $\rho_1\otimes \rho_{2}$ have the same fidelity $\bra{\Psi_{00}}\rho_k\ket{\Psi_{00}}=F$, for $k=1,2$. 
Using the Werner approximation for both states in this scenario, the initial state is $\mathcal{W}(\rho_1)\otimes \mathcal{W}(\rho_2)$, and the end-to-end state $\rho'_{\mathcal{W}}$ is therefore also Werner with fidelity $F'_{\mathcal{W}} = \bra{\Psi_{00}}\rho'_{\mathcal{W}}\ket{\Psi_{00}} =  F^2 + (1-F)^2/3$ \cite{munro2015inside}. However, depending on the exact form of the initial states, the true fidelity $F' = \bra{\Psi_{00}}\rho'\ket{\Psi_{00}}$ of the end-to-end state can lie (potentially significantly) above or below the value of $F'_{\mathcal{W}}$.
The same holds if we use the Bell-diagonal approximations to obtain end-to-end fidelity $F'_{\mathcal{B}}$.
The principle question that we address in our work is: \textit{what is the maximum difference $|F'-F'_{\mathcal{B}}|$ ($|F'-F'_{\mathcal{W}}|$) between the true end-to-end fidelity $F'$ and the end-to-end fidelity with the twirled approximation $F'_{\mathcal{B}}$ ($F'_{\mathcal{W}}$)?} 

% In particular, the principle question we will address in this work is: \textit{what is the maximum deviation in the true post-swap fidelity $F'$ from the post-swap fidelity with the twirled approximation $F'_{\mathcal{B}}$ ($F'_{\mathcal{W}}$)?} 
%In particular, $F' = \bra{\Psi_{00}}\rho' \ket{\Psi_{00}}$, $F'_{\mathcal{B}} = \bra{\Psi_{00}}\rho'_{\mathcal{B}} \ket{\Psi_{00}}$ and $F'_{\mathcal{W}} = \bra{\Psi_{00}}\rho'_{\mathcal{W}} \ket{\Psi_{00}}$
% For example, let us consider swapping a pair of two-qubit states $\rho_1$ and $\rho_2$, each with fidelity $F$. Using the Werner approximation for both states in this scenario, we swap two Werner states with fidelity $F$, and the end-to-end state $\rho'_{\mathcal{W}}$ is therefore also Werner with fidelity $F'_{\mathcal{W}} =  F^2 + (1-F)^2/3$ \cite{munro2015inside}. However, depending on the exact form of the $\rho_k$, the true post-swap fidelity $F'$ can exhibit a (potentially significant) variation above or below this value. 
% In this work, we are interested in quantifying the maximum deviation from the twirled approximations, $|F' - F'_{\mathcal{B}}|$ and $|F' - F'_{\mathcal{W}}|$. 
%To ask the same question in the case of Bell-diagonal twirling, as well as the fidelity we fix the remaining Bell-diagonal elements.

We consider two different scenarios: \textit{postselected} and \textit{non-postselected} swapping. 
In postselected swapping, the end-to-end state $\rho'_{\vec{s}}$ is postselected on the BSM outcomes $\vec{s} = (s_1,\dots,s_{N-1})$ obtained at the $N-1$ repeater nodes. 
In non-postselected swapping, the end-to-end state is a weighted average of the postselected outcomes. Let $p'_{\vec{s}}$ be the probability of measuring $\vec{s}$. Then, the end-to-state after non-postselected swapping is $\sum_{\vec{s}} p'_{\vec{s}} \rho'_{\vec{s}}$. This is the expected end-to-end state after swapping, with the expectation taken over all BSM outcomes. We note that in both schemes, the relevant Pauli corrections are still applied for each syndrome $\vec{s}$ to rotate the state to the target maximally-entangled state $\ket{\Psi_{00}}$.

% We note that this is equivalent to not having knowledge of the specific BSM outcomes $\vec{s}$ obtained (although we note that, in the swapping protocol, the BSM outcomes are still communicated and Pauli corrections are still applied). 
%\todo{Explain relevance without controversy with Stephanie?}
%We will see that it is natural to carry out the study in two different scenarios: \textit{unconditional} and \textit{conditional} swapping. 
%In conditional swapping, the end-to-end state is postselected on the BSM outcome obtained at the repeater nodes. 
%In unconditional swapping, the BSM outcome is not known, and the end-to-end state is given by a mixture of the conditional outcomes. We note that in an unconditional swap, the BSM outcome is still communicated and Pauli corrections are applied to complete the swapping protocol, only the information of the outcome is not retained. 
%For example, a user of a quantum network may not be provided with this information because entanglement generation may be handled by lower layers of the quantum network stack and therefore seen as a black box by the network user \cite{pompili2022delivery}.

Having introduced the problem, we now outline our main contributions. 
In Section \ref{sec:advantages_full}, we study the case of \textbf{non-postselected swapping} on a chain with $N$ initial states $\bigotimes_{k=1}^N\rho_k$ and $N-1$ repeaters, we consider a general class of entanglement-swapping protocols that we term \textit{swap-and-correct} protocols (see Definition \ref{def:swap_and_correct_protocol}). Swap-and-correct protocols consist of BSMs and Pauli corrections that can be applied at any node in the chain. For all such protocols, we show that:
    %with a general entanglement swapping protocol on a repeater chain of arbitrary length, we show that:
    \begin{enumerate}[(i)]
        \item $\mathcal{B}(\rho') = \rho_{\mathcal{B}}'$, i.e. the Bell-diagonal approximation is exact for the computation of the Bell-diagonal components of the end-to-end state (Theorem \ref{thm:swap_and_correct_equivalence}). The Bell-diagonal components include the fidelity, and so $F' = F'_{\mathcal{B}}$.
        \item If the initial fidelities $F_k = \bra{\Psi_{00}}\rho_k \ket{\Psi_{00}}$ satisfy $1-F_k\ll 1/N$ for all $k=1,...,N$, then $F' \approx F'_{\mathcal{W}}$, i.e. the Werner approximation accurately approximates the end-to-end fidelity. More precisely, we have $|F' - F'_{\mathcal{W}}| = \mathcal{O}((1-F)^2 N^2)$ (Theorem \ref{thm:werner_approximation_accuracy}). 
        % If, for all $k$, the initial fidelity $F_k = \bra{\Psi_{00}}\rho_k \ket{\Psi_{00}}$ satisfies $1-F_k\ll 1/N$, then $F' \approx F'_{\mathcal{W}}$, i.e. the Werner approximation is accurate for the computation of the end-to-end fidelity. More precisely, it approximates this quantity to $|F' - F'_{\mathcal{W}}| = \mathcal{O}((1-F)^2 N^2)$ (Theorem \ref{thm:werner_approximation_accuracy}). 
    \end{enumerate}
    A key insight is that, in many important cases, non-postselected swapping and the use of the Bell-diagonal approximation are equivalent.
    Such cases include protocols whose performance may be expressed solely in terms of the Bell-diagonal elements $\mathcal{B}(\rho')$ of the end-to-end state $\rho'$. 
    For example, consider the channel $\Lambda^{\mathrm{tel}}_{\rho'}$ induced by standard quantum teleportation with resource state $\rho'$ \cite{bennett1993teleporting}. In Proposition \ref{prop:teleportation_inv_twirling}, we show that $\Lambda^{\mathrm{tel}}_{\rho'} = \Lambda^{\mathrm{tel}}_{\mathcal{B}(\rho')}$. By result ($\mathrm{i}$), we thus have $\Lambda^{\mathrm{tel}}_{\mathcal{B}(\rho')} = \Lambda^{\mathrm{tel}}_{\rho_{\mathcal{B}}'}$. In particular, the Bell-diagonal approximation for each initial state in the chain is exact when subsequently performing teleportation over the end-to-end state. However, exactness does not hold for all applications: we also see that for quantum key distribution, using the twirled approximation for certain input states can lead to a large reduction in performance (Section \ref{sec:qkd}).
    %The first result justifies the use of the Bell-diagonal approximation for a scenario where the end-to-end state is subsequently used in an application or protocol that does not postselect on the BSM outcomes at each repeater node. In the same scenario, if it also holds that the infidelity of each initial state is small compared to $1/N$, where $N$ is the number of initial states in the chain, then the second result justifies the use of the Werner approximation.
    %The second result identifies a regime of the initial fidelity where the use of the Werner approximation in situations where the initial fidelity of 
    %In the above results, the Bell-diagonal (Werner) approximation refers to making the approximation for any or all of the initial states $\rho_k$.
    
    In Section \ref{sec:advantages_full}, we study the case of \textbf{postselected swapping} on a repeater chain with two initial states $\rho_1 \otimes \rho_2$ and one repeater, we restrict to swapping initial states of the form $\rho_k = p\ketbra{\Psi_{00}} + (1-p)\sigma_k$ such that $F = \bra{\Psi_{00}}\rho_k\ket{\Psi_{00}}$. The density matrix $\sigma_k$ is interpreted as an arbitrary noise term. We fix $F$ to perform a comparison with the twirled approximation, and also fix $p\in [0,F]$ to provide meaningful bounds. (It is necessary to fix a second parameter because for any $F$, there exists a state $\omega$ with $F=\bra{\Psi_{00}}\omega\ket{\Psi_{00}}$ such that there is a probability $p'_s>0$ of obtaining $\rho'_s = \ketbra{\Psi_{00}}$ when swapping the initial states $\omega^{\otimes 2}$ \cite{Bose99}. Thus, if $F$ is the only fixed parameter, it is always possible to obtain unit end-to-end fidelity with some non-zero probability.) Letting $F'_{s}$ denote the end-to-end fidelity postselected on BSM outcome $s$, we find:
    %In the case of conditional swapping, the end-to-end state $\rho'$ may depend on the BSM outcome. For the case $N=2$, we consider each initial state $\rho_k$ to be of a general noisy form that is parameterised by the fidelity $F$, as well as a parameter $p\in[0,F]$ such that $\rho_k = p\ketbra{\Psi_{00}} + (1-p)\sigma$, for $k=1,2$. The density matrix $\sigma$ is interpreted as an arbitrary noise term.
    \begin{itemize}
        \item[$(\mathrm{iii})$] A tight, analytical upper bound for the achievable end-to-end fidelity, $F'_{s} \leq 1-2p(1-F)$ (Theorem \ref{thm:tight_upper_bound}). 
        %This is found firstly by showing the inequality holds by applying a well-known identity for the fidelity of separable states \cite{horodecki1999general}. 
        We show tightness by finding an example of a state $\rho_{\mathrm{opt}}$ such that $ \rho_{\mathrm{opt}}^{\otimes 2}$ achieves this upper bound. The state $\rho_{\mathrm{opt}}$ has a physical interpretation: it corresponds to applying a $Y$-rotation (of a specified angle) with probability $1-p$ to a qubit of $\ket{\Psi_{00}}$.        
        %We find simple, analytical upper and lower bounds in terms of $F$ and $p$ for the achievable fidelity $F'_{s}$ (Theorem \ref{thm:tight_upper_bound} and Proposition \ref{prop:analytical_lower_bound}).
        %Due to the simple parameterisation of our noisy states, these bounds are easily comparable with the twirled approximations (see Section \ref{sec:bounds_discussion}).
        %In Section \ref{sec:analytical_bounds}, we provide upper and bounds in terms of $F$ and $p$ on the post-swap fidelity when swapping two states of such a form and show that the upper bound is tight.
        %These constitute bounds on the deviation between the true value of the post-swap fidelity and the corresponding fidelity after the twirled approximation.
        %We present an example of initial states $\rho_1\otimes \rho_{2}$ that saturate the upper bound, showing that it is tight (Theorem \ref{thm:tight_upper_bound}). 
        \item[$(\mathrm{iv})$] A simple, analytical lower bound for $F'_{s}$ in terms of $p$ and $F$. Unlike the upper bound, this is not tight. Moreover, we find a tighter lower bound for $F'_{s}$ by formulating the problem as a semi-definite program. We perform a symmetry reduction of the problem, enabling efficient computation of the bound.        
        %We find a tighter lower bound for $F'_{s}$ by formulating the problem as a semi-definite program (Section \ref{sec:sdp}).
    \end{itemize}
Our simple formulation with the parameters $p$ and $F$, as well as the efficiently computable bounds, allows for direct interpretation and comparison with twirled approximations (see Section \ref{sec:bounds_discussion}). For example, let us consider swapping the initial states $\rho_R^{\otimes 2}$, where $$\rho_R = p\ketbra{\Psi_{00}} + (1-p)\ketbra{01}.$$ The state $\rho_R$ (up to a local unitary rotation) closely approximates states generated in certain physical entanglement generation schemes \cite{campbell2008singleclick,cabrillo1999singleclick}.
%Because there are high-performing purification protocols tailored to the noise component $\ketbra{01}$, one may expect this 
It has $$p = \bra{\Psi_{00}}\rho_R \ket{\Psi_{00}} =  F.$$ By $(\mathrm{iii})$, we see that $$F'_{s}(\rho_R^{\otimes 2})\leq 1 - 2F(1 - F) =   F^2 + (1-F)^2.$$ 
Therefore, $F'_{s}(\rho_R^{\otimes 2}) - F'_{\mathcal{W}} \leq  2(1-F)^2/3$. For large $F$, we see that the Werner approximation does not cause a large reduction in the maximum end-to-end fidelity.
% Due to the simple parameterisation of the noisy states with $F$ and $p$, our results can be used to directly infer bounds on the maximum deviation from the twirled approximations, $|F'_{s} - F'_{\mathcal{B}}|$ and $|F'_{s} - F'_{\mathcal{W}}|$ (Section \ref{sec:bounds_discussion}). 
%in post-swap fidelity from the value obtained with the twirled approximation of the initial states. 
% (see Section \ref{sec:bounds_discussion}).    
%    \item We find a tighter lower bound for the  by formulating the problem as a semi-definite program (SDP).

Building on this example, in Section \ref{sec:bounds_discussion}, we provide further discussion of how our bounds may be used to assess whether the twirled approximation is accurate for given input states. In Section \ref{sec:qkd}, we further discuss the implications of our results for an explicit example application: specifically, we look at the impact of twirled approximations for the performance of quantum key distribution over a repeater chain.

\section{Related work}
\label{sec:related_work}
%Here, we outline previous work on  related problems. 
%conditional entanglement swapping can help you 
The idea that entanglement can increase (or decrease) after a postselected entanglement swap has existed for many years \cite{Bose99}. 
%In some contexts, this is also referred to as purification via entanglement swapping. 
Much work has since focused on a fundamental investigation of how much the entanglement can change after the entanglement swap, when compared to the initial states \cite{Oppliger2023,Song2014,Xie2022,Roa2014,Modlawska2008,Li2015,Kirby2016_2,Sen2005,Dajka2013, Bergou2021,cong2025formal}. To answer this question, in many studies, the concurrence has been used as a measure of entanglement \cite{Oppliger2023,Song2014,Roa2014,Modlawska2008,Li2015,Kirby2016_2,Sen2005, Bergou2021,cong2025formal}
%which is a true measure in the sense that it cannot be increased with local operations, and 
because for this there exists a computable formula in the two-qubit case \cite{wootters1998arbitraryconcurrence}. 
Other work has used the negativity \cite{dajka2013swapping}, or instead of measures of entanglement, used measures of quantum correlation \cite{Xie2022}. 
Entanglement swapping \cite{zukowski1993event} can be seen as applying teleportation to one end of an entangled state \cite{bennett1993teleporting, hu2023progress}. Much work has focused on the analysis and optimisation of quantum teleportation with a noisy resource state -- see e.g. \cite{brauer2024enhancing,oh2002noisyteleportation,bang2018fidelity,badziag2000local}. However, to our knowledge, no systematic comparison has been performed with the twirled approximation. 
%Moreover, the results of such a study would not necessarily be extendable to the case of entanglement swapping to two noisy states, because the 
% We note that in order to be extendable to the case of entanglement swapping of noisy states, one must consider noise also on the qubit to be teleported. This is done in \cite{}. However, \todo{explain why this is not applicable to our problem.}
Moreover, the idea of a postselected swap is related to that of probabilistic teleportation \cite{mor1999conclusiveteleportation,li2000probabilistic,agrawal2002probabilisticteleportation}, where a qubit may be teleported with maximum fidelity, even if the resource state not maximally entangled. This is often made possible by measuring in a non-maximally-entangled basis and postselecting based on the measurement outcome.
Other work has focused on the effect of noise in the resource state on probabilistic teleportation \cite{fortes2016probabilistic}. Again, to our knowledge, no systematic comparison of (probabilistic) teleportation has been carried out with the twirled approximation of the resource state.

%maybe remove below part as it only creates questions.
%argue: need to look at teleportation with noisy inputs 

% Moreover, studies of quantum teleportation usually focus on the average teleportation fidelity as a performance metric \cite{horodecki1999general}. 
% Although the average teleportation fidelity is in direct correspondence with the entanglement fidelity \cite{horodecki1999general}, which is the post-swap fidelity after teleporting one half of an entangled state, this can only take into account noise acting on one qubit. In entanglement swapping, we are interested in the more general case where noise may be acting on both qubits of the teleported state.

Finally, we note that twirling is a technique that is used widely outside of the context of repeater chains. For example, twirling is used in quantum error correction to reduce a general noise channel to a Pauli channel \cite{dur2005standardviadepol}. 
% , which is efficient to simulate classically \cite{gottesman1998GKthm}. 
There has been work on quantifying the accuracy of such an approximation for the calculation of the error correction threshold \cite{geller2013efficient,gutierrez2015comparison}. 
Twirling is also used as a simplifying step in security proofs \cite{pironio2009device}.
Twirling is also used in randomised benchmarking \cite{emerson2005scalable, helsen2022general}, not as an approximation to the noise model, but as a tool that can be applied to extract important information about a noisy gate set. %check jonas helsen thesis for intro?

%Moreover, such a scenario may not be directly used to study the case of entanglement swapping, where instead noise acts on two qubits of a state to be teleported,

%This is possible by only accepting certain Bell-state outcomes: if these are not obtained, then the protocol fails.  
% Although there has been work on the analysis and optimisation of probabilistic teleportation in the presence of a noisy resource state \cite{}, there has not been a systematic comparison carried out with the twirled approximation. Moreover, we do not view such work as being relevant for entanglement swapping because...

% This is because a noisy two-qubit entangled state cannot be necessarily rewritten as a perfect maximally entangled state with the noise acting on only one side. The propagation of noise after entanglement swapping can therefore not be understood with teleportation only.

%port-based teleportation?

%syndrome information can be helpful

% \begin{itemize}
%     \item Entanglement swapping and changes in entanglement measures: 
%     \cite{Oppliger2023,Song2014,Xie2022,Kirby2016,Roa2014,Modlawska2008,Li2015,Kirby2016_2,Kirby2016_3,Sen2005,Dajka2013, Bergou2021}, \cite{Halder2021}?
%     \item Teleportation \cite{brauer2024enhancing,oh2002noisyteleportation}/probabilistic teleportation \cite{fortes2016probabilistic}
%     \item Syndrome information can help you: \cite{Namiki2016}
%     \item Port-based teleportation?
% \end{itemize}

\section{non-postselected swapping}
\label{sec:advantages_twirled}
%In this section, we will see that twirled approximations have many advantages in the case of unconditional swapping.
\subsection{Preliminaries}
We now introduce the basic notation and twirling results that will be used in this work.
%We now introduce the twirled forms of two-qubit states, and demonstrate how these can simplify calculations in a repeater chain.
Let $X$, $Y$ and $Z$ denote the usual Pauli gates, given by
\begin{equation}
    X \equiv 
    \left( \begin{array}{cc}
        0 & 1 \\
        1 & 0
    \end{array}\right),\;
    Y \equiv 
    \left( \begin{array}{cc}
        0 & -i \\
        i & 0
    \end{array}\right),\;
    Z \equiv 
    \left( \begin{array}{cc}
        1  &  0 \\
        0  & -1
    \end{array}\right).
\end{equation}
In what follows, we will denote the Bell basis vectors as
\begin{equation}
    \ket{\Psi_{ij}} \coloneqq I_2 \otimes X^i Z^j \left[ \frac{1}{\sqrt{2}}\left(\ket{00} + \ket{11}\right)\right],
    % \ket{\Psi_{ij}} \coloneqq X_B^i Z_B^j \left[ \frac{1}{\sqrt{2}}\left(\ket{00} + \ket{11}\right)_{AB} \right],
    \label{eqn:bell_basis}
\end{equation}
where $i,j \in \{0,1\}$.

In the following lemmas, we restate the well-known results for the Bell-diagonal and Werner twirling of a two-qubit state. 
%Firstly, we see that the application of uniformly-chosen correlated Pauli gates removes the off-Bell-diagonal elements of the state.
\begin{lemma}[Bell-diagonal twirl]
Suppose that Alice and Bob share the two-qubit state $\rho$ and each apply the same gate chosen from $\{I_2,X,Y,Z\}$ uniformly at random, creating the channel
\begin{multline}
       \rho \mapsto \mathcal{B}(\rho) =  \frac{1}{4} \Big[ \rho + (X\otimes X)\rho (X^{\dag} \otimes X^{\dag}) \\ +  (Y\otimes Y)\rho (Y^{\dag} \otimes Y^{\dag})+ (Z\otimes Z)\rho (Z^{\dag} \otimes Z^{\dag}) \Big]. \label{eqn:pauli_twirling}
\end{multline}
Then, $\mathcal{B}(\rho)$ is diagonal in the Bell basis and $\bra{\Psi_{ij}} \mathcal{B}(\rho) \ket{\Psi_{ij}} = \bra{\Psi_{ij}} \rho \ket{\Psi_{ij}}$, for all $i,j\in \{0,1 \}$. In other words, the eigenvalues of $\mathcal{B}(\rho)$ are given by the diagonal elements of $\rho$ when written in the Bell basis.
\label{lem:pauli_twirling}
\end{lemma}
\begin{proof}
    See Appendix A of \cite{bennett1996mixed}.
\end{proof}
%We now look at a stronger form of twirling.
\begin{lemma}[Werner twirl]
    Suppose that Alice and Bob share the two-qubit state $\rho$. Alice applies the unitary $U$ to register $A$ and Bob applies $U^*$ to register $B$. The unitary $U$ is chosen uniformly at random from the Haar measure, creating the channel
    \begin{equation}
        \rho \mapsto \mathcal{W}(\rho) = \int \left( U\otimes U^* \right) \rho \left(U \otimes U^* \right)^{\dag} \dd U.
        \label{eqn:werner_twirling}
    \end{equation}
    Then, the resultant state is of the form
    \begin{equation}
        \mathcal{W}(\rho) = \frac{4F - 1}{3} \ket{\Psi_{00}}\bra{\Psi_{00}} + \frac{1-F}{3} I_4,
        \label{eqn:werner_state}
    \end{equation}
    where $F = \bra{\Psi_{00}}\rho \ket{\Psi_{00}}$ is the fidelity of $\rho$ to $\ket{\Psi_{00}}$, and $I_4$ is the identity matrix.
\label{lem:unitary_twirling}
\end{lemma}
\begin{proof}
    See Section V of \cite{horodecki1999reduction}.
\end{proof}
In what follows, we refer to states of the form (\ref{eqn:werner_state}) as \textit{Werner states} \cite{werner1989correlations}, which is standard terminology in the field of quantum networks (see e.g. \cite{dur1999repeaters,munro2015inside,azuma2021tools}). States of the form (\ref{eqn:werner_state}) are also sometimes referred to as isotropic states, which are the states that are invariant under the application of $U\otimes U^*$, for any unitary $U$. However, note that the term \textit{Werner state} is also sometimes used to refer to the states which have $U\otimes U$ symmetry, which were originally studied in \cite{werner1989correlations}. These are equivalent to the states (\ref{eqn:werner_state}) up to a Pauli $Y$ rotation on one of the two qubits. We note that in order to avoid sampling unitaries uniformly from the Haar measure, which can be computationally expensive and difficult to realise experimentally, it is possible to implement the map (\ref{eqn:werner_twirling}) instead by sampling from a finite set of correlated Pauli gates \cite{bennett1996mixed}.
\begin{definition}
    Given a two-qubit state $\rho$, we refer to $\mathcal{B}(\rho)$ as the \textit{Bell-diagonal approximation} of $\rho$. We refer to $\mathcal{W}(\rho)$ as the \textit{Werner approximation} of $\rho$. 
    \label{def:twirled_approximations}
\end{definition}
In this work, we refer to the Bell-diagonal (Werner) approximation \textit{in a repeater chain} as when the approximation is used for all initial states in the chain (see Figure \ref{fig:twirled_approx_in_swapping}).

% We note that the Bell-diagonal and Werner approximations may be implemented physically, using the maps given in Lemmas \ref{lem:pauli_twirling} and \ref{lem:unitary_twirling}. In this work, we will be interested in how the Bell-diagonal and Werner approximations of component states in a repeater chain affect the end-to-end state, and in particular its fidelity.

\subsection{Repeater chains with $N=2$ initial states}
\label{sec:N=2}
In this section, we consider non-postselected swapping on repeater chains with $N=2$ initial states. This section can be seen as a warm-up for our main results for chains with $N>2$ initial states, which are presented in Sections \ref{sec:N>2} and \ref{sec:advantages_full}.

\begin{figure}[t]
\includegraphics[width=80mm]{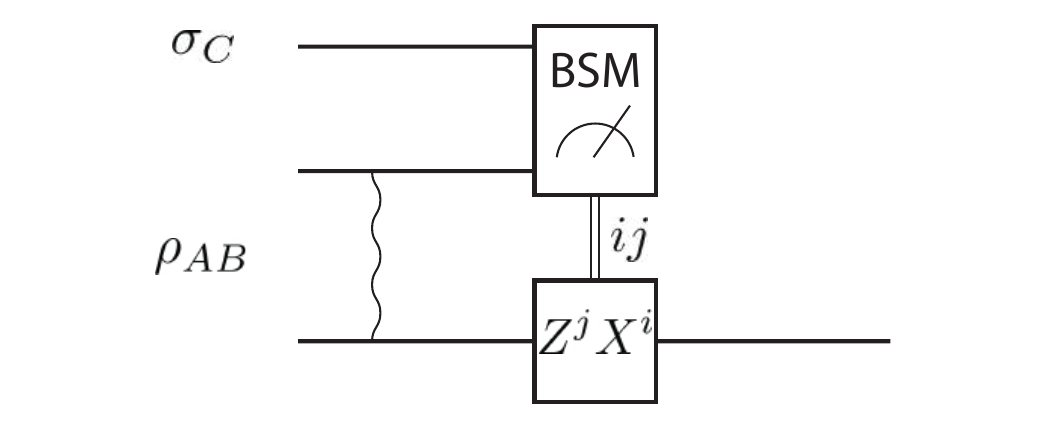}
\caption{\raggedright \label{fig:teleportation} \textbf{The standard teleportation protocol} \cite{bennett1993teleporting}. Input qubit state $\sigma_C$ and an entangled two-qubit resource state $\rho_{AB}$. A Bell-state measurement (BSM) is performed on registers $CA$, to obtain outcome $ij$. Finally, the correction $Z^j X^i$ is applied to register $B$.}
\end{figure}

Entanglement swapping was originally introduced in \cite{zukowski1993event} as a method to entangle particles that have not directly interacted.
The entanglement swapping protocol for $N=2$ initial states that we consider in this work is an extension of the protocol from \cite{zukowski1993event} that is implemented using the standard teleportation protocol from \cite{bennett1993teleporting}. We firstly outline the teleportation protocol in detail. We then go on to show a simple exactness result for the Bell-diagonal approximation in the standard teleportation channel in Proposition \ref{prop:teleportation_inv_twirling}, and extend it to the case of entanglement swapping in Corollary \ref{cor:swap_inv_twirling}. In Lemma \ref{lem:swap_bell_diagonal} and Corollary \ref{cor:swap_werner}, we see how the computation of the end-to-end state is simplified with the twirled approximations.

Given that two parties initially share entanglement, the standard teleportation protocol uses a BSM, classical communication and Pauli corrections to transport a quantum state between the two parties. The protocol is illustrated in Figures \ref{fig:teleportation} and \ref{fig:teleportation_and_swap}a. If all states and measurements involved are perfect, then the protocol teleports a quantum state perfectly and with probability one. If the initial shared entanglement is noisy, the teleportation protocol effectively sends the quantum state down a noisy channel. Properties of such channels, such as the teleportation fidelity, have been widely studied \cite{hu2023progress}. The shared entangled state used for teleportation is also referred to as a \textit{resource state}, because this is consumed in the protocol in order to teleport the target state.

% \begin{algorithm}
% \caption{The standard teleportation protocol \cite{bennett1993teleporting}}
% \label{alg:teleportation}
% \begin{algorithmic}[1]
% \State Input qubit state $\sigma_C$, two-qubit resource state $\rho_{AB}$.
% \State Perform measurement on registers $CA$ in Bell basis (\ref{eqn:bell_basis}), to obtain outcome $ij$.
% \State Apply correction $Z^j X^i$ on register $B$.
% \end{algorithmic}
% \end{algorithm}

% \begin{figure}[htbp]
% \centering
% \begin{minipage}{0.9\linewidth}
% \begin{algorithmic}[1]
%   \State Input qubit state $\sigma_C$, two-qubit resource state $\rho_{AB}$.
%   \noindent\rule{\linewidth}{0.4pt}
%   \State Perform measurement on registers $CA$ in Bell basis (\ref{eqn:bell_basis}), to obtain outcome $ij$.
%   \State Apply correction $Z^j X^i$ on register $B$.
% \end{algorithmic}
% \end{minipage}
% \caption{The standard teleportation protocol \cite{bennett1993teleporting}}
% \label{alg:teleportation}
% \end{figure}

% We now highlight the distinction between conditional and unconditional protocols. 
Postselected on BSM outcome $ij$, the output state after standard teleportation is
\begin{equation}
         \sigma'_{ij} = \frac{1}{p'_{ij}}(Z^j X^i)_B\Tr_{CA}\left[ \ketbra{\Psi_{ij}}_{CA} \sigma_C \otimes \rho_{AB} \right](Z^j X^i)_B^{\dag},
         \label{eqn:conditional_teleportation}
\end{equation}
where 
\begin{equation}
    p'_{ij} = \Tr \left[\ketbra{\Psi_{ij}}_{CA} \sigma_C \otimes \rho_{AB}\right]
\end{equation}
is the associated probability of obtaining measurement outcome $ij$. The weighted average of the postselected outcomes is given by 
\begin{align}
    \sigma' &= \sum_{i,j = 0}^1 \sigma'_{ij} p'_{ij} \! \\ &= \! \sum_{i,j= 0}^1 (Z^j X^i)_B\! \bra{\Psi_{ij}} \sigma_C \otimes \rho_{AB} \ket{\Psi_{ij}}_{CA}\! (Z^j X^i)_B^{\dag} \label{eqn:teleportation_weighted_average} \\ &\eqqcolon \Lambda_{\rho}^{\mathrm{tel}}(\sigma),
    \label{eqn:teleportation_channel}
    %=  \!\sum_{i,j= 0}^1 Z^j X^i \bra{\Psi_{ij}} \sigma_C \otimes \rho_{AB} \ket{\Psi_{ij}}_{CA} (Z^j X^i)^{\dag}.
\end{align}
where $\Lambda_{\rho}^{\mathrm{tel}}$ is the channel induced by standard teleportation with resource state $\rho$ \cite{bennett1993teleporting}.
% , since if the BSM outcome is not taken into account, one must then consider a weighted average of all possibilities.
% In the following result we will see that the Bell-diagonal approximation of the resource state is exact for unconditional teleportation. 
\begin{proposition}[Exactness of Bell-diagonal approximation for teleportation]
    Let $\Lambda^{\mathrm{tel}}_{\rho}(\sigma)$ denote the result of teleporting a qubit state $\sigma$ (register C) with a two-qubit resource state $\rho$ (registers AB) using the standard teleportation channel, as defined in (\ref{eqn:teleportation_channel}).
 % \begin{equation}
 %     \Lambda_{\rho}^{\mathrm{st}}(\sigma)\! = \!\sum_{i,j= 0}^1 Z^j X^i \bra{\Psi_{ij}} \sigma_C \otimes \rho_{AB} \ket{\Psi_{ij}}_{CA} (Z^j X^i)^{\dag}.
 %     \label{eqn:teleportation_channel}
 % \end{equation}
    Let $\mathcal{B}(\rho)$ denote the Bell-diagonal approximation of $\rho$. Then,
    \begin{equation}
        \Lambda_{\rho}^{\mathrm{tel}}(\sigma) = \Lambda_{\mathcal{B}(\rho)}^{\mathrm{tel}}(\sigma),
    \end{equation}
    i.e. the channel $\Lambda^{\mathrm{tel}}_{\rho}$ is invariant under the Bell-diagonal twirling of $\rho$. 
    \label{prop:teleportation_inv_twirling}
\end{proposition}
\begin{proof}
    Recalling that $\ket{\Psi_{ij}}_{CA} = (X^i Z^j)_A \ket{\Psi_{00}}_{CA}$, we may rewrite (\ref{eqn:teleportation_weighted_average}) in the following way: bringing the sum inside the inner product and using the identity  (\ref{eqn:pauli_twirling}) for Bell-diagonal twirling yields
    \begin{equation}
         \Lambda_{\rho}^{\mathrm{tel}}(\sigma)\! = 4 \bra{\Psi_{00}} \sigma_C \otimes \mathcal{B}(\rho_{AB}) \ket{\Psi_{00}}_{CA}.
         \label{eqn:twirl_inside_operators}
    \end{equation}
    Now,  
    \begin{align}
         \Lambda_{\mathcal{B}(\rho)}^{\mathrm{tel}}(\sigma)\! &= 4\bra{\Psi_{00}} \sigma_C \otimes \mathcal{B}^2(\rho_{AB}) \ket{\Psi_{00}}_{CA} \\
         &= 4 \bra{\Psi_{00}} \sigma_C \otimes \mathcal{B}(\rho_{AB}) \ket{\Psi_{00}}_{CA} = \Lambda_{\rho}^{\mathrm{tel}}(\sigma),
    \end{align}    
    where in the second line we have used the fact that the Bell-diagonal twirling of Bell-diagonal states leaves them invariant. 
\end{proof}
Interestingly, Proposition \ref{prop:teleportation_inv_twirling} allows us to derive a simple form for the standard teleportation channel.
\begin{corollary}[Standard teleportation is a Pauli channel]
Let $\Lambda^{\mathrm{tel}}_{\rho}(\sigma)$ denote the result of teleporting a qubit state $\sigma$ (register C) with a two-qubit resource state $\rho$ (registers AB) using the standard teleportation channel, as defined in (\ref{eqn:teleportation_channel}). Then, $\Lambda^{\mathrm{tel}}_{\rho}(\sigma)$ is a Pauli channel.
\label{cor:teleportation_pauli_channel}
\end{corollary}
\begin{proof} See Appendix \ref{app:unconditional_swapping_N=2}.
\end{proof}
%For $\sigma_C$ a qubit state and $\rho_{AB}$ a pure Bell state, this generates $\sigma$ in register $B$ with probability one. 
%If $\rho_{AB}$ is noisy, the protocol effectively sends $\sigma$ down a noisy channel, which has been widely studied \cite{horodeckis, more_noisy_teleportation}. We also refer to $\rho_{AB}$ as the \textit{resource state}, because this is consumed in the protocol in order to teleport $\sigma$ from $C$ to $B$. 

\begin{figure}[t]
\includegraphics[width=90mm]{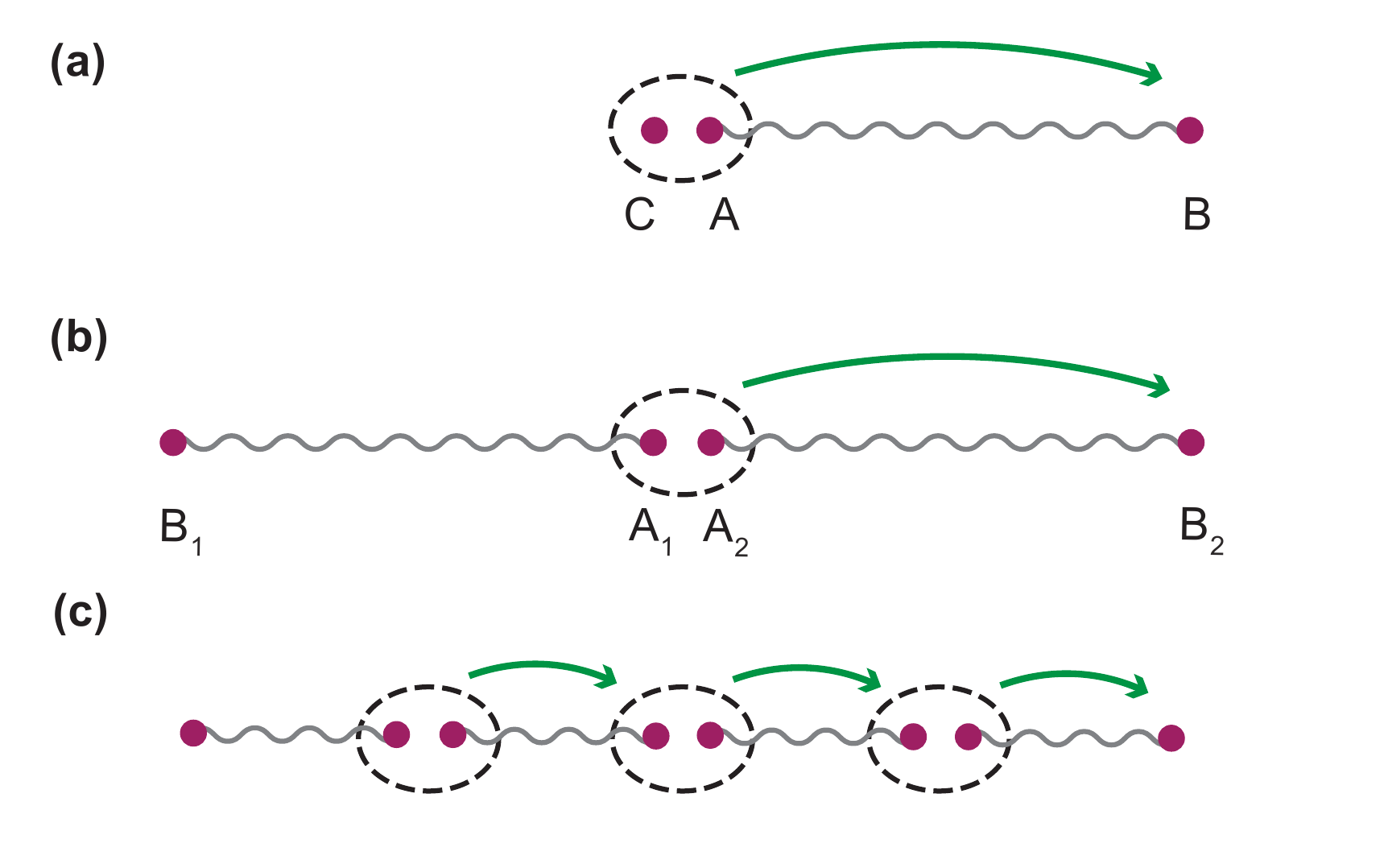}
\caption{\raggedright \label{fig:teleportation_and_swap} \textbf{Teleportation and entanglement swapping.} (a) The standard teleportation protocol involves a BSM on the target qubit and one qubit of an entangled resource state (wavy line), classical communication and corrections on the other half of the resource state (green arrow). (b) The standard entanglement swapping protocol involves teleporting one half of an entangled state. (c) An example of a swap-and-correct protocol for a repeater chain is applying teleportation sequentially.}
\end{figure}

We now define the entanglement swapping protocol for $N=2$ initial states \cite{zukowski1993event,bennett1993teleporting}. In the following, we let $\mathcal{D}(\mathcal{H})$ denote the set of density operators acting on Hilbert space $\mathcal{H}$.
\begin{definition}
    For $k=1,2,$ let $\mathcal{H}_{A_k}$ and $\mathcal{H}_{B_k}$ be qubit Hilbert spaces. Given a pair of two-qubit initial states $\rho_1\otimes\rho_2$ such that $\rho_k \in \mathcal{D}(\mathcal{H}_{A_k} \otimes \mathcal{H}_{B_k})$, the \textit{entanglement swapping} protocol (or just \textit{swapping}) is defined by applying the standard teleportation protocol (Figure \ref{fig:teleportation}) to teleport register $A_1$ to register $B_2$, using $\rho_2$ as the resource state. 
    
    \label{def:entanglement_swapping_protocol}
\end{definition}
% We note that when defining the entanglement swap, there is freedom in choosing which state is the resource state, to which the Pauli corrections are applied, and which state is to be teleported. In fact, one may also define a generalised swapping protocol for which the node at which Pauli corrections are applied can vary. We will explore this idea later on in the section.

See Figures \ref{fig:teleportation_and_swap}a and \ref{fig:teleportation_and_swap}b for an illustration of the entanglement swapping protocol for $N=2$.

Explicitly, the end-to-end state after a \textit{postselected swap} with BSM outcome $ij$ is then given by applying the map (\ref{eqn:conditional_teleportation}) to the initial states $\rho_1 \otimes \rho_2$.

The end-to-end state $\rho'$ after an \textit{non-postselected swap} is given by applying the standard teleportation channel (\ref{eqn:teleportation_channel}) to the initial states $\rho_1 \otimes \rho_2$:
\begin{equation}
   \rho' =  (I_2 \otimes  \Lambda_{\rho_2}^{\mathrm{tel}}) (\rho_1).
   \label{eqn:unconditional_e2e_state_N=2}
\end{equation}
A simple extension of Proposition \ref{prop:teleportation_inv_twirling} means that the Bell-diagonal approximation of the state $\rho_2$ is exact for non-postselected entanglement swapping (see Corollary \ref{cor:swap_inv_twirling}). 
However, this is not the case for postselected entanglement swapping. We will study postselected swapping in Section \ref{sec:advantages_full}.

% Motivated by the idea that the Bell-diagonal approximation is exact for unconditional swapping, 
We now compute the result of swapping Bell-diagonal states. For convenience, we denote a Bell-diagonal state as a length-four vector,
\begin{equation}
    \mathcal{B}(\rho) = \sum_{i,j} \lambda_{ij} \ketbra{\Psi_{ij}} \equiv (\lambda_{00},\lambda_{01},\lambda_{10},\lambda_{11})^T.
    \label{eqn:BD_as_4-vector}
\end{equation}
\begin{lemma}[Postselected swapping of Bell-diagonal states]
    Let $\rho'_{ij}$ be the end-to-end state after performing a postselected swap on a pair of Bell-diagonal states $\mathcal{B}(\rho_1) \otimes \mathcal{B}(\rho_2)$ such that $\mathcal{B}(\rho_1) \equiv (\lambda_0,\dots,\lambda_3)^T$ and $\mathcal{B}(\rho_2) \equiv (\mu_0,\dots,\mu_3)^T$, with BSM outcome $ij$. Then, $\rho'_{ij} = \rho'_{\mathcal{B}}$ for all $i$ and $j$, where
    \begin{equation}
        \rho'_{\mathcal{B}} \equiv (\lambda_0',\dots,\lambda_3')^T,
        \label{eqn:e2e_bell_diag}
    \end{equation}
is Bell-diagonal, and
\begin{equation}
\left( \begin{array}{c}
    \lambda_0'   \\
    \lambda_1'   \\
    \lambda_2'   \\
    \lambda_3'   \\
\end{array} \right)  = 
\left( \begin{array}{c}
    \lambda_0 \mu_0 + \lambda_1 \mu_1 + \lambda_2 \mu_2 + \lambda_3 \mu_3   \\
    \lambda_0 \mu_1 + \lambda_1 \mu_0 + \lambda_2 \mu_3 + \lambda_3 \mu_2   \\
    \lambda_0 \mu_2 + \lambda_2 \mu_0 + \lambda_3 \mu_1 + \lambda_1 \mu_3   \\
    \lambda_0 \mu_3 + \lambda_3 \mu_0 + \lambda_1 \mu_2 + \lambda_2 \mu_1   
\end{array} \right).
\label{eqn:swap_bell_diag_outcome}
\end{equation}
Moreover, the probability of this BSM outcome is
\begin{equation}
    p'_{ij} = \frac{1}{4}
\end{equation}
for all $i$ and $j$.
\label{lem:swap_bell_diagonal}
\end{lemma}
\begin{proof}
    See Appendix \ref{app:unconditional_swapping_N=2}.
\end{proof}
We see from Lemma \ref{lem:swap_bell_diagonal} that, when the initial states are Bell-diagonal, the end-to-end state after a non-postselected swap is given by
\begin{equation}
    \rho' = \sum_{i,j = 0}^1 \frac{1}{4}\rho'_{ij} = \rho'_{\mathcal{B}}.
    \label{eqn:unconditional_swap_BD}
\end{equation}
% This is because the result of the unconditional swap is a weighted average of the conditional outcomes -- see (\ref{eqn:teleportation_weighted_average}). 
Therefore, when the initial states of the chain are Bell-diagonal, we see that postselected and non-postselected swapping give the same end-to-end state. 
\begin{corollary}
     Let $\rho'$ be the end-to-end state after performing a non-postselected entanglement swap on a pair of two-qubit initial states $\rho_1\otimes \rho_2$, as defined in (\ref{eqn:unconditional_e2e_state_N=2}). Then, 
     \begin{align}
         \mathcal{B}(\rho') = \rho'_{\mathcal{B}}
     \end{align}
  where $\rho'_{\mathcal{B}}$ is the output state (\ref{eqn:e2e_bell_diag}) given by swapping the Bell-diagonal approximations $\mathcal{B}(\rho_1)$ and $\mathcal{B}(\rho_2)$.
    \label{cor:swap_inv_twirling}
\end{corollary}
\begin{proof}
    See Appendix \ref{app:unconditional_swapping_N=2}.
\end{proof}

%OLD VERSION
% \begin{corollary}
%      Consider the unconditional swapping of the two-qubit states $\rho_1$ and $\rho_2$, which have Bell-diagonal elements
%      \begin{align}
%          \mathcal{B}(\rho_1) &\equiv (\lambda_0,\dots,\lambda_3)^T \\ 
%          \mathcal{B}(\rho_2) &\equiv (\mu_0,\dots,\mu_3)^T .
%      \end{align}
%      Let $\rho'$ be the outcome state after performing an unconditional entanglement swap. Then, 
%      \begin{align}
%          \mathcal{B}(\rho') = \rho'_{\mathcal{B}}
%      \end{align}
%     where the $\mu_i$ are given by (\ref{eqn:swap_bell_diag_outcome}).
%     \label{cor:swap_inv_twirling}
% \end{corollary}

%From Proposition \ref{prop:swap_bell_diagonal}, we see that when the states in a repeater chain are Bell-diagonal and the standard swapping protocol is applied recursively, one can compute the end-to-end state by applying the map (\ref{eqn:swap_bell_diag_outcome}). See Figure \ref{fig:teleportation_and_swap}c for an illustration of the recursive swapping protocols.

From Corollary \ref{cor:swap_inv_twirling}, we see that the Bell-diagonal approximation is exact for the computation of the Bell-diagonal elements of the end-to-end state $\rho'$. This greatly simplifies the calculation of many important properties of the end-to-end state. For quantum network protocols and applications that use the end-to-end state as a resource, the Bell-diagonal approximation $\mathcal{B}(\rho')$ contains important information. For example, if one is performing QKD, the secret-key fraction of certain well-known protocols is invariant under Bell-diagonal twirling of the resource state \cite{bennett2014bb84,bruss1998sixstate,pironio2009device}. Furthermore, as we have seen in Proposition \ref{prop:teleportation_inv_twirling}, if standard teleportation is performed over the end-to-end link, then the only contributing components are again the Bell-diagonal elements of $\rho'$. 
%For example, this may be the case in blind quantum computation \cite{broadbent2009universalbqc,leichtle2021verifying}.
Even if the protocol performance is not only dependent on the Bell-diagonal elements, these elements may still contain important information about application feasibility. For example, $\mathcal{B}(\rho')$ contains the information of the fidelity to $\ket{\Psi_{00}}$, which is an important metric to deduce the feasibility and performance of entanglement purification protocols. It is a common assumption in purification protocols that the initial states are Bell-diagonal twirled \cite{deutsch1996dejmps,bennett1996mixed,jansen2022enumerating,krastanov2019optimized}.
%, as this can greatly simplify the analysis of such protocols. 

From Lemma \ref{lem:swap_bell_diagonal}, we obtain the well-known formula for the end-to-end state when the initial states are Werner \cite{munro2015inside}.
\begin{corollary}
    Let $\rho'_{\mathcal{W}}$ be the end-to-end state after swapping a pair of two-qubit Werner states $\mathcal{W}(\rho_1)\otimes \mathcal{W}(\rho_2)$ with fidelities $F_1$, $F_2$. Then, $\rho'_{\mathcal{W}}$ is a Werner state with fidelity $F'_{\mathcal{W}} =  F_1 F_2 + (1-F_1)(1-F_2)/3$.
    \label{cor:swap_werner}
\end{corollary}
\begin{proof}
    A Werner state with fidelity $F$ is Bell-diagonal, with the final three eigenvalues equal to one another:
    \begin{equation}
         \left(F, \frac{1-F}{3},\frac{1-F}{3},\frac{1-F}{3} \right)^T.
        \label{eqn:werner_state_bell_diag_coeffs}
    \end{equation}
    Then, by applying Lemma \ref{lem:swap_bell_diagonal}, swapping two Werner states with fidelities $F_1$ and $F_2$ results in a Werner state with fidelity $F'_{\mathcal{W}} = F_1 F_2 + (1-F_1)(1-F_2)/3$. As for Lemma \ref{lem:swap_bell_diagonal}, this result holds for both postselected and non-postselected swapping.
\end{proof}
Following from Corollary \ref{cor:swap_werner}, defining the \textit{Werner parameter}
\begin{equation}
    w \coloneqq \frac{4F - 1}{3},
    \label{eqn:werner_parameter}
\end{equation}
we see that the Werner parameter of the end-to-end state is $w' = w_1 w_2$. Then, to compute the Werner parameter for the output state, one only needs to multiply the Werner parameters of the initial states. This is a well-used result in the performance analysis of quantum networks \cite{munro2015inside}. If one is studying a large network, which could be a repeater chain or a more complex graph topology, it simplifies the analysis greatly to only consider the quality of each link to be described by one parameter $F$, which evolves under an entanglement swap according to the simple multiplicative relation -- see e.g. \cite{dur1999repeaters,vardoyan2023qnum,inesta2023optimal}. 
We note that similar multiplicative relations have also been found for general Bell-diagonal states \cite{goodenough2024swapasap}.

\begin{comment} 
\begin{lemma}
    Let $\Lambda^{\text{st}}_{\rho}(\sigma)$ denote the result of teleporting a qubit state $\sigma$ (register C) with a two-qubit resource state $\rho$ (registers AB) using the standard teleportation protocol (Algorithm \ref{alg:teleportation}), and averaged over all measurement outcomes, i.e.
    \Guus{Is it obvious that by average we mean the weighted average? I think it would be good to use some extra words to clarify that, and maybe also to show why the equation below is the correct equation for the average post-measurement state. By the way, isn't the weighted average effectively part of the definition of the standard teleportation protocol? Maybe then it would be better to derive/state the channel description when we first introduce the standard teleportation protocol?}
 \begin{equation}
     \Lambda_{\rho}^{\mathrm{st}}(\sigma)\! = \!\sum_{m,n= 0}^1 Z^n X^m \bra{\Psi_{mn}} \sigma_C \otimes \rho_{AB} \ket{\Psi_{mn}}_{CA} (Z^n X^m)^{\dag}.
     \label{eqn:teleportation_channel}
 \end{equation}
    Let $\mathcal{B}(\rho)$ denote $\rho$ under the application of random Pauli gates, as in (\ref{eqn:pauli_twirling}). Then,
    \begin{equation}
        \Lambda_{\rho}^{\mathrm{st}}(\sigma) = \Lambda_{\mathcal{B}\!(\rho)}^{\mathrm{st}}(\sigma),
    \end{equation}
    i.e. the channel $\Lambda^{\mathrm{st}}_{\rho}$ is invariant after Pauli twirling of $\rho$. In other words, it only depends on the Bell-diagonal elements of $\rho$.
    \label{lem:teleportation_inv_twirling}
\end{lemma}
\end{comment}
\subsection{Repeater chains with $N>2$ initial states}
\label{sec:N>2}
% We have seen in Corollary \ref{cor:swap_inv_twirling} that the Bell-diagonal approximation is exact for repeater chains with two initial states. One may apply the same reasoning in a repeater chain. 
Here, we consider non-postselected swapping over repeater chains with $N$ initial states and $N-1$ repeaters. Due to the freedom in the order in which to perform BSMs and Pauli corrections, we present a generalised class of swapping protocols on chains with $N$ initial states that we term \textit{swap-and-correct} protocols. For non-postselected swapping, we then go on to generalise the results of Section \ref{sec:N=2}, presenting an exactness result for the Bell-diagonal approximation (Theorem \ref{thm:swap_and_correct_equivalence}). We also present an accuracy result for the Werner approximation (Theorem \ref{thm:werner_approximation_accuracy}).

Suppose that non-postselected entanglement swapping is applied \textit{sequentially}. By sequentially, we mean that the standard swapping protocol (Definition \ref{def:entanglement_swapping_protocol}) is performed $N-1$ times, moving from one side of the chain to the other (Figure \ref{fig:teleportation_and_swap}c). Then, by Proposition \ref{prop:teleportation_inv_twirling}, the Bell-diagonal approximation is again exact for each of the $N-1$ states that were treated as the resource state. 
%Moreover, to compute the Bell-diagonal components of the end-to-end state, the Bell-diagonal approximation of the first state is also exact. 
% As a simple extension of Corollary \ref{cor:swap_inv_twirling}, in order to compute $\mathcal{B}(\rho')$, one may simply assume the Bell-diagonal approximation for all $N$ initial states and recursively apply the map (\ref{eqn:swap_bell_diag_outcome}) $N-1$ times. 

In practice, though, entanglement swapping is not likely to be implemented sequentially, because it requires classical communication and Pauli corrections after every BSM before the next BSM can be applied. With this strategy there is excessive classical communication time, and each swapped state has to spend an increasingly large amount of time waiting in memory before a BSM is applied to its qubit(s). This is problematic if qubits are subject to time-dependent noise while stored in memory, since added noise on the initial states can be detrimental to the quality of the final end-to-end state. For example, it may be beneficial to, instead of applying Pauli corrections sequentially after each BSM, apply them at the end nodes after all $N-1$ BSMs have been carried out. In this way, all $N-1$ BSMs at each node and classical communication of the outcomes may be carried out simultaneously, reducing the total amount of time the initial states must spend waiting in memory. We therefore look to generalise Corollary \ref{cor:swap_inv_twirling} to \textit{all} possible strategies of performing BSMs and Pauli corrections for a repeater chain of arbitrary length. 

We firstly present a definition of the class of entanglement swapping protocols under consideration. We term these \textit{swap-and-correct} protocols. The outcomes of the $N-1$ BSMs define the \textit{syndrome} $\vec s$. The syndrome is a length-$(N-1)$ list of Pauli operators such that $s_i \equiv X^m Z^n$ means that outcome $mn$ was measured on node $i$. The Pauli correction at each stage of a swap-and-correct protocol depends on the result of the syndrome up to that point. 
%NON-TECHNICAL DEFINITION
\begin{definition}[Swap-and-correct protocol, informal]
    For a length-$N$ repeater chain, a \textit{swap-and-correct} protocol $\mathcal{P}$ dictates where to apply Pauli corrections, given the $N-1$ BSM outcomes that form the syndrome $\vec s$. More specifically, $\mathcal{P}$ is a map 
    \begin{equation}
        \mathcal{P} : \{I,X,Z,XZ\}^{N-1} \rightarrow \{I,X,Z,XZ\}^{N+1}
    \end{equation}
    %\begin{equation}
    %    \mathcal{P} : \{\otimes_{k=1}^{N-1} \ket{\Psi_{i_k j_k}} \} \rightarrow \{I,X,Z,ZX\}^{\otimes (N+1)}
    %\end{equation}
     such that $\mathcal{P}_k(\vec s)$ is the correction applied to node $k$ for $k=0,\dots, N$. Moreover, for any syndrome $\vec s$, $\mathcal{P}$ transforms $\ketbra{\Psi_{00}}^{\otimes N}$ into $\ketbra{\Psi_{00}}$. We refer to this as the \textit{correctness property}.
     % $\mathcal{P}$ satisfies the \textit{correctness property}:
     % \begin{enumerate}
     %     %\item \textbf{$\mathcal{P}$ is physically implementable.} For any swap-and-correct protocol $\mathcal{P}$, there exists an associated permutation in which the $N-1$ BSMs are carried out. For $\mathcal{P}$ to be physical, at any given point in the protocol, corrections must only depend on outcomes of BSMs that have already been carried out, and not on future outcomes. 
     %     \item \textbf{$\mathcal{P}$ is correct.} For any syndrome $\vec s$, $\mathcal{P}$ transforms $\ketbra{\Psi_{00}}^{\otimes N}$ into $\ketbra{\Psi_{00}}$.
     % \end{enumerate}
     \label{def:swap_and_correct_protocol}
\end{definition}

We note that for clarity, some details have been omitted from the above. For example, there must be some associated ordering of the BSMs, so that corrections always depend on past outcomes. This is an important property that imposes more restrictions on $\mathcal{P}$. We refer to Appendix \ref{app:swap_and_correct} for the full technical definition.
We also note that, for the $N-1$ repeater nodes, the protocol does not specify which of the two qubits in the node the correction is applied to. This is because a BSM projection will be applied following any correction, which means that both choices are equivalent (see Appendix \ref{app:swap_and_correct}).

Some examples of swap-and-correct protocols are:
\begin{itemize}
    \item \textit{Sequential teleportation.} Here, $\mathcal{P}_0(\vec s) = \mathcal{P}_1(\vec s) = I_2$ and $\mathcal{P}_k(\vec s) = s_{k-1}$ for $k = 2,\dots,N$. 
    \item \textit{Correct at end.} Here, $\mathcal{P}_k(\vec s) = I_2$ for all $k = 0,\dots,N-1$, and  $\mathcal{P}_N(\vec s) = \prod_{k=1}^{N-1} s_k$. 
    %An associated permutation is for example the identity $\alpha(k) = k$, although the BSMs in this protocol may be carried out simultaneously as there are no intermediate corrections.
    % and $\mathcal{P}_N(s) = \prod_{k=1}^{N-1} s_k$. 
    %An associated permutation is again the identity. 
    %In this protocol, the BSMs cannot be carried out simultaneously as there are intermediate corrections.
\end{itemize}

%%%%%%%%%
%%%%%%%%%
%%%%%%%%%
%CONTINUE FROM HERE 
Given a swap-and-correct protocol $\mathcal{P}$, we denote the outcome state of a postselected swap with syndrome $\vec{s}$ as $\rho'_{\vec{s}}$, and the probability of measuring $\vec{s}$ as $p'_{\vec{s}}$. Then, the end-to-end state after non-postselected swapping is defined to be
\begin{align}
   \rho' &=  \sum_{\vec{s}} \rho'_{\vec{s}} p'_{\vec{s}} \\ &\eqqcolon \Lambda_{\mathcal{P}}(\rho_{\mathrm{in}}),
\end{align}
where $\rho_{\mathrm{in}} = \otimes_{k=1}^N \rho_k$ is the initial state comprised of $N$ two-qubit entangled states, and $\Lambda_{\mathcal{P}}$ is the channel induced by non-postselected swapping with $\mathcal{P}$. We note that both $\rho'_{\vec{s}}$ and $\rho'$ implicitly contain all syndrome-dependent Pauli corrections that are applied during the swap-and-correct protocol to rotate to the target maximally-entangled state $\ket{\Psi_{00}}$.

%We note that realistically, the initial state is comprised of $N$ two-qubit input states $\rho_{\mathrm{in}} =\otimes_{k=1}^N \rho_k$, however the theorem we will present below does not 

We now present a generalisation of Corollary \ref{cor:swap_inv_twirling} for swap-and-correct protocols.
\begin{theorem}[Exactness of Bell-diagonal approximation]
     Let $\mathcal{P}$ be a swap-and-correct protocol for repeater chains with $N$ initial states. Let $\rho_{\mathrm{in}} = \otimes_{k=1}^N \rho_k$ denote the $N$ initial two-qubit states. Let $\Lambda_{\mathcal{P}}$ be the channel induced by non-postselected swapping with $\mathcal{P}$. Let $\mathcal{B}_{[N]}$ denote the Bell-diagonal twirling of states $1,\dots N$, such that
     \begin{equation}
         \mathcal{B}_{[N]} \left(\rho_{\mathrm{in}}\right) = \otimes_{k=1}^N \mathcal{B}(\rho_k)
     \end{equation}
     is the Bell-diagonal approximation of the initial states.
     % $\rho_1, \dots, \rho_N$ be the two-qubit input states, such that  
     % \begin{equation}
     %     \rho_{\mathrm{in}} = \rho_1 \otimes \dots \otimes \rho_N,
     % \end{equation}
     % is the full initial state, and
    %  \begin{equation}
    %      \mathcal{B}_1 \dots \mathcal{B}_N \left(\rho_{\mathrm{in}}\right) = \otimes_{k=1}^N \mathcal{B}(\rho_k)
    %  \end{equation}
    %  denote the Bell-diagonal approximation of each initial state. Then, letting
    % $\Lambda_{\mathcal{P}}$ be the map induced by unconditional swapping with protocol $\mathcal{P}$, it follows that
    Then,
     \begin{equation}
         \mathcal{B}\left( \Lambda_{\mathcal{P}}(\rho_{\mathrm{in}}) \right) = \Lambda_{\mathcal{P}}( \mathcal{B}_{[N]} \left(\rho_{\mathrm{in}}\right)).
     \end{equation}
     Moreover, the above quantity is independent of the swap-and-correct protocol $\mathcal{P}$, i.e. 
     \begin{equation}
         \mathcal{B}\left( \Lambda_{\mathcal{P}}(\rho_{\mathrm{in}}) \right) = \Lambda_{\mathrm{seq}}\left( \mathcal{B}_{[N]} \left(\rho_{\mathrm{in}}\right) \right)
     \end{equation}
     where $\mathrm{seq}$ is the protocol where standard teleportation is applied sequentially on each repeater. 
     \label{thm:swap_and_correct_equivalence}
\end{theorem}
\begin{proof}
    See Appendix \ref{app:swap_and_correct}.
\end{proof}
We see from Theorem \ref{thm:swap_and_correct_equivalence} that for non-postselected swapping with any swap-and-correct protocol, the Bell-diagonal approximation is exact for the computation of the Bell-diagonal components of the end-to-end state. As for sequential swapping, one may then simply recursively apply the map (\ref{eqn:swap_bell_diag_outcome}) $N-1$ times to compute the end-to-end state.

We now turn to study the Werner approximation. With the following results, we quantify the error incurred by using the Werner approximation to compute the end-to-end fidelity in a repeater chain. 

%find upper and lower bounds for the output fidelity after swapping several Bell-diagonal states. By Theorem \ref{thm:swap_and_correct_equivalence}, these bounds then also hold for the output fidelity after applying a swap-and-correct protocol to general input states. 

%Suppose now that we have a repeater chain with $N+1$ nodes and $N$ two-qubit states $\rho_1,\dots,\rho_n$. We see from Corollary \ref{cor:swap_inv_twirling} that if one only cares about the \textit{average} outcome of the $n-1$ Bell-state measurements, then in order to compute the output state it suffices to only consider the Bell-diagonal form of $n-1$ input states. This simplifies the computation, and allows us to understand better the changes in state fidelity when performing a swap through the following results.

\begin{lemma}
    Consider applying the (non-postselected or postselected) sequential swapping protocol to a repeater chain with $N$ Bell-diagonal states $\otimes_{k=1}^N \mathcal{B}(\rho_k)$, where $F_k = \bra{\Psi_{00}}\rho_k \ket{\Psi_{00}}$ is the fidelity of the $k$-th initial state. Let $F' = \bra{\Psi_{00}}\rho' \ket{\Psi_{00}}$ be the fidelity of the end-to-end state $\rho'$. Then,
    \begin{equation}
        \prod_{k=1}^N F_k \leq F' \leq \frac{1}{2}\prod_{k=1}^N(2F_k-1)+ \frac{1}{2}.
    \label{eqn:bound_e2e_fidelity_BD}
    \end{equation}
    \label{lem:bound_e2e_fidelity_BD}
\end{lemma}
\begin{proof}
    See Appendix \ref{app:swap_and_correct}.
\end{proof}
One may combine Lemma \ref{lem:bound_e2e_fidelity_BD} and Theorem \ref{thm:swap_and_correct_equivalence} to obtain the following result for when swapping of $N$ general two-qubit states.
\begin{corollary}
     Let $\mathcal{P}$ be a swap-and-correct protocol for repeater chains with $N$ initial states. Let $\rho_{\mathrm{in}} = \otimes_{k=1}^N \rho_k$ denote the $N$ initial two-qubit states, where $\rho_k$ has fidelity $F_k$. Let $\Lambda_{\mathcal{P}}$ be the channel induced by non-postselected swapping with $\mathcal{P}$. Then, the end-to-end fidelity after non-postselected swapping with $\mathcal{P}$ satisfies
    \begin{equation}
        \prod_{k=1}^N F_k \leq \bra{\Psi_{00}} \Lambda_{\mathcal{P}}(\rho_{\mathrm{in}}) \ket{\Psi_{00}} \leq \frac{1}{2}\prod_{k=1}^N(2F_k-1)+ \frac{1}{2}.
        \label{eqn:bound_swapping_N_general_states}
    \end{equation}
    \label{cor:bounds_swap_and_correct}
\end{corollary}
We make the following remarks. 
%Firstly, the lower bound from Lemma \ref{lem:bound_e2e_fidelity_BD} is not tight as this requires the last three components of the input states to be orthogonal, and therefore the input states cannot be identical. 
The upper bound from (\ref{eqn:bound_swapping_N_general_states}) is tight: for example, this is saturated when the twirled initial states are of the form $\mathcal{B}(\rho_k) = (F_k,1-F_k,0,0)$. From (\ref{eqn:swap_bell_diag_outcome}), it can be seen that swapping two Bell-diagonal states of rank two (in the same subspace) outputs another Bell-diagonal state of rank two, in the identical subspace. Given that the rank-two ansatz is preserved, one may find a simple rule for the fidelity decay after a swap: the parameter $x \coloneqq 2F-1$ evolves multiplicatively under swapping as $x' = x_1 x_2$, which is analogous to the evolution of the Werner parameter as we saw in (\ref{eqn:werner_parameter}). We therefore see that the upper bound is tight. The lower bound is tight for $N=2$. For example, this is saturated when $\mathcal{B}(\rho_1) \equiv (F_1,1-F_1,0,0)$ and $\mathcal{B}(\rho_2) \equiv (F_2,0,1-F_2,0)$. We do not believe that the lower bound is tight for $N>2$, but we leave further investigation of this to future work.

% Recalling from Corollary \ref{cor:swap_werner} that Werner states are also Bell-diagonal and are characterised by the fidelity, Corollary \ref{cor:bounds_swap_and_correct} may be applied directly to bound the error when using the Werner approximation for all $N$ initial states. 
\begin{theorem}[Accuracy of Werner approximation]
     Let $\mathcal{P}$ be a swap-and-correct protocol for repeater chains with $N$ initial states. Let $\rho_{\mathrm{in}} = \otimes_{k=1}^N \rho_k$ denote the $N$ initial two-qubit states. Let $\Lambda_{\mathcal{P}}$ be the map induced by non-postselected swapping with $\mathcal{P}$. 
     % Let $\mathcal{W}_k$ denote the Bell-diagonal twirling of the $k$th state, such that
     % \begin{equation}
     %     \mathcal{W}_1 \dots \mathcal{W}_N \left(\rho_{\mathrm{in}}\right) = \otimes_{k=1}^N \mathcal{W}(\rho_k)
     % \end{equation}
     % is the Werner approximation of the initial states.
    % Let $\mathcal{P}$ be a swap-and-correct protocol for repeater chains of length $N$. Let $\rho_1, \dots, \rho_N$ be the two-qubit input states, such that  
    %  \begin{equation}
    %      \rho_{\mathrm{in}} = \rho_1 \otimes \dots \otimes \rho_N,
    %  \end{equation}
    %  is the full initial state, and
    %  \begin{equation}
    %      \mathcal{W}_1 \dots \mathcal{W}_N \left(\rho_{\mathrm{in}}\right) = \mathcal{W}_1(\rho_1) \otimes \dots \otimes \mathcal{W}_N(\rho_N)
    %  \end{equation}
    %  is the initial state after Werner-twirling all input states according to (\ref{eqn:werner_twirling}). 
    Let 
     \begin{equation}
         F' = \bra{\Psi_{00}} \Lambda_{\mathcal{P}}(\rho_{\mathrm{in}})\ket{\Psi_{00}}
     \end{equation}
     be the true end-to-end fidelity and $F'_{\mathcal{W}}$ be the end-to-end fidelity with the Werner approximation. Let $F_k$ be the fidelity of $\rho_k$. If $\epsilon_k = 1-F_k < \epsilon$ for all $k$, then
     \begin{equation}
         |F' - F'_{\mathcal{W}}| \leq {N \choose 2}\epsilon^2 + \mathcal{O}(N^3 \epsilon^3).
     \end{equation}
     In particular, if $N\epsilon \ll 1$, then
     \begin{equation}
         F' \approx F'_{\mathcal{W}}.
     \end{equation}
    %  where 
    %  Then, 
    % % $\Lambda_{\mathcal{P}}$ be the map corresponding to unconditional swapping with $\mathcal{P}$, it follows that
    %  \begin{multline}
    %  T\left[\mathcal{W}\left( \Lambda_{\mathcal{P}}(\rho_{\mathrm{in}}))\right), \Lambda_{\mathcal{P}}\left( \mathcal{W}_1 \dots \mathcal{W}_N \left(\rho_{\mathrm{in}}\right)\right)\right] \\ \leq {N \choose 2}\epsilon^2 + \mathcal{O}(N^3 \epsilon^3),,
    %  \end{multline}
    %  where $T$ is the trace norm.
     \label{thm:werner_approximation_accuracy}
\end{theorem}
\begin{proof}
    See Appendix \ref{app:swap_and_correct}.
\end{proof}
% From the above, we see that in order to compute the end-to-end fidelity of the final state, using the Werner approximation for the initial states is accurate if the infidelity of each link satisfies $1-F_k\ll 1/N$, for all $k$. 
We note that if we do \textit{not} have $N\epsilon \ll 1$, we do not expect the Werner approximation to be accurate: for example, swapping $N$ identical Werner states with fidelity $F$ results in a Werner state with fidelity 
\begin{equation}
    F_{\mathrm{out}} = \frac{3}{4} \left(\frac{4F-1}{3} \right)^N + \frac{1}{4},
    \label{eqn:werner_fidelity_N_swaps}
\end{equation}
which with $F$ fixed goes to $\frac{1}{4}$ as $N\rightarrow \infty$. By contrast, for identical initial states the lower and upper bounds in (\ref{eqn:bound_e2e_fidelity_BD}) go to $0$ and $\frac{1}{2}$ respectively.

%OLD BOUNDS ANALYSIS
\begin{comment}    
However, for $N \ll 1/(1-F)$, the Werner state approximation is accurate:
from Lemma \ref{lem:bound_e2e_fidelity_BD}, one may extract the following bounds for the cost of using a Werner state approximation for each state in the chain. Letting $\epsilon = 1-F$ be the infidelity of the initial states, one may rewrite (\ref{eqn:bound_swapping_N_general_states}) as
\begin{equation}
        (1-\epsilon)^N \leq \bra{\Psi_{00}} \Lambda_{\mathcal{P}}(\rho^{\otimes N}) \ket{\Psi_{00}} \leq \frac{1}{2} (1-2\epsilon)^N+ 
    \frac{1}{2}.
\end{equation}
To second order in $\epsilon$, this is
\begin{multline}
     1-N\epsilon + {N \choose 2} \epsilon^2 + o(N^2 \epsilon^2 ) \\ \leq \bra{\Psi_{00}} \Lambda_{\mathcal{P}}(\rho^{\otimes N}) \ket{\Psi_{00}} \\ \leq 1-  N \epsilon +{N \choose 2} 2 \epsilon^2+o(N^2 \epsilon^2 ),
\end{multline}
\Guus{I'm a little confused, isn't the third term in the top line above also of order $N^2 \epsilon^2$?}
and so using a Werner state approximation to $\rho$ causes an error of at most $${N \choose 2}  \epsilon^2 + o(N^2 \epsilon^2 ),$$
or alternatively the approximation starts becoming inaccurate when the second-order terms in $\epsilon N$ become significant. We therefore see that, when interested in the end-to-end fidelity averaged over the measurement outcomes, using a Werner state approximation is accurate for repeater chains of $N$ states if $N \ll 1/\epsilon$.
\end{comment}

\section{postselected swapping}
\label{sec:advantages_full}
% We have seen in the previous section that for unconditional entanglement swapping, using the Bell-diagonal approximation is exact and can greatly ease the computation of the end-to-end state, for repeater chains of length $N\geq 2$. Moreover, we have seen that when computing the end-to-end fidelity, the Werner approximation is accurate if the infidelity of each initial link is small compared to $1/N$. 
\subsection{Parameterisation of initial states}
In the previous section, we studied the accuracy of twirled approximations for non-postselected swapping. 
We now study the accuracy of approximations for postselected swapping, for repeater chains with $N=2$ initial states.
Specifically, we address the following question: \textit{how large (or small) can the fidelity of the outcome state become, postselected on a specific BSM outcome?} We will see that for certain initial states, after a postselected swap, the end-to-end state can exhibit a large variation in fidelity from what is obtained with a twirled approximation.

To illustrate the potential effect of postselecting on the BSM outcome, we consider an example that was introduced in \cite{Bose99}.
Consider swapping the initial states $\ketbra{\Psi_{\theta}}^{\otimes 2}$, where
\begin{equation}
    \ket{\Psi_{\theta}} = \cos(\theta)\ket{00} + \sin(\theta)\ket{11}
    \label{eqn:theta_state}
\end{equation}
with the standard swapping protocol (Definition \ref{def:entanglement_swapping_protocol}).
This state has fidelity to $\ket{\Psi_{00}}$ given by 
\begin{equation}
\big|\bra{\Psi_{00}}\ket{\Psi_{\theta}}\big|^2=\cos^2(\theta-\frac{\pi}{4}),
\end{equation}
which can take any value between $0$ and $1$, depending on the value of $\theta$. The possible outcomes for the end-to-end state after the swap are 
\begin{equation}
  \begin{cases}  
    \ket{\Psi_{00}},\;\;\;\;\;\; \text{with prob. } 2 \sin^2(\theta) \cos^2(\theta), \\
    \!\begin{aligned}[t]
     \frac{1}{C}\big(  \cos^2(\theta)&\ket{00}+   \sin^2(\theta)\ket{11}\big), \\  & \text{with prob. } C^2 = \cos^4 (\theta) + \sin^4(\theta).
    \end{aligned}
  \end{cases}
  \label{eqn:outcomes_example_pure_state}
\end{equation}

  %  \begin{cases}
  %      $\ket{\Psi^{\pm}}$, \text{ each with prob. } $\sin^2(\theta) \cos^2(\theta)$,
  %       \\
  %       $\frac{1}{N}\left( \cos^2(\theta)\ket{00}\pm  \sin^2(\theta)\ket{11}\right)$, \\ \text{ each with prob. } $N^2 = \frac{\cos^4 (\theta) + \sin^4(\theta)}{2}$.
  %  \end{cases}

In the above, the first outcome occurs when obtaining a measurement outcome corresponding to the odd-parity Bell states $ \ket{\Psi_{1j}}$, and the second outcome occurs when obtaining a measurement outcome corresponding to even-parity Bell states, $ \ket{\Psi_{0j}}$. We see from the above that for any $\theta \notin \left\{0,\pi/2,\pi,3\pi/2 \right\}$, there is a non-zero probability of obtaining an outcome state that is maximally entangled. This is in contrast to calculating the outcome of a non-postselected swap: the Bell-diagonal approximation of each initial state is
%$\rho = \ket{\Psi_{\theta}}\bra{\Psi_{\theta}}$ is 
\begin{equation}
    \mathcal{B}\left(\ketbra{\Psi_{\theta}}^{\otimes 2}\right) = F \ketbra{\Phi^+} + (1-F) \ketbra{\Phi^-}
    \label{eqn:Psi_theta_twirled}
\end{equation}
where $F=\cos^2(\theta-\frac{\pi}{4})$. By Corollary \ref{cor:swap_inv_twirling} and 
%the outcome after the unconditional swapping of two states $\ketbra{\Psi_{\theta}}$ has a post-swap fidelity that is the same as swapping two states of the form (\ref{eqn:Psi_theta_twirled}). 
Lemma \ref{lem:swap_bell_diagonal}, the end-to-end fidelity of a non-postselected swap is $F^2 + (1-F)^2$. One may also check that this is the weighted average of the outcomes in (\ref{eqn:outcomes_example_pure_state}). 
We therefore see that for certain states, after a postselected entanglement swap, there is a non-zero probability of obtaining a significantly higher (or lower) fidelity outcome than the non-postselected case. Recalling Lemma \ref{lem:swap_bell_diagonal},
%, where we saw that for Bell-diagonal states the outcome of a conditional swap is independent of the BSM outcome, 
this variation may be attributed to off-Bell-diagonal terms in the initial states (in this example, $\ketbra{\Psi_{\theta}}$).
%when the input states are written in the Bell basis. 

The variation in end-to-end fidelity is useful to characterise because some applications benefit from further information about the quality of the state. Postselecting on the Bell-state outcome can make certain tasks feasible: typically quantum applications are only feasible if on average, the level of noise in the resource state is below a certain threshold \cite{leichtle2021verifying,lipinska2020verifiable,scarani2009security}. Then, if a protocol is carried out by consuming many copies of the resource state, then by postselecting on certain swap outcomes a protocol can be made feasible, when for the non-postselected case it may not be. Furthermore, if the protocol is already feasible in the non-postselected case, knowledge of the full distribution of end-to-end fidelity can improve performance by making use of postselection. In Section \ref{sec:qkd} we provide an example of this for the case of quantum key distribution.

%for example, this has been demonstrated in the case of QKD \cite{goodenough2024swapasap,jing2020errordetection}.

In the remainder of the section, we will find bounds on the variation in the end-to-end fidelity after a postselected swap. 
%given that the initial states have fixed fidelity $F$ and that they are of a common noisy form, characterised by a second parameter $p$. Fixing the fidelity to $\ket{\Psi_{00}}$ allows for a direct comparison of the post-swap fidelity with the Werner approximation. Introducing the extra noise parameter $p$ allows us to obtain meaningful upper bounds for the post-swap fidelity, because as we have seen in the previous example, if one fixes only the fidelity of a state, it is always possible to find a state that obtains a post-swap fidelity of one. 
Specifically, for this problem we consider states of the form 
\begin{equation}
    \rho  = p \ketbra{\Psi_{00}} + (1-p)\sigma,
    \label{eqn:noisy_form}
\end{equation}
where $F = \bra{\Psi_{00}}\rho \ket{\Psi_{00}} $ is the fidelity,  $\sigma$ is a density matrix, and necessarily $p \leq F$. The last condition follows from the fact that $\rho$ and $\sigma$ are density matrices (positive semi-definite operators with unit trace).

The form (\ref{eqn:noisy_form}) is relevant for two reasons. Firstly, every state may be written in this form, which can be interpreted as an ensemble of the pure Bell state (probability $p$), and the state $\sigma$ (probability $1-p$). The state $\sigma$ can be interpreted as a noise component, which is not necessarily orthogonal to $\ket{\Psi_{00}}$. The parameters $p$ and $F$ may be computed efficiently given the state $\rho$ (see Section \ref{sec:analytical_bounds}), if not just directly deducible by inspection of the form of $\rho$. Then, understanding the limits of the end-to-end fidelity of states of the form (\ref{eqn:noisy_form}) has direct applications in a practical scenario. 
%This is in contrast to true measures of entanglement, which are notoriously convoluted to compute \cite{} 
Secondly, fixing the parameter $p$ as well as $F$ is more restrictive than only fixing $F$, which makes it possible to find meaningful bounds. In order to formalise this idea, we firstly define the set of states of interest for fixed $p$ and $F$.
\begin{definition}
    Let $F\in [0,1]$ and $p\leq F$. We denote
    \begin{align}
        S_{p,F} \coloneqq \Big\{ \rho : \exists \; \sigma \text{ s.t. }  &\rho = p \ketbra{\Psi_{00}} + (1-p)  \sigma, \nonumber\\ &\bra{\Psi_{00}} \rho \ket{\Psi_{00}} = F, \nonumber \\  &\sigma \text{ density matrix} \Big\}.
    \end{align}
to be the set of all states of the form (\ref{eqn:noisy_form}).
\end{definition}
\begin{proposition}
    Let $F\in [0,1]$ and $ p_2 < p_1 \leq F $. Then, $S_{p_1,F} \subset S_{p_2, F}$, but $S_{p_2, F}\nsubseteq S_{p_1,F}$.
    \label{prop:p_1>p_2}
\end{proposition}
\begin{proof}
    See Appendix \ref{app:conditional_swapping}.
\end{proof}
From Proposition \ref{prop:p_1>p_2}, we see that increasing $p$ (keeping $F$ fixed) provides a more restrictive form for the state (\ref{eqn:noisy_form}). For example, the set $S_{0,F}$ contains all valid two-qubit states with fidelity $F$, and the set $S_{F,F}$ contains only the states that have $\sigma$ orthogonal to $\ketbra{\Psi_{00}}$. The set $S_{p,1}$ has only one element, which is $\ketbra{\Psi_{00}}$.

In the following, we will study the limits of the end-to-end fidelity for states $\rho\in S_{p,F}$. 
We firstly introduce simplified notation for the end-to-end fidelity, which will be helpful in the following sections.
\begin{definition}
\label{def:post_swap_fid_and_prob}
    Consider performing a postselected swap on a pair of two-qubit states $\rho_{1}\otimes \rho_2$. If the BSM outcome $ij$ is obtained, we denote the end-to-end fidelity by $F_{ij}'(\rho_1 \otimes \rho_2)$, and the probability of this BSM outcome as $p_{ij}'(\rho_1 \otimes \rho_2)$. These quantities are written explicitly as
\begin{align}
   F'_{ij} &=  \frac{1}{p'_{ij}}\Tr \left[ \ketbra{\Psi_{ij}}_{B_1 B_2} \ketbra{\Psi_{ij}}_{A_1 A_2} \rho_1 \otimes \rho_2 \right]\label{eqn:post_swap_fid_def} \\ p'_{ij} &=\Tr \left[ \ketbra{\Psi_{ij}}_{A_1 A_2} \rho_1 \otimes \rho_2 \right],\label{eqn:swap_prob_def}
\end{align}
%\begin{align}
%   F'_{ij} &= \frac{\bra{\Psi_{ij}} \bra{\Psi_{ij}} \rho_1\otimes \rho_2 \ket{\Psi_{ij}}_{A_2 B_2} \ket{\Psi_{ij}}_{A_1 B_1} }{p'_{ij}}\label{eqn:post_swap_fid_def} \\ p'_{ij} &= \Tr_{A_1 B_1} \left[\bra{\Psi_{ij}} \rho_1 \otimes \rho_2  \ket{\Psi_{ij}}_{A_2 B_2} \right].\label{eqn:swap_prob_def}
%\end{align}
where we refer to Figure \ref{fig:teleportation_and_swap}b for a depiction of the qubit registers $B_1$, $A_1$, $A_2$ and $B_2$.
\end{definition}
We now wish to find bounds on the end-to-end fidelity $F'_{ij}$, for initial states $\rho_k \in S_{p,F}$. For clarity, we take the initial states to have the same parameters $p$ and $F$. We note that the results from this section also hold for the more general case ($\rho_k \in S_{p_k,F_k}$), for which the proofs are carried out in the Appendix.

We firstly show that it is enough to consider just a single BSM outcome.
\begin{proposition}
    Let $F\in [0,1]$ and $p\leq F$, and 
    \begin{align*}
        F'_{ij,\max} &\coloneqq \max_{\rho_k} \left\{ F'_{ij}(\rho_1 \otimes \rho_2) \;\;\mathrm{ s.t. }\; \rho_k \in S_{p,F}\right\} \\ F'_{ij,\min} &\coloneqq \min_{\rho_k} \left\{ F'_{ij}(\rho_1\otimes \rho_2) \;\; \mathrm{ s.t. } \; \rho_k \in S_{p,F}\right\}
    \end{align*}
    where $F'_{ij}$ is the postselected end-to-end fidelity (Definition \ref{def:post_swap_fid_and_prob}).
    Then, the above quantities are independent of $i$ and $j$,
    or alternatively
    \begin{align}
        F'_{ij,\max} &=  F'_{00,\max} \equiv  F'_{\max}(p,F) \label{eqn:def_F'max} \\
        F'_{ij,\min} &=  F'_{00,\min} \equiv F'_{\min}(p,F) \label{eqn:def_F'min}
    \end{align}
    \label{prop:mmt_outcome_independent}
    for all $i$, $j$.
\end{proposition}
\begin{proof} See Appendix \ref{app:conditional_swapping}. We use the idea that one may always rotate $\rho_2$ by a suitable Pauli to find a $\omega_2 \in S_{p,F}$ such that $F'_{ij}(\rho_1\otimes \rho_2) = F'_{00}(\rho_1\otimes \omega_2)$. 
\end{proof}
\subsection{Analytical bounds}
\label{sec:analytical_bounds}
Here, we use analytical methods to find an exact expression for $F'_{\max}(p,F)$ and a lower bound for $F'_{\min}(p,F)$. In Section \ref{sec:sdp}, we find tighter lower bounds on $F'_{\min}(p,F)$ using semidefinite programming (SDP).

To study $F'_{\max}(p,F)$ and $F'_{\min}(p,F)$ we firstly establish a simplified formula for the end-to-end fidelity $F'_{00}$. 
\begin{lemma}
Consider performing a postselected swap on a pair of two-qubit states $\rho_1\otimes \rho_2$ such that $\rho_k \in S_{p,F}$ and 
\begin{equation}
    \rho_k = p \ketbra{\Psi_{00}} + (1-p) \sigma_k,
\end{equation}
where $\sigma_k$ is a density matrix, for $k=1,2$. Let $F'_{ij}$ and $p'_{ij}$ denote the postselected end-to-end fidelity and probability (Definition \ref{def:post_swap_fid_and_prob}). Then, 
\begin{equation}
    F'_{ij}(\rho_1\! \otimes \! \rho_2) = \frac{ 2pF - p^2 + 4(1-p)^2 \tilde{p}_{ij}'\tilde{F}_{ij} '}{2p - p^2 + 4(1-p)^2 \tilde{p}'_{ij}},
    \label{eqn:formula_post_swap_fidelity_noisy_states}
\end{equation}
and
\begin{equation}
    p'_{ij}(\rho_1\! \otimes \! \rho_2) = \frac{p}{2} - \frac{p^2}{4} +  (1-p)^2  \tilde{p}'_{ij}
\label{eqn:formula_swap_prob_noisy_states}
\end{equation}
where $\tilde{F}'_{ij} \coloneqq F_{ij}'(\sigma_1\! \otimes \! \sigma_2)$ and $ \tilde{p}'_{ij} \coloneqq p_{ij}'(\sigma_1\! \otimes \! \sigma_2)$ are the swap statistics of the noisy components (Definition \ref{def:post_swap_fid_and_prob}).
\label{lem:formulae_swapping_noisy_states_id}
\end{lemma}
\begin{proof}
    See Appendix \ref{app:conditional_swapping}.
\end{proof}
From Lemma \ref{lem:formulae_swapping_noisy_states_id}, we see that the swap statistics of states in $S_{p,F}$ may be understood only in terms of the corresponding swap statistics of the noisy components $\sigma_k$. Lemma \ref{lem:formulae_swapping_noisy_states_id} is used in the proof of Theorem \ref{thm:tight_upper_bound}, where we find an exact expression for $F'_{\max}(p,F)$.
\begin{theorem}
    Consider performing a postselected swap on a pair of two-qubit states $\rho_1 \otimes \rho_2$. 
    Let $F_{\max}'(p,F)$ be the maximum achievable end-to-end fidelity for $\rho_k \in S_{p,F}$ with $k=1,2$, as defined in (\ref{eqn:def_F'max}). Then, 
    \begin{equation}
        F'_{\max}(p,F) =1 - 2 p (1-F).
        \label{eqn:tight_upper_bound_expression_p_F}
    \end{equation}
    In particular, the initial states $\rho_{\mathrm{opt}}^{\otimes 2}$ satisfy $F'_{00}(\rho_{\mathrm{opt}}^{\otimes 2}) = F'_{\max}(p,F)$, where
    \begin{align}
    \rho_{\mathrm{opt}} = p\ketbra{\Psi_{00}} + (1-p)\ketbra{\psi},
    \label{eqn:optimal_state_main_text}
\end{align}
and
\begin{equation}
    \ket{\psi} = \sqrt{\tilde{F}} \ket{\Psi_{00}} +  \sqrt{1-\tilde{F}} \ket{\Psi_{11}}
    \label{eqn:psi_best_state_main_text}
\end{equation}
with $\tilde{F} = (F-p)/(1-p)$
\begin{comment}
In particular, for  $\rho_k \in S_{p,F}$ such that 
\begin{align}
    \rho_1 = \rho_2 = p\ketbra{\Psi_{00}} + (1-p)\ketbra{\psi},
    \label{eqn:optimal_state}
\end{align}
with 
\begin{equation}
    \ket{\psi} = \sqrt{\tilde{F}} \ket{\Psi_{00}} + \sqrt{1-\tilde{F}} \ket{\Psi_{11}}
    \label{eqn:psi_best_state}
\end{equation}
where $\tilde{F} = (F-p)/(1-p)$, we have that 
\begin{equation}
    F'_{00}(\rho_1,\rho_2) = F_{\max}'(p,F).
\end{equation}
%i.e. this state provides the highest post-swap fidelity for all states in $S_{p,F}$.
\end{comment}
\label{thm:tight_upper_bound}
\end{theorem}
\begin{proof}
    See Appendix \ref{app:conditional_swapping}.
\end{proof}
The saturating state (\ref{eqn:optimal_state_main_text}) may be interpreted as $\ket{\Psi_{00}}$ having undergone a $Y$-rotation of a specified angle with probability $1-p$. %This coherent error is a $Y$-rotation on one of the qubits.
We note that in the more general case $\rho_k \in S_{p_k,F_k}$, the equality (\ref{eqn:tight_upper_bound_expression_p_F}) instead becomes an upper bound (see Appendix \ref{app:conditional_swapping}).

We notice that Theorem \ref{thm:tight_upper_bound} implies a similar result for the swapping of identical states. More specifically, for any $\rho \in S_{p,F}$ it follows that
\begin{equation}
    F'_{ij}(\rho^{\otimes 2}) \leq 1 - 2p(1-F)
    \label{eqn:upper_bound_identical_states}
\end{equation}
for any BSM outcome $ij$. 
%The statement in Theorem \ref{thm:tight_upper_bound} is made for non-identical states $\rho_1 \otimes \rho_2$ for the convenience of BSM outcome independence. %The bound (\ref{eqn:upper_bound_identical_states}) is only tight for certain choices of $(i,j)$.
%Moreover, for BSM outcome $\ket{\Psi_{00}}$ the inequality (\ref{eqn:upper_bound_identical_states}) is saturated by the optimal states from Theorem \ref{thm:tight_upper_bound} (see (\ref{eqn:optimal_state}) and (\ref{eqn:psi_best_state})). 

We see from Theorem \ref{thm:tight_upper_bound} that $F'_{\max}(p,F)$ is decreasing in $p$. The decreasing behaviour is expected, because from Proposition \ref{prop:p_1>p_2}, the set $S_{p,F}$ shrinks as $p$ increases and $F$ is fixed. 
%This may be used to tighten the upper bound (\ref{eqn:upper_bound_identical_states}), given identical fixed initial states. 
% To illustrate this, let us consider the swapping of two identical states $\rho \in S_{p,F}$. Given knowledge of the parameters $p$ and $F$, one may then use (\ref{eqn:upper_bound_identical_states}) in order to bound the post-swap fidelity. However, this bound can sometimes be tightened by noticing that $p$ is not unique, for a fixed $\rho$. 
One may tighten the bound (\ref{eqn:upper_bound_identical_states}) by finding the largest possible $q$ such that $\rho \in S_{q,F}$, which can be achieved by solving the optimisation problem
\begin{maxi}|s|
{q}{q\;\;\;\;\;\;\;\;\;\;\;\;\;\;\;\;\;\;\;}{}{p^* = }
\addConstraint{\rho - q \ketbra{\Psi_{00}} \geq 0}.
\label{opti:find_best_p}
\end{maxi}
The problem (\ref{opti:find_best_p}) may be solved efficiently using a simple SDP solver, which we provide in our repository \cite{github_repo}. The tightened bound is then given by
\begin{equation}
    F'_{ij}(\rho^{\otimes 2}) \leq 1 - 2p^*(1-F).
\end{equation}
A simple demonstration of this procedure is with that of the Werner state (\ref{eqn:werner_state}). We rewrite this as
\begin{equation}
        \mathcal{W}(\rho) = p \ket{\Psi_{00}}\bra{\Psi_{00}} + (1-p) \frac{I_4}{4},
        \label{eqn:werner_state_rewrite}
\end{equation}
where $p= (4F-1)/3$ and $F = \bra{\Psi_{00}}\rho \ket{\Psi_{00}}$. Since $I_4/4$ is a density matrix, it follows that $ \mathcal{W}(\rho) \in S_{p,F}$. However, we notice that $\mathcal{W}(\rho)$ may also be rewritten in Bell-diagonal form with the coefficients as given in (\ref{eqn:werner_state_bell_diag_coeffs}), and therefore $\rho_W\in S_{F,F}$. The second case gives the tighter upper bound for the end-to-end fidelity,
\begin{equation}
    F'_{ij}\big(\mathcal{W}(\rho)^{\otimes 2}\big) \leq F'_{\max}(F,F) =  1 - 2F(1-F),
\end{equation}
which can be easily validated with Corollary \ref{cor:swap_werner}.

We make the following further observations about Theorem \ref{thm:tight_upper_bound}. At $p = 0$, the expression simplifies to $F'_{\max}(0,F) = 1$. This is expected, because recalling the state $\ket{\Psi_{\theta}}$ from (\ref{eqn:theta_state}), we have $\ketbra{\Psi_{\theta}} \in S_{0,F}$ such that $F = \cos^2(\theta - \frac{\pi}{4})$. Then,
\begin{equation}
    1 = F'_{ij}\left(\ketbra{\Psi_{\theta}}^{\otimes 2}\right) \leq F'_{\max}(0,F),
\end{equation}
which implies the same.
%This is consistent with our observations at the beginning of the section, where we saw that $\ketbra{\Psi_{\theta}}$ has non-zero probability of swapping to unit fidelity.
If $\rho\in S_{F,F}$, the noisy component $\sigma$ is orthogonal to $\ketbra{\Psi_{00}}$. In such a case, we see that (\ref{eqn:tight_upper_bound_expression_p_F}) simplifies to 
\begin{align}
    F'_{ij}(\rho^{\otimes 2}) \leq F'_{\max}(F,F) &= 1- 2F(1-F)  \\ &= F^2 + (1-F)^2 \\ &= F'_{ij}(\rho_{\mathrm{R}2}^{\otimes 2}).
    \label{eqn:rank_two_saturates_bound_p=F}
\end{align}
where $\rho_{\mathrm{R}2}$ is any Bell-diagonal state of rank two with fidelity $F$. In the final step, we have recalled the formula for the postselected swapping of Bell-diagonal states from Lemma \ref{lem:swap_bell_diagonal}. 
%we have recalled from Lemma \ref{lem:bound_e2e_fidelity_BD} 
% which we see from Lemma \ref{lem:bound_e2e_fidelity_BD} is what is achieved by the best Bell-diagonal state. 
In particular, the state $\rho_{\mathrm{R}2} \in S_{F,F}$ provides optimal end-to-end fidelity for initial states in $S_{F,F}$.

%TIGHTNESS OF UPPER BOUND FOR p1 = p2 = 0, F1 ~= F2
\begin{comment}
We therefore see that when the noisy component is orthogonal to $\ket{\Psi_{00}}$, any rank-two Bell-diagonal state provides the maximum post-swap fidelity for fixed $F_1$, $F_2$, and that the bound from Theorem \ref{thm:tight_upper_bound} is once again tight:
\begin{equation}
    F'_{\max}(\vec{F},\vec{F}) = F_1 F_2 + (1-F_1)(1-F_2).
\end{equation}
\end{comment}
Now that we have characterised $F'_{\max}(p,F)$, we turn to studying $F'_{\min}(p,F)$.
%which from (\ref{eqn:def_F'min}) we recall is the minimum value of the post-swap fidelity after swapping $\rho_1$ and $\rho_2$, for $\rho_k \in S_{p,F}$. 
In the following Proposition, we derive an analytical lower bound for this quantity.
\begin{proposition}
 Let $F_{\min}'(p,F)$ be the minimum end-to-end fidelity for $\rho_k \in S_{p,F}$, as defined in (\ref{eqn:def_F'min}). Then, 
    \begin{equation}
        F'_{\min}(p,F) \geq \frac{p(2F - p)}{1+(1-p)^2}.
        \label{eqn:lower_bound_fidelity}
    \end{equation}
\label{prop:analytical_lower_bound}
\end{proposition}
%PROPOSITION STATEMENT FOR p_1~=p_2, F_1~=F_2
\begin{comment}
\begin{proposition}
    Consider $\rho_k \in S_{p_k,F_k}$. Let $F'_{00}$ denote the post-swap fidelity after performing a swap on $\rho_1 \otimes \rho_2$ with Bell-state measurement outcome $\ket{\Psi_{00}}$. Then, 
    \begin{equation}
        F'_{00} \geq \frac{p_1 F_2 + p_2 F_1 - p_1 p_2}{1+(1-p_1)(1-p_2)}.
        \label{eqn:lower_bound_fidelity}
    \end{equation}
\label{prop:analytical_lower_bound}
\end{proposition}
\end{comment}
\begin{proof}
    See Appendix \ref{app:conditional_swapping}.
\end{proof}
% Unlike the upper bound in Theorem \ref{thm:tight_upper_bound}, the lower bound (\ref{eqn:lower_bound_fidelity}) is not tight. In the following section, we find a tighter lower bound for $F'_{\min}(p,F)$ using semi-definite programming (SDP).

\subsection{Lower bound with SDP}
\label{sec:sdp}
Unlike the upper bound in Theorem \ref{thm:tight_upper_bound}, the lower bound (\ref{eqn:lower_bound_fidelity}) is not tight. In this section, we find a tighter lower bound for $F'_{\min}(p,F)$ using semi-definite programming (SDP).
% In Lemma \ref{lem:formulae_swapping_noisy_states_id}, we saw that the study of $F'_{\max}(p,F)$ and $F'_{\min}(p,F)$ for a fixed $p$ and $F$ reduces to the impact of the noisy component of the state, through the quantities $\tilde{F}'_{ij}$ and $\tilde{p}'_{ij}$. 
Recalling its definition in (\ref{eqn:def_F'min}), $F'_{\min}(p,F)$ is the solution to the optimisation problem
\begin{mini}|s|
{\rho_1 \! \otimes \! \rho_2}{F'_{ij}(\rho_1 \otimes \rho_2)\;\;\;\;\;\;\;\;\;\;\;\;\;\;\;\;\;\;\;\;\;\;\;\;\;\;\;\;\;}
{}{}
\addConstraint{\Tr \big[ \ket{\Psi_{00}}\bra{\Psi_{00}} \rho_k \big] = F}
\addConstraint{\Tr[\rho_k]= 1}
\addConstraint{\rho_k - p\ketbra{\Psi_{00}} \geq 0,\text{ for }k=1,2} 
.\label{opt:original0}
\end{mini}
Here, the constraints ensure that we are optimising over $\rho_k\in S_{p,F}$. The first constraint enforces $\rho_k$ has fixed fidelity $F$, and the final two constraints ensure that the noisy component $\sigma_k$ is a valid density matrix.
Recalling from Lemma \ref{lem:formulae_swapping_noisy_states_id} the formula for swapping two noisy states, (\ref{opt:original0}) may be written as 
\begin{mini}|s|
{\sigma_1 \! \otimes \! \sigma_2}{\frac{ 2pF - p^2 + 4(1-p)^2 \cdot \tilde{p}'_{00}\cdot \tilde{F}'_{00} }{2p - p^2 + 4(1-p^2)\cdot \tilde{p}'_{00}}}
{}{}
\addConstraint{p + (1-p)\mathrm{Tr}\big[\ket{\Psi_{00}}\bra{\Psi_{00}} \sigma_k \big] = F}
\addConstraint{\mathrm{Tr}[\sigma_k]= 1}
\addConstraint{\sigma_k \geq 0, \;\;\;\;\;\;\;\;\;\;\text{ for }k=1,2} 
.\label{opt:original}
\end{mini}
To obtain (\ref{opt:original}), we have reparameterised the problem to optimise over the noisy components $\sigma_k$. The quantities \linebreak $\tilde{F}'_{ij} \coloneqq F_{ij}'(\sigma_1 \otimes  \sigma_2)$ and $ \tilde{p}'_{ij} \coloneqq p_{ij}'(\sigma_1 \otimes  \sigma_2)$ are the corresponding swap statistics when only swapping the noisy components $\sigma_k$. 
%As was the case in (\ref{opt:original0}), the three constraints ensure that \ an optimisation over initial states from $S_{p,F}$.
\begin{comment}
\begin{mini}|s|
{\rho_{\mathrm{noise}}}{F'_{00}(\rho_{\mathrm{noise}}^{\otimes 2})}
{}{}
\addConstraint{p + (1-p)\mathrm{Tr}\big[\ket{\Psi_{00}}\bra{\Psi_{00}} \rho_{\mathrm{noise}}\big] = F}
\addConstraint{\mathrm{Tr}[\rho_{\mathrm{noise}}]= 1}
\addConstraint{\rho_{\mathrm{noise}}\geq 0}
.\label{opt:original}
\end{mini}
\end{comment}

The domain in (\ref{opt:original}) is the set of product states $\sigma_1 \otimes \sigma_2$, where $\sigma_k$ is a two-qubit density matrix. The domain is non-convex. Moreover, the objective function is rational, and not manifestly convex. These two details make (\ref{opt:original}) difficult to approach using numerical methods. We will therefore perform a relaxation of the domain, which transforms this the problem into one that is solvable with SDP. SDP is a commonly-used technique in quantum information \cite{rains2002semidefinite,tavakoli2024semidefinite}. The SDP formulation opens up the possibility of using several well-studied and efficient solvers, and moreover has an important feature that, under certain conditions, the solver converges to a global optimum. %The conditions hold in our case (see Appendix \ref{app:sdp}).

In order to study (\ref{opt:original}) with SDP, we perform two steps. Firstly, we linearise the objective function. Since the objective function of (\ref{opt:original}) is rational, we fix its denominator and introduce the new constraint 
\begin{align}
    \delta &= \frac{p}{2} - \frac{p^2}{4} + (1-p)^2 \tilde{p}_{00}' \nonumber \\
    &= \frac{p}{2} - \frac{p^2}{4} + (1-p)^2 \Tr\big[\ket{\Psi_{00}}\bra{\Psi_{00}}_{A_1 A_2} \sigma_1 \otimes \sigma_2 \big]. 
    \label{eqn:introduce_constraint_delta}
\end{align}
For conciseness, we rewrite the above as 
\begin{equation} \Tr\big[\ket{\Psi_{00}}\bra{\Psi_{00}}_{A_1 A_2} \sigma_1 \otimes \sigma_2 \big] = \tilde{\delta}(p,\delta)
    \label{eqn:deltatilde_constraint}
\end{equation}
where 
\begin{equation}
    \tilde{\delta}(p,\delta) \coloneqq \frac{4\delta - 2p + p^2}{4(1-p)^2}.
    \label{eqn:tilde_delta}
\end{equation}
Recalling Lemma \ref{lem:formulae_swapping_noisy_states_id}, this is fixing the total probability to be $\delta$. Since $p$ is fixed, $\delta$ and $\tilde{\delta}$ are interchangeable via the linear relation (\ref{eqn:tilde_delta}).
%here we are fixing the probability of obtaining BSM outcome $00$ to $\tilde{\delta}$ if one were to only swap the two noisy components, $\sigma_1$ and $\sigma_2$. Equivalently, 
%this is fixing the total probability to be $\delta$. Since $p$ is fixed, $\delta$ and $\tilde{\delta}$ are interchangeable via the linear relation (\ref{eqn:tilde_delta}).

Similarly, we rewrite both fidelity constraints as 
\begin{equation}
\mathrm{Tr}\big[\ket{\Psi_{00}}\bra{\Psi_{00}}_{A_1 A_2} \sigma_k \big] = \tilde{F},
\label{eqn:Ftilde_constraint}
\end{equation}
where
\begin{equation}
    \tilde{F}(p,F) \coloneqq \frac{F-p}{1-p}
\end{equation}
is the fidelity of the noisy component. Moreover, given that $\delta$ is fixed, we notice that the objective function is given by
\begin{align*}
    \frac{1}{4\delta}\left(  2pF - p^2 + 4(1-p)^2 \cdot \tilde{p}'_{00} \tilde{F}'_{00}\right),
\end{align*}
and so it suffices to optimise over
\begin{equation}
    \tilde{p}'_{00} \tilde{F}'_{00}\! = \!\mathrm{Tr}\big[ \ket{\Psi_{00}}\bra{\Psi_{00}}_{B_1 B_2}\! \ket{\Psi_{00}}\bra{\Psi_{00}}_{A_1 A_2} \sigma_1\!\otimes\! \sigma_2\big]
    \label{eqn:objective_fn_linearised}
\end{equation}
which is a linear function of $\sigma_1 \otimes \sigma_2$.
With our constraints and objective function reformulated as (\ref{eqn:deltatilde_constraint}), (\ref{eqn:Ftilde_constraint}) and (\ref{eqn:objective_fn_linearised}), we are now interested in the solution to
\begin{mini}|s|
{ \sigma_1\otimes \sigma_2}{ \mathrm{Tr}\big[ \ket{\Psi_{00}}\bra{\Psi_{00}}_{B_1 B_2} \ket{\Psi_{00}}\bra{\Psi_{00}}_{A_1 A_2} \sigma_1\otimes \sigma_2\big]}
{}{}
\addConstraint{ \mathrm{Tr}\big[\ket{\Psi_{00}}\bra{\Psi_{00}}_{A_1 A_2} \sigma_1\otimes \sigma_2 \big] = \tilde{\delta}(p,\delta)}
\addConstraint{\mathrm{Tr}\big[\ket{\Psi_{00}}\bra{\Psi_{00}}_{A_k B_k} \sigma_k \big] = \tilde{F}(p,F)}
\addConstraint{\mathrm{Tr}[\sigma_k]= 1}
\addConstraint{\sigma_k\geq 0, \;\;\text{ for }k=1,2}
.\label{opt:simplified}
\end{mini}
%Then, $F'_{\min}(p,F)$ may then subsequently be found by sweeping over $\delta$ (equivalently, sweeping over $\tilde{\delta}$). 
Letting $H^*(p,F,\delta)$ be the solution to (\ref{opt:simplified}), we have
% \begin{equation}
%     \frac{1}{\delta}\left( \frac{Fp}{2} - \frac{p^2}{4} +(1-p)^2 H^*(p,F,\delta)\right)
% \end{equation}
% is the minimum achievable value of the conditional post-swap fidelity, given that states are in $S_{p,F}$ and the swap outcome probability is $\delta$. It follows that
\begin{equation}
    F'_{\min}(F,p) = \min_{\delta} \frac{1}{\delta}\left( \frac{Fp}{2} - \frac{p^2}{4} +(1-p)^2 H^*(p,F,\delta)\right).
    \label{eqn:minimum_over_delta}
\end{equation}
As well as linearising the objective function, fixing $\delta$ allows one to study the rate-fidelity trade-off in the entanglement swapping protocol. This is useful because, in the performance analysis of quantum networks, it is important to understand both fidelity metrics and rate metrics in entanglement distribution protocols. For example, if a state provides a high fidelity with an excessively low probability of success, then this may no longer be very useful or relevant. Notice that the definitions of $F'_{\max}(p,F)$ and $F'_{\min}(p,F)$ in (\ref{eqn:def_F'max}) and (\ref{eqn:def_F'min}) are currently agnostic to the probability of obtaining the BSM outcome with minimum and maximum fidelity. Fixing the swap probability is a mechanism to study this: with such a constraint, for a given probability $\delta$ of a given swap outcome, one may study the limits of the fidelity. The same study was carried out in \cite{rozpkedek2018optimizing}, where the authors use SDP to study the maximum fidelity that can be achieved with practical purification protocols, given a fixed success probability of purification. In Appendix \ref{app:more_SDP_lower_bound}, we provide further discussion and analysis of the rate-fidelity trade-off. 
%We note that the swap probability may be fixed in the corresponding maximisation problem - see Appendix \ref{app:sdp} for more details.
Recalling that the domain over which we optimise in (\ref{opt:simplified}) is not convex (product states), we perform a relaxation of the domain. In particular, we use
\begin{equation}
    \sigma_1 \otimes \sigma_2 \in \mathrm{SEP} \subset \mathrm{PPT},
    \label{eqn:product_states_in_PPT}
\end{equation}
where $\mathrm{SEP}$ is the set of separable states, and $\mathrm{PPT}$ is the set of four-qubit states that are still positive after taking the partial transpose with respect to the registers $A_2$ and $B_2$ \cite{Peres1996}.
Relaxing the domain of (\ref{opt:simplified}) results in the following:
\begin{mini}|s|
{\sigma }{ \mathrm{Tr}\big[ \ketbra{\Psi_{00}}_{B_1 B_2} \ketbra{\Psi_{00}}_{A_1 A_2} \sigma \big]}
{}{}
\addConstraint{ \mathrm{Tr}\big[\ketbra{\Psi_{00}}_{A_1 A_2} \sigma \big] = \tilde{\delta}(p,\delta)}
\addConstraint{\mathrm{Tr}\big[\ketbra{\Psi_{00}}_{B_1 A_1} \sigma \big] = \tilde{F}(p,F)}
\addConstraint{\mathrm{Tr}\big[\ketbra{\Psi_{00}}_{B_2 A_2} \sigma \big] = \tilde{F}(p,F)}
\addConstraint{\mathrm{Tr}[\sigma]= 1}
\addConstraint{\sigma\geq 0, \;\;\; \sigma^{\Gamma}\geq 0.}
\label{opt:PPT1}
\end{mini}
where $M^{\Gamma}$ denotes taking the partial transpose of $M$ on the registers $A_2$ and $B_2$. The optimisation problem (\ref{opt:PPT1}) may now be solved with SDP. One may greatly reduce the number of parameters in the optimisation by using the fact that the objective function and all constraints are invariant under the application of correlated unitaries. See Appendix \ref{app:sdp} for the full details of the symmetrisation procedure. After symmetrisation, the number of free parameters in the optimisation is reduced from 256 to fewer than 48. 

\begin{figure*}[t!]
    \centering
    \begin{subfigure}[b]{0.5\textwidth}
        \centering
        \includegraphics[width=87mm]{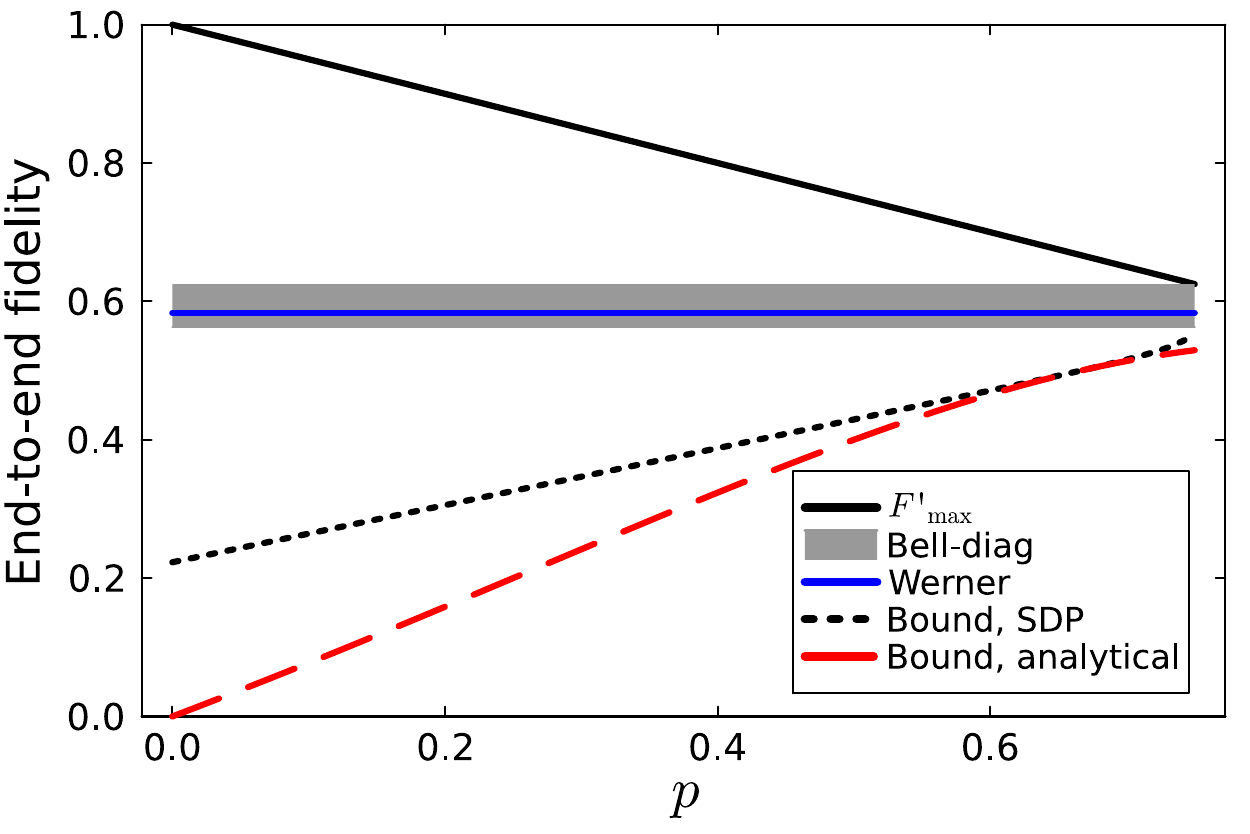}
        \caption{\label{fig:p_vs_F'_1}}
    \end{subfigure}%
    ~ 
    \begin{subfigure}[b]{0.5\textwidth}
        \centering
        \includegraphics[width=87mm]{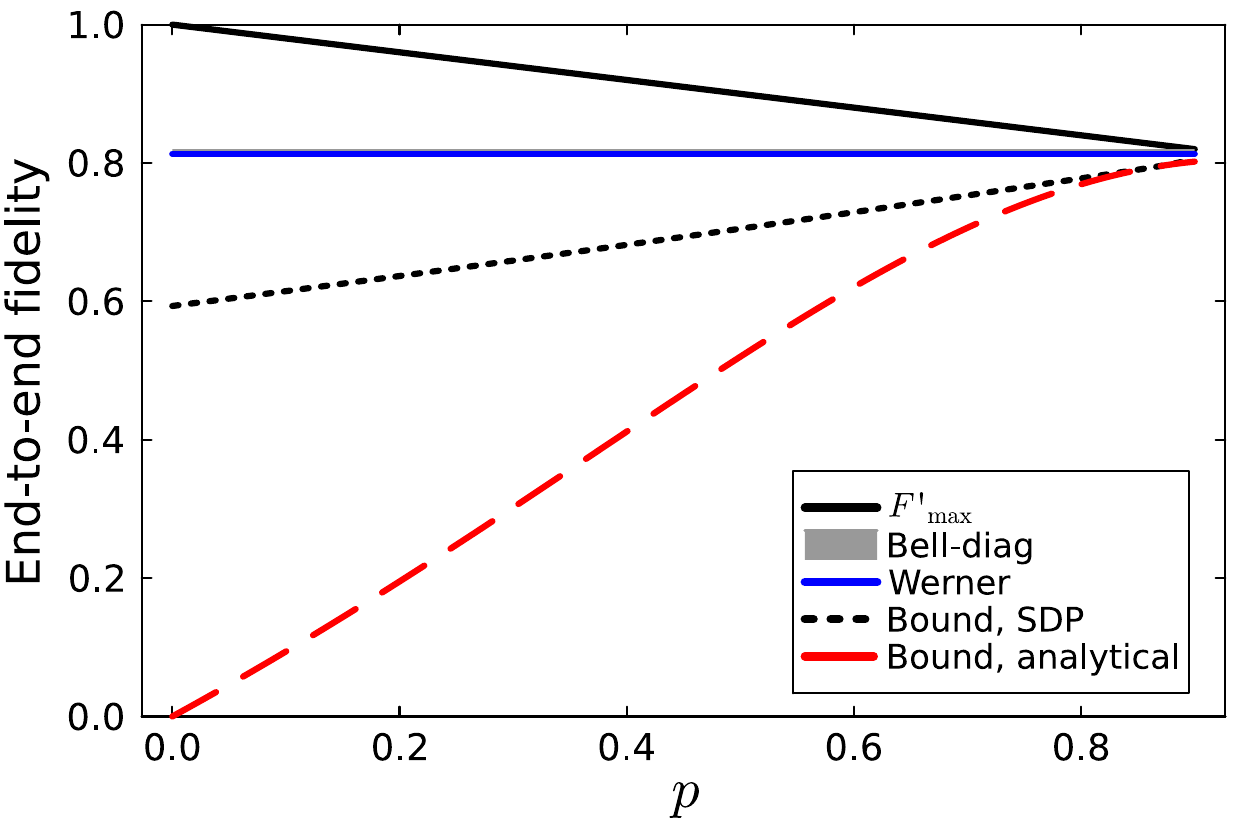}
        \caption{\label{fig:p_vs_F'_2}}
    \end{subfigure}
    \caption{\label{fig:p_vs_F'} \raggedright \textbf{Bounds for the end-to-end fidelity when swapping states $\rho_1\otimes \rho_2$ with $\rho_k = p\ketbra{\Psi_{00}}+ (1-p)\sigma_k$, for $p\in[0,F]$ ($\rho_k \in S_{p,F}$ for $k=1,2$), with (a) $F=0.75$ and (b) $F=0.9$.}  The black solid line is the tight upper bound on the postselected end-to-end fidelity, $F'_{\max}(p,F)$. The red dashed line is the analytical lower bound on the postselected -end-to-end fidelity, $F'_{\min}(p,F)$. The black dotted line is the SDP lower bound for $F'_{\min}(p,F)$, given in (\ref{eqn:lower_bound_with_SDP}). With the Bell-diagonal approximation, the end-to-end fidelity $F'_{\mathcal{B}}$ will lie in the grey region (for $F=0.9$ this is not visible). With the Werner approximation, the end-to-end fidelity $F'_{\mathcal{W}}$ will lie on the blue line. The plot is made for 100 values $p$ uniformly spaced within this interval.}
\end{figure*}

\begin{comment}Moreover, if there is further symmetry under permutation of certain qubit registers, this number may be reduced even further. For example, when interested in the case where the states being swapped are identical, there is further symmetry under exchanging registers $A_k \leftrightarrow B_k$, and the number of parameters may be reduced by another factor of two.
\end{comment}
Letting $H_{\mathrm{rel}}^*(p,F,\delta)$ be the solution to (\ref{opt:PPT1}), by the relaxation (\ref{eqn:product_states_in_PPT}) it follows that
\begin{equation}
    H^* (p,F,\delta) \geq H_{\mathrm{rel}}^*(p,F,\delta).
\end{equation}
Recalling (\ref{eqn:minimum_over_delta}), $F'_{\min}$ is then bounded below by 
\begin{equation}
    F'_{\min}(p,F) \geq \min_{\delta} \frac{1}{\delta}\left( \frac{Fp}{2} - \frac{p^2}{4} +(1-p)^2 H_{\mathrm{rel}}^*(p,F,\delta)\right),
    \label{eqn:lower_bound_with_SDP}
\end{equation}
where the optimisation is performed in the feasible region of $\delta$ (see Appendix \ref{app:sdp} for the calculation of the feasible region).
After symmetrisation of (\ref{opt:PPT1}), since the numerical optimisation over $\delta$ is over a single parameter in a bounded domain, (\ref{eqn:lower_bound_with_SDP}) is efficient to compute (on the order of a few seconds).
\begin{comment}
Note that we have not shown that this lower bound is tight. In the following section, we will see that in the parameter regimes explored, it lies above the analytical lower bound from Proposition \ref{prop:analytical_lower_bound}, but there does not exist a product state $\rho_1\otimes \rho_2$ with $\rho_k\in S_{p,F}$ that provides a post-swap fidelity saturating this bound.
\end{comment}

% \begin{figure}[t]
% \includegraphics[width=90mm]{figures/paper_bounds_F0.75p_values_100.pdf}
% \caption{\label{fig:p_vs_F'_1}Bounds for the post-swap fidelity when swapping states $\rho_1\otimes \rho_2$ with initial fidelity $F=0.75$ and $\rho_k = p\ketbra{\Psi_{00}}+ (1-p)\sigma_k$, for $p\in[0,F]$. The plot is made for 100 values $p$ uniformly spaced within this range.
% \Guus{It almost looks like the red and black striped lines are intersecting, it may be good to state somewhere whether they indeed intersect and/or touch (I think they don't though?)}
% }
% \end{figure}

% \begin{figure}[t]
% \includegraphics[width=90mm]{figures/paper_bounds_F0.9p_values_100.pdf}
% \caption{\label{fig:p_vs_F'_2}Bounds for the post-swap fidelity when swapping states $\rho_1\otimes \rho_2$ with initial fidelity $F=0.9$ and $\rho_k = p\ketbra{\Psi_{00}}+ (1-p)\sigma_k$, for $p\in[0,F]$. The plot is made for 100 values $p$ uniformly spaced within this range.}
% \end{figure}

\section{Discussion}
\subsection{Bounds comparison}
\label{sec:bounds_discussion}
Here, we illustrate the results from Sections \ref{sec:advantages_twirled} and \ref{sec:advantages_full} with examples. In particular, we will see how the parameters $p$ and $F$ of the initial states affect the accuracy of twirled approximations.

For fixed fidelity $F$, plotted in Figures \ref{fig:p_vs_F'_1} and \ref{fig:p_vs_F'_2} is $F'_{\max}(p,F)$ as found in Theorem \ref{thm:tight_upper_bound}, and the lower bounds for $F_{\min}'(p,F)$. In Figure \ref{fig:vary_F_p=0}, $p=0$ is fixed, and the same quantities are plotted. 
%The red dashed line is the lower bound from Proposition \ref{prop:analytical_lower_bound} that is found with analytical methods, and the black dotted line is the lower bound (\ref{eqn:lower_bound_with_SDP}) that is found with SDP.
In all cases we have tested, the SDP lower bound for $F_{\min}'(p,F)$ is tighter than the analytical lower bound.
%\todo{Discuss whether they are touching (check within tolerance of SDP solver)}
The grey region is where the end-to-end fidelity will lie if the Bell-diagonal approximation is used for the initial states. 
In particular, the grey region $[F^2, F^2 + (1-F)^2]$ is the region between the best and worst end-to-end fidelity for Bell-diagonal states of fidelity $F$, from Lemma \ref{lem:bound_e2e_fidelity_BD}. 
The grey region depends only on $F$, and hence is constant in Figures \ref{fig:p_vs_F'_1} and \ref{fig:p_vs_F'_2}. By Theorem \ref{thm:swap_and_correct_equivalence}, the grey region is also where the end-to-end fidelity will lie after a non-postselected swap.

In Theorem \ref{thm:werner_approximation_accuracy}, we saw that for $1-F\ll 1/N$, the Werner approximation is accurate for non-postselected swapping. 
%From Corollary \ref{cor:bounds_swap_and_correct}, the difference in the post-swap fidelity between the best and worst Bell-diagonal states when $N=2$ is of order $(1-F)^2$. 
% For this reason, in Figure \ref{fig:p_vs_F'_2} the Bell-diagonal region is barely distinguishable from the Werner line because the input fidelity $F=0.9$ is large. The same can be seen in Figure \ref{fig:vary_F_p=0}: for large $F$, the Bell-diagonal region is not distinguishable from the Werner line. 
Then, given that $N=2$ is fixed in Figure \ref{fig:vary_F_p=0}, for large $F$ the Bell-diagonal region is concentrated tightly around the Werner line.
% After Bell-diagonal or Werner twirling of the input states, the post-swap fidelity will always lie between these lines, and will not vary with the measurement outcome (Proposition \ref{prop:swap_bell_diagonal}). 

In Figure \ref{fig:vary_F_p=0}, because $p=0$ is fixed, the maximum end-to-end fidelity is constant at $F'_{\max}(0,F)=1$. 
This is expected from the discussion at the beginning of Section \ref{sec:advantages_full} where we saw that, when only fixing the fidelity of the input states, one may always find states that swap to unit fidelity. 
As well as the lower bounds for $F'_{\min}(0,F)$, we have plotted the lowest-fidelity outcome of the state $\ket{\psi}$ that was given in (\ref{eqn:psi_best_state_main_text}) as an example of a state giving output fidelity $F'_{\max}(p,F)$. The state $\ket{\psi}$ provides very good postselected swap statistics for the output fidelity of certain BSM outcomes. Since the end-to-end fidelity for a non-postselected swap must lie within the grey region and this is the weighted average of the postselected outcomes (Corllary \ref{cor:swap_inv_twirling}), the low-fidelity outcomes lie significantly below the grey region. In particular, the state $\ket{\psi}$ can also give an exceptionally low end-to-end fidelity. We plot this line in order to give an upper bound for the tightness of the SDP lower bound for $F'_{\min}(p,F)$. 

In Figures \ref{fig:p_vs_F'_1} and \ref{fig:p_vs_F'_2}, we see that for $p=F$, $F'_{\max}$ meets the upper limit of the grey region. The reason is what was seen in (\ref{eqn:rank_two_saturates_bound_p=F}): when the noisy component $\sigma_k$ is orthogonal to $\ket{\Psi_{00}}$, any rank-two Bell-diagonal state $\rho_{\mathrm{R}2}$ provides an optimal end-to-end fidelity, but also lies in the grey region due to being Bell-diagonal.
% i.e. for $\rho_{\mathrm{R}2} \in S_{F,F}$
% \begin{equation}
%     F'_{ij}(\rho_{\mathrm{R}2}^{\otimes 2}) = F'_{\max}(F,F) = F^2 + (1-F)^2.
% \end{equation}
% When $F$ is large, we have seen that Bell-diagonal states
% Consequently, for large $F$, when the noisy component $\sigma_k$ is orthogonal to $\ket{\Psi_{00}}$, a large increase in fidelity after an entanglement swap is not possible. 
We outline the practical relevance in the following way. Let consider swapping the initial states $\rho_1 \otimes \rho_2$ with $\rho_k \in S_{F,F}$. Let $F'_{\mathcal{B}}$ ($F'_{\mathcal{W}}$) denote the end-to-end fidelity with the Bell-diagonal (Werner) approximation, such that
\begin{align}
    F'_{\mathcal{B}} &= F'_{ij}(\mathcal{B}(\rho_1) \otimes \mathcal{B}(\rho_2)), \\ 
    F'_{\mathcal{W}} &= F'_{ij}(\mathcal{W}(\rho_1) \otimes \mathcal{W}(\rho_2)).
\end{align}
Let $(ij)^*$ denote the highest-fidelity BSM outcome after swapping $\rho_1 \otimes \rho_2$, such that 
\begin{equation}
    F'_{(ij)^*}(\rho_1 \otimes \rho_2) = \max_{i,j} F'_{ij}(\rho_1 \otimes \rho_2).
\end{equation}
Then, the corresponding output fidelity necessarily satisfies $F'_{(ij)^*}(\rho_1 \otimes \rho_2) \geq F'_{\mathcal{B}}$, since by Corollary \ref{cor:swap_inv_twirling},
\begin{equation}
    F'_{\mathcal{B}} = \sum_{ij} p'_{ij} F'_{ij}.
\end{equation}
Recalling from Lemma \ref{lem:bound_e2e_fidelity_BD} that $F'_{\mathcal{B}} \geq F^2$, the maximum deviation above the Bell-diagonal approximation is bounded as
\begin{align}
    F'_{(ij)^*}(\rho_1 \otimes \rho_2) - F'_{\mathcal{B}} &\leq F'_{\max}(F,F) - F^2\\ &= (1-F)^2. \label{eqn:BD_approx_deviation}
\end{align}
Recalling from Corollary \ref{cor:swap_werner} that $F'_{\mathcal{W}} = F^2 + (1-F)^2/3$, the maximum deviation above the Werner approximation is therefore
\begin{align}
    F'_{(ij)^*}(\rho_1 \otimes \rho_2) - F'_{\mathcal{W}} &\leq F'_{\max}(F,F) - F'_{\mathcal{W}}\\ &= \frac{2}{3}(1-F)^2. \label{eqn:W_approx_deviation}
\end{align}
Then, by (\ref{eqn:BD_approx_deviation}) and (\ref{eqn:W_approx_deviation}) we see that for large $F$, a large deviation above the twirled approximation is not possible when the input states $\rho_k \in S_{F,F}$ have an orthogonal component, i.e.
\begin{align}
    F'_{(ij)^*}(\rho_1 \otimes \rho_2) &\approx F'_{\mathcal{B}} \label{eqn:post_swap_fidelity_orth_states_1}\\ 
    F'_{(ij)^*}(\rho_1 \otimes \rho_2) &\approx F'_{\mathcal{W}}.\label{eqn:post_swap_fidelity_orth_states_2}
\end{align}
%For the deviations of post-swap fidelity \textit{below} the twirled approximation, we refe
% and so for large $F$, a large deviation in post-swap fidelity above the twirled approximation is not possible. 
% \begin{align}
%     F'_{ij}(\rho_1 \otimes \rho_2) - F'_{\mathcal{B}} &\leq (1-F)^2, \\ |F'_{ij}(\rho_1 \otimes \rho_2) - F'_{\mathcal{W}}| &\leq \frac{2}{3}(1-F)^2.
% \end{align}
% For large $F$, we therefore see that the Bell-diagonal (Werner) approximations for the initial states do not incur a large inaccuracy in the approximation of the end-to-end fidelity.
For example, consider swapping the initial states $\rho_R^{\otimes 2}$ with
\begin{equation}
    \rho_R = p\ketbra{\Psi_{00}}+(1-p)\ketbra{01},
    \label{eqn:R_state}
\end{equation}
which in some contexts is referred to as the R state. Up to a local unitary rotation, such a state closely approximates states generated in certain physical entanglement generation schemes \cite{campbell2008singleclick,cabrillo1999singleclick}. It has an orthogonal, non-Bell-diagonal noisy component $\ketbra{01}$. 
%Although there exist high-performing entanglement purification schemes tailored to the form of $\rho_R$ that take advantage of the form of the the noisy component \cite{nickerson2014epl,rozpkedek2018optimizing}, 
We see from our analysis that, for large $F$ ($p$), twirled approximations will not cause a large decrease in end-to-end fidelity because of the orthogonal noisy component.

As another example of a direct application of our bounds, let us consider the S state \cite{rozpkedek2018optimizing},
\begin{equation}
    \rho_S = p\ketbra{\Psi_{00}}+(1-p)\ketbra{11}.
    \label{eqn:S_state}
\end{equation}
The state $\rho_S$ has a non-orthogonal noisy component $\ketbra{11}$, with fidelity $|\bra{\Psi_{00}}\ket{11}|^2 = 1/2$.
By direct inspection of $\rho_S$, we see that $\rho_S \in S_{p,F}$, where $F = (1+p)/2$. By Theorem \ref{thm:tight_upper_bound},
\begin{align}
    F'_{ij}(\rho_S^{\otimes 2}) &\leq F'_{\max}\!\left(p,\frac{1+p}{2}\right) \\ &= 1-2p(1- (1+p)/2) \\ &= (1-p)^2 + p.
    \label{eqn:upper_bound_rho_S}
\end{align}
% Moreover, one may show that the highest-fidelity outcomes after swapping $\rho_S^{\otimes 2}$ are the even-parity BSM outcomes $0j$, which have a post-swap fidelity
% \begin{equation}
%     F'_{0j}(\rho_S^{\otimes 2}) = \frac{1}{2-p}
% \end{equation}
% (see Appendix \ref{app:S_states}).
% Then, for $1-p \ll 1$, we have 
% \begin{align}
%     F'_{0j}(\rho_S^{\otimes 2}) &= (1+ 1-p)^{-1} \\ &= 1 - (1-p) + \mathcal{O}((1-p)^2) \\ &= p + \mathcal{O}((1-p)^2).
%     \label{eqn:exact_fid_rho_S}
% \end{align}
% Combining (\ref{eqn:upper_bound_rho_S}) and (\ref{eqn:exact_fid_rho_S}),
% \begin{align}
%     F'_{\max}\!\left(p,\frac{1+p}{2}\right) - F'_{0j}(\rho_S^{\otimes 2}) = \mathcal{O}((1-p)^2) \ll 1,
% \end{align}
% from which we see that 
Moreover, we have
\begin{align}
    F'_{\mathcal{B}} \geq F^2 &= \left(\frac{1+p}{2}\right)^2 = p + \frac{1}{4}(1-p)^2, \label{eqn:BD_approx_rho_S}
\end{align}
and 
\begin{align}
    F'_{\mathcal{W}} &= F^2 + \frac{(1-F)^2}{3} \\  &= \left(\frac{1+p}{2}\right)^2 + \left(\frac{1-p}{2}\right)^2 \\  &= p + \frac{1}{2}(1-p)^2. \label{eqn:W_approx_rho_S}
\end{align}
Therefore, combining (\ref{eqn:upper_bound_rho_S}), (\ref{eqn:BD_approx_rho_S}) and (\ref{eqn:W_approx_rho_S}), we see that the maximum deviation \textit{above} the Bell-diagonal (Werner) approximations when swapping the initial states $\rho_S^{\otimes 2}$ is bounded above by
\begin{align}
    F'_{\max}\!\left(p,\frac{1+p}{2}\right) - F'_{\mathcal{B}} &= \frac{3(1-p)^2}{4}. \\ F'_{\max}\!\left(p,\frac{1+p}{2}\right) - F'_{\mathcal{W}} &= \frac{(1-p)^2}{2}.
\end{align}
Consequently, for large $p$ (equivalently, large $F$), we conclude that twirled approximations do not cause large inaccuracies in estimating the output fidelity when the initial states are $\rho_S^{\otimes 2}$.

We note that one may also perform a similar study for the deviation \textit{below} the twirled approximations by computing the difference with the analytical or SDP lower bounds for $F'_{\min}(p,F)$.
\begin{figure}[t!]
\includegraphics[width=90mm]{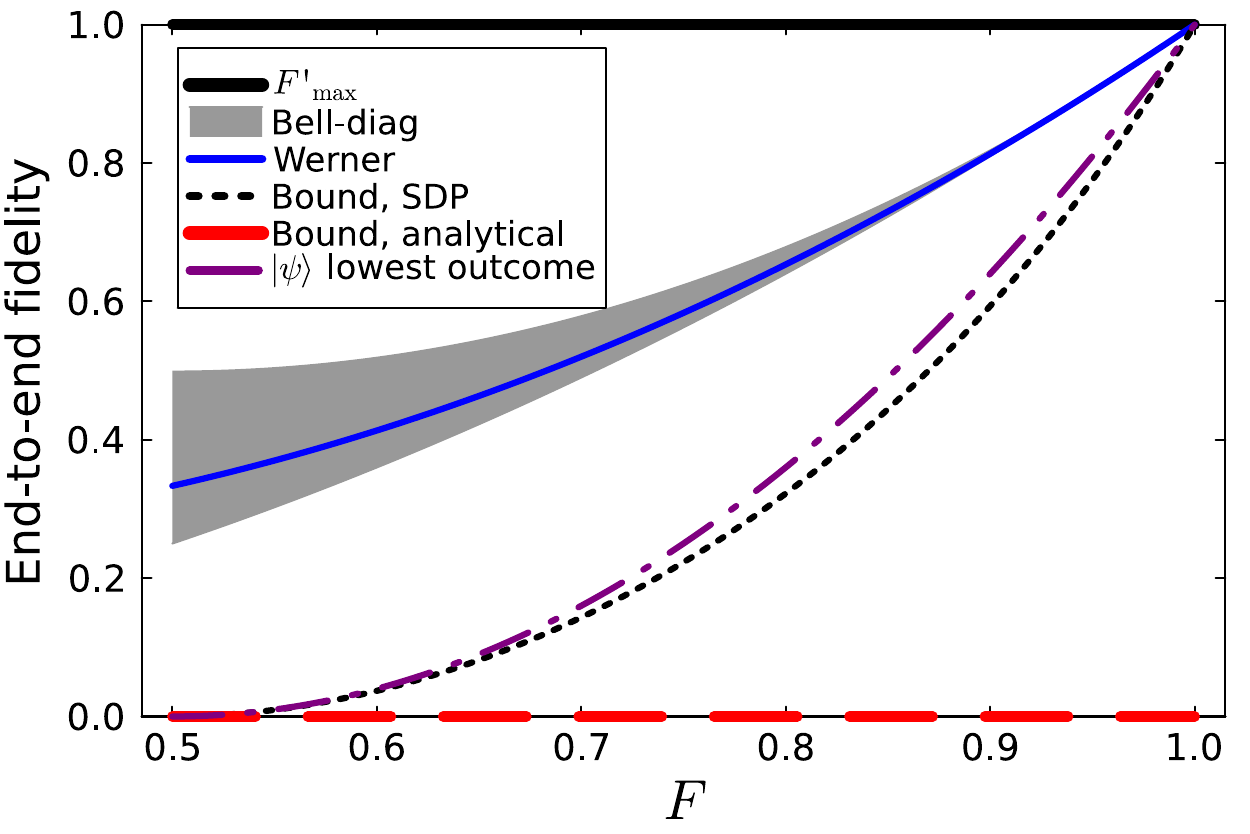}
\caption{\label{fig:vary_F_p=0} \raggedright \textbf{Bounds for the end-to-end fidelity when swapping states $\rho_1\otimes \rho_2$ with initial fidelity $F\in [0.5,1]$ ($\rho_k \in S_{0,F}$ for $k=1,2$).} The black solid line is the tight upper bound on the postselected end-to-end fidelity, $F'_{\max}(0,F)$. The red dashed line is the analytical lower bound for the postselected end-to-end fidelity, $F'_{\min}(0,F)$. The black dotted line is the SDP lower bound for $F'_{\min}(0,F)$. The purple dot-dash line is the lowest-fidelity outcome of $\ket{\psi}$, as defined in (\ref{eqn:psi_best_state_main_text}) with $p=0$. With the Bell-diagonal approximation, the end-to-end fidelity $F'_{\mathcal{B}}$ will lie in the grey region. With the Werner approximation, the end-to-end fidelity $F'_{\mathcal{W}}$ will lie on the blue line. The plot is made for 100 values of $F$ uniformly spaced within the interval.}
\end{figure}

\begin{figure*}[t!]
    \centering
    \begin{subfigure}[b]{0.5\textwidth}
        \centering
        \includegraphics[width=87mm]{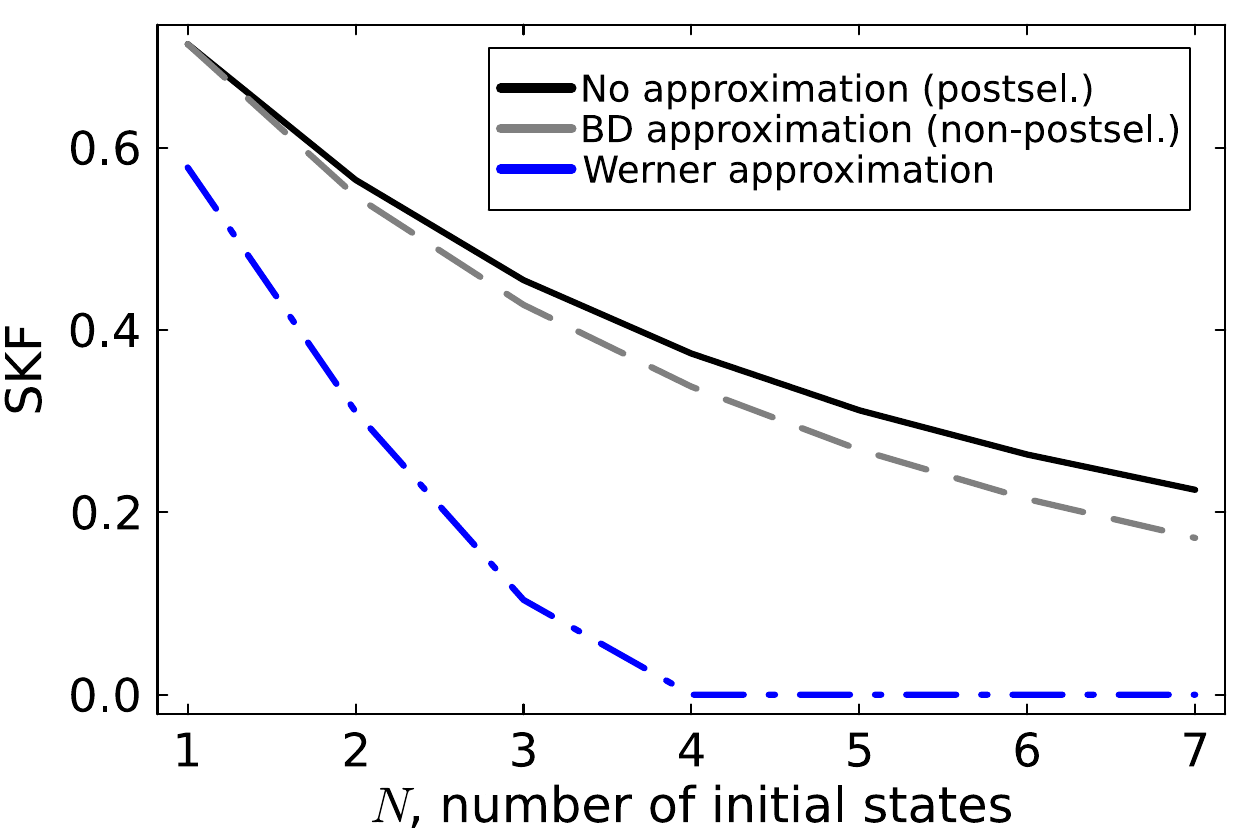}
        \caption{\label{fig:skf_opt_state}}
    \end{subfigure}%
    ~ 
    \begin{subfigure}[b]{0.5\textwidth}
        \centering
        \includegraphics[width=87mm]{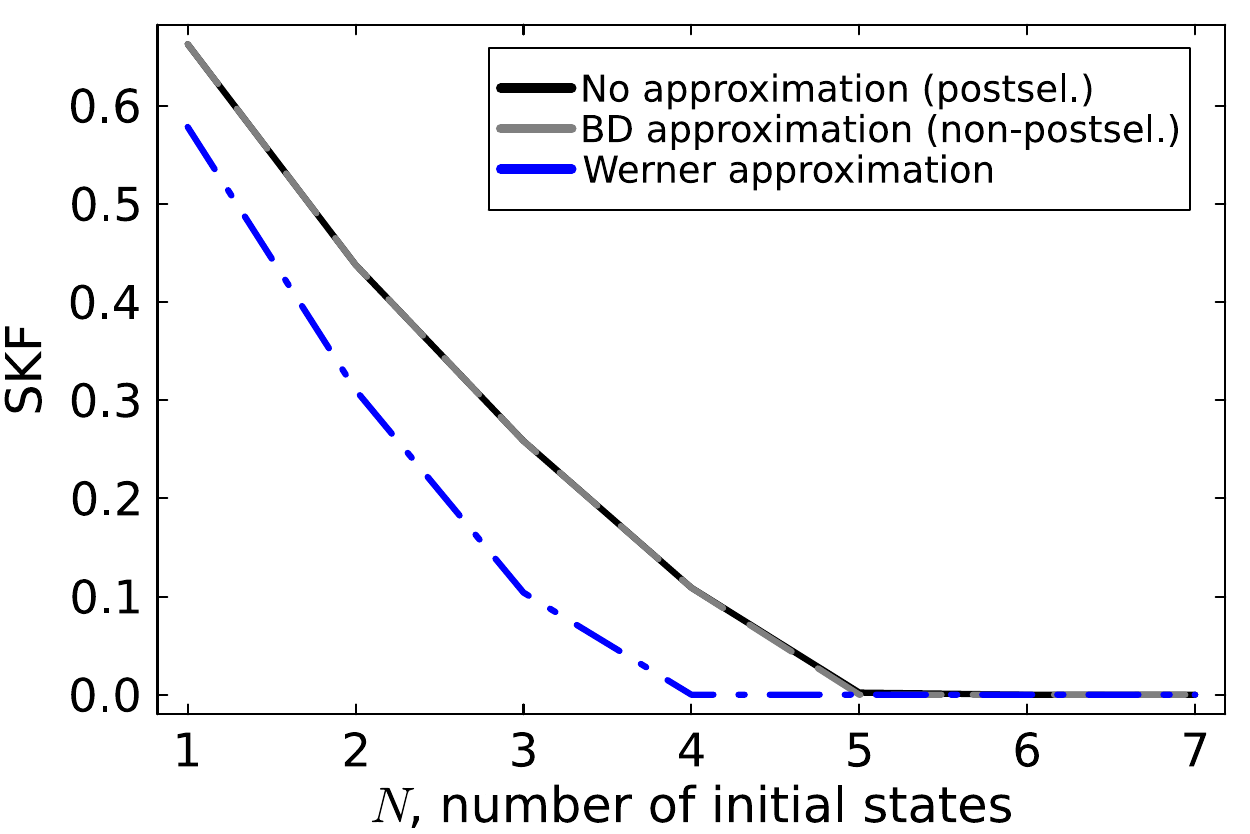}
        \caption{\label{fig:skf_r_state}}
    \end{subfigure}
    \caption{ \raggedright \textbf{Secret-key fraction when performing quantum key distribution over a repeater chain with $N$ initial states $\rho^{\otimes N}$} when (a) $\rho=\rho_{\mathrm{opt}}$ with $F=0.95$ and $p=0.5$ from (\ref{eqn:optimal_state_main_text}), and (b) $\rho=\rho_{R}$ from (\ref{eqn:R_state}) with fidelity $F=0.95$. The black line is the secret-key fraction of the postselected protocol, the grey dashed line is the secret-key fraction when the Bell-diagonal approximation is used for the initial states $\mathcal{B}(\rho)^{\otimes N}$ (or equivalently, the secret-key fraction that is obtained with the non-postselected protocol), and the blue dot-dashed line is the secret-key fraction when the Werner approximation is used for the initial states $\mathcal{W}(\rho)^{\otimes N}$.}
    \label{fig:skf}
\end{figure*}

\subsection{Example: quantum key distribution}
\label{sec:qkd}
% We saw in Section \ref{sec:advantages_twirled} that the loss in performance due to Bell-diagonal twirling has a direct operational interpretation: this is equivalent to unconditional swapping, when interested in the Bell-diagonal components of the end-to-end state. Moreover, we saw in Section \ref{sec:advantages_full} that when performing a conditional swap, certain states can exhibit a significant deviation in output fidelity from the unconditional outcome (twirled approximation).
% We now consider an explicit example of how one may use this fact to improve the performance of quantum key distribution (QKD).
%For a concrete example of how keeping track of swap outcomes can improve network performance for states that are not Bell diagonal, we here turn to quantum key distribution.
We now carry out a numerical study of the accuracy of twirled state approximations when quantum key distribution (QKD) is performed.
It has been shown previously that postselecting on the syndrome when using error correction in a repeater chain can give an advantage~\cite{Namiki2016, jing2020errordetection, wo2023}.
Here, we extend these results by pointing out that an advantage can also be obtained in the absence of error correction by postselecting on the swap outcomes, and moreover that this advantage is the exact loss in performance when using the Bell-diagonal approximation.

Let us consider performing QKD over the end-to-end state of a repeater chain that initially has $N$ identical two-qubit states, $\rho^{\otimes N}$. 
We assume that entanglement swapping is performed with the \textit{correct-at-end} protocol, where all BSMs may be performed simultaneously and a single Pauli correction is performed at one of the end nodes. 
%Entanglement swapping is performed at each repeater, and corrections are made only at the end.
Recalling Definition \ref{def:swap_and_correct_protocol}, this is a swap-and-correct protocol, and we denote it as $\mathcal{P}$. Let $\vec{s}$ be the swap syndrome, which holds the information of the $N-1$ swap outcomes.
We let $\rho'_{\vec{s}}$ be the final two-qubit state held by the end nodes, postselected on the syndrome being $\vec s$, and $p'_{\vec{s}}$ the probability of measuring $\vec s$.
% Let $\vec S$ be the syndrome, i.e., a vector holding all the swap outcomes, and 
The end nodes use the resulting end-to-end states to perform the BBM92 protocol for QKD~\cite{bennett1992a}, which is also known as entanglement-based BB84~\cite{bennett2014bb84}. 
The quantum bit error rate (QBER) of this protocol in the X (Z) basis is the probability that when both end nodes measure their state in the X (Z) basis, they obtain different outcomes.
In the protocol, the end nodes randomly perform such measurements and then use their outcomes to distil a secret key between them.
The number of secret bits that can be obtained per measurement of a state $\sigma$ in the asymptotic limit 
%(counting only instances when both end nodes measure in the same basis) 
is the secret-key fraction, which is given by~\cite{shor2000}
\begin{equation}
\text{SKF}(\sigma) = \max(0, 1 - h(Q_X(\sigma)) - h(Q_Z(\sigma))),
\label{eqn:skf}
\end{equation}
where $h(x) = -x\log_2(x) - (1-x) \log_2(1-x)$ is the binary entropy function and $Q_X(\sigma)$ and $Q_Z(\sigma)$ are the QBER of the state $\sigma$ in the X and Z basis respectively.
While the secret-key fraction is one in the perfect case when both QBERs are zero, it will become zero when the error rates are too large. \bd{I'm not sure I agree with this statement: if $Q_X = 1$ and $Q_Z = 0$, then the SKF is also one.} We note that, the two Pauli bases used throughout the protocol (in this case $X$ and $Z$) may be chosen from the Pauli bases. For example, if the $X$ and $Y$ bases are chosen instead, then the secret-key fraction (\ref{eqn:skf}) will instead depend on the QBER in the $Y$-basis, $Q_Y(\sigma)$. The secret-key fraction is invariant under the Bell-diagonal approximation, 
\begin{equation}
    \text{SKF}(\sigma) = \text{SKF}(\mathcal{B}(\sigma)).
    \label{eqn:SKF_inv_twirling}
\end{equation}
For a proof of this, we refer to Appendix \ref{app:SKF}.
Let $\lambda_{ij}$ be the Bell-diagonal elements of $\sigma$, such that 
\begin{equation}
    \mathcal{B}(\sigma) = \sum_{i,j=0}^1\lambda_{ij} \ketbra{\Psi_{ij}}.
\end{equation}
Then, the QBER in each measurement basis is given by 
\begin{align}
    Q_X(\sigma) &= \lambda_{01} + \lambda_{11} \label{eqn:qberX}\\ 
    Q_Y(\sigma) &= \lambda_{10} + \lambda_{01} \label{eqn:qberY} \\ 
    Q_Z(\sigma) &= \lambda_{10} + \lambda_{11}. \label{eqn:qberZ}
\end{align}
\\
Given our setup, we compare two different ways in which the end nodes can distil a secret key.
The first option is to process measurement outcomes without keeping track of the syndrome.
We call this the \textit{non-postselected protocol}, and has secret-key fraction $\text{SKF}(\rho')$, where
% and write $\text{SKF}_{\text{uncon}}$ for the corresponding secret-key fraction.
% It is given by
% \begin{equation}
% \text{SKF}_\text{uncon} = \text{SKF}(\rho').
% \end{equation}
% Here,  
% $\rho'$ is the unconditional quantum state given by
\begin{equation}
   \rho' = \sum_{\vec s} p'_{\vec s} \rho'_{\vec{s}} =  \Lambda_{\mathcal{P}}(\rho^{\otimes N}),
\end{equation}
% given by
% \begin{equation}
% \bar{\rho}' = \Lambda_{\mathcal{P}}(\rho^{\otimes N}),
% \end{equation}
%  Here, we are using the notation for the channel corresponding to unconditional swapping with protocol $\mathcal{P}$ as used in Theorem \ref{thm:swap_and_correct_equivalence}. Alternatively, this state may be written as 
% \begin{equation}
% \bar{\rho}' = \sum_{\vec s} p(\vec s) \rho'_{\vec{s}}  ,
% \end{equation}
% where $p(\vec s)$ is the probability of obtaining syndrome $\vec s$. 
and $\Lambda_{\mathcal{P}}(\rho^{\otimes N})$ is the channel induced by non-postselected swapping with $\mathcal{P}$. By (\ref{eqn:SKF_inv_twirling}) and Theorem \ref{thm:swap_and_correct_equivalence}, we have 
\begin{equation}
    \text{SKF}(\rho') = \text{SKF}(\mathcal{B}(\rho')) = \text{SKF}(\rho'_{\mathcal{B}}),
\end{equation}
where $\rho'_{\mathcal{B}}$ is the output state with the Bell-diagonal approximation.
In particular, the secret-key fraction with the non-postselected protocol is exactly what is obtained with the Bell-diagonal approximation.

The second option for the distillation of secret key is to divide all measurement outcomes into different blocks based on their corresponding syndromes and process each block separately.
We call this the \textit{postselected protocol} and its secret-key fraction can be calculated as
\begin{equation}
\sum_{\vec s} p'_{\vec s} \: \text{SKF}(\rho'_{\vec{s}}).
\end{equation}
Because the SKF function is convex within the domain where it is nonzero, we have
\begin{equation}
    \sum_{\vec s} p'_{\vec s} \: \text{SKF}(\rho'_{\vec{s}}) \geq \text{SKF}(\rho'). 
\end{equation}
% It has been shown previously that conditioning on the syndrome when using error correction in a repeater chain can give an advantage~\cite{Namiki2016, jing2020errordetection, wo2023}.
% Here, we extend these results by pointing out that an advantage can also be obtained in the absence of error correction by conditioning on the swap outcomes, and moreover that this advantage is the exact loss in performance when using the Bell-diagonal approximation. 

% By the equivalence of unconditional swapping protocols and the Bell-diagonal approximation as given in Theorem \ref{thm:swap_and_correct_equivalence}, the difference between the conditional SKF and the unconditional SKF is the inacurracy incurred by using the Bell-diagonal approximation for each resource state in the chain. We also compare with the Werner approximation.

We consider two types of initial states: firstly, we consider the state $\rho_{\mathrm{opt}}$ from (\ref{eqn:optimal_state_main_text}) that achieves the highest end-to-end fidelity. Secondly, we consider the R state $\rho_R$ from (\ref{eqn:R_state}). The QKD measurement basis for each state was the one found to provide maximum secret-key fraction.
%To show that the difference between $\text{SKF}_\text{uncon}$ and $\text{SKF}_\text{con}$ can be substantial, we have numerically calculated these two quantities for the case where $\rho$ is the state (\ref{eqn:optimal_state_main_text}) that achieves the optimum post-swap fidelity $F'_{\max}$.
% i.e. .
% \begin{equation}
%     \rho = p\ketbra{\Psi_{00}} + (1-p)\ketbra{\psi},
% \end{equation}
% with $ \ket{\psi} = \sqrt{\tilde{F}} \ket{\Psi_{00}} + \sqrt{1-\tilde{F}} \ket{\Psi_{11}}$, where       $\tilde{F} = (F-p)/(1-p)$ is the fidelity of the noisy component.
The results can be seen in Figure~\ref{fig:skf}.
We see from Figure \ref{fig:skf_opt_state} that, when the initial states are $\rho_{\mathrm{opt}}^{\otimes N}$, the Bell-diagonal approximation (equivalently, the non-postselected protocol) causes a significant reduction in the secret-key fraction, especially for repeater chains with a larger number of initial states $N$. By contrast, from Figure \ref{fig:skf_r_state} we see that when the initial states are $\rho_{R}^{\otimes N}$, the Bell-diagonal approximation (non-postselected protocol) causes a negligible reduction in the resulting secret-key fraction. This behaviour reflects the discussion in Section \ref{sec:bounds_discussion}, where we saw that there is not a significant difference in end-to-end fidelity from the Bell-diagonal approximation when swapping $\rho_R^{\otimes 2}$ (see (\ref{eqn:post_swap_fidelity_orth_states_1})). This is in contrast to swapping $\rho_{\mathrm{opt}}^{\otimes 2}$, which admits the greatest possible variation in end-to-end fidelity above the Bell-diagonal approximation.

We see in Figure \ref{fig:skf_opt_state} that the secret-key fraction is reduced drastically when the Werner approximation is used for the initial states $\mathcal{W}(\rho_{\mathrm{opt}})^{\otimes N}$, and we see that the length of the chain over which it is possible to distil key is limited to $N=3$ initial states. By contrast, with the Bell-diagonal approximation, one can distil key for any length of chain. The reason for this is as follows: we note that $\mathcal{B}(\rho_{\mathrm{opt}}) = F\ketbra{\Psi_{00}} + (1-F)\ketbra{\Psi_{11}}$ is a rank-two Bell-diagonal state. Then, recalling Lemma \ref{lem:bound_e2e_fidelity_BD} and surrounding discussion, the resulting state $\rho'_{\mathrm{opt}}$ after the non-postselected swapping of $\rho_{\mathrm{opt}}^{\otimes N}$ has Bell-diagonal components
\begin{align}
   \mathcal{B}\left(\rho'_{\mathrm{opt}}\right) = &\left(\frac{1}{2} + \frac{1}{2}(2F-1)^N \right)\ketbra{\Psi_{00}} \\ &+ \left(\frac{1}{2} - \frac{1}{2}(2F-1)^N \right)\ketbra{\Psi_{11}}.
\end{align}
By (\ref{eqn:qberX})-(\ref{eqn:qberZ}), the resulting QBER in each basis is then
\begin{align}
    Q_X(\rho'_{\mathrm{opt}}) &= \frac{1}{2} - \frac{1}{2}(2F-1)^N \\ 
    Q_Y(\rho'_{\mathrm{opt}}) &= 0                                  \\ 
    Q_Z(\rho'_{\mathrm{opt}}) &= \frac{1}{2} - \frac{1}{2}(2F-1)^N.
\end{align}
Choosing to measure in the $X$ and $Y$ bases then provides the highest secret-key fraction, given by 
\begin{align}
    \text{SKF}\left(\rho'_{\mathrm{opt}}\right) &= 1 - h\left(\frac{1}{2} - \frac{1}{2}(2F-1)^N\right) \\ &> 1 - h\left(\frac{1}{2}\right) = 0. \label{eqn:skf_rho_opt>0}
\end{align}
In particular, the fact that $\mathcal{B}(\rho_{\mathrm{opt}})$ is of rank two means that the secret-key fraction is greater than zero for any number of initial states $N$ in the chain. By contrast, with the Werner approximation for the initial states $\mathcal{W}(\rho_{\mathrm{opt}})^{\otimes N}$, by Corollary \ref{cor:swap_werner} the end-to-end state $\rho'_{\mathcal{W}}$ is also Werner with fidelity given in (\ref{eqn:werner_fidelity_N_swaps}). The corresponding QBER for each basis is then 
\begin{align}
    Q_X(\rho'_{\mathcal{W}}) = Q_Y(\rho'_{\mathcal{W}})  &= Q_Z(\rho'_{\mathcal{W}}) \nonumber \\  &= \frac{1}{2} - \frac{1}{2} \left(\frac{4F-1}{3} \right)^N.
\end{align}
We therefore see that 
\begin{equation}
    \text{SKF}(\rho_{\mathcal{W}}') = \max \! \left(0,  1- 2h\!\left(\frac{1}{2} - \frac{1}{2} \left(\frac{4F-1}{3} \right)^N\right) \right) ,
\end{equation}
which will eventually decrease to zero as $N$ increases. Since in Figure \ref{fig:skf_r_state}, the initial states $\rho_R$ are each set to have the same fidelity $\bra{\Psi_{00}}\rho_R \ket{\Psi_{00}} = \bra{\Psi_{00}}\rho_{\mathrm{opt}} \ket{\Psi_{00}} = 0.95$, we have $ \mathcal{W}(\rho_R)=\mathcal{W}(\rho_{\mathrm{opt}})$, and the Werner approximation gives the same result in both cases.

When the initial states are instead R states $\rho_R^{\otimes N}$, in contrast to the case of the optimal states, the secret-key fraction with the Bell-diagonal approximation will eventually reach zero.  This is because the Bell-diagonal approximation of an R state is given by 
\begin{align}
    \mathcal{B}(\rho_R) = F\ketbra{\Psi_{00}}  + &\frac{1}{2}(1-F) \ketbra{\Psi_{10}} \nonumber \\ &+ \frac{1}{2}(1-F) \ketbra{\Psi_{11}},
\end{align}
and this has rank three. From the map (\ref{eqn:swap_bell_diag_outcome}), it can be seen that swapping identical rank-three Bell-diagonal states results in a rank-four state. Therefore, the non-postselected outcome of swapping the states $\rho_R^{\otimes N}$ will result in an end-to-end state $\rho'_R$ such that $\mathcal{B}(\rho'_R)$ has rank four. In particular, the secret-key fraction will eventually decrease to zero as the number of swaps $N$ increases, unlike the behaviour we saw for $\rho_{\mathrm{opt}}$ in (\ref{eqn:skf_rho_opt>0}), where the secret-key fraction was always positive since the end-to-end state was always within the rank-two subspace.

% In Figure \ref{fig:skf_r_state}, we see that the Bell-diagonal approximation (unconditional protocol) causes negligible reduction in the secret-key fraction. Again, the Werner approximation causes significant reduction in the secret-key fraction. \todo{Explanation}

% In particular, the figure shows that the number of swaps before the secret-key fraction becomes zero can be extended using the conditional protocol.
% While only three states can be swapped together in the unconditional protocol, the conditional protocol can still be used to distill a key for six elementary links, effectively doubling the distance over which quantum key distribution is feasible. Alternatively, using the Bell-diagonal approximation halves the feasible distance.
% \begin{figure}[t]
% \includegraphics[width=90mm]{figures/qkd_number_of_swaps.png}
% \caption{ \raggedright
%     \textbf{Obtainable secret-key fraction (SKF) when performing BBM92 quantum key distribution in a repeater chain of $N$ initial states $\rho^{\otimes N}$, where $\rho$ is given by (\ref{eqn:optimal_state_main_text}) with $p=0.6$ and $F = 0.97$.} The blue line is $\text{SKF}_\text{uncon}$, or equivalently when using the Bell-diagonal approximation of the initial states. The orange dashed line is the SKF of the conditional protocol, which gives an advantage when the Bell-diagonal approximation is not used. 
%     % The  $\rho$, copy of the state (\ref{eqn:optimal_state_main_text}) that was found to maximise the conditional post-swap fidelity, for $p=0.6$ and $F = 0.97$.
% }
% \label{fig:qkd}
% \end{figure}

\section{Conclusion}
\label{sec:conclusion}
We have seen that, for non-postselected swapping, using twirled approximations in a repeater chain can be exact or highly accurate in certain important scenarios. In particular, the Bell-diagonal approximation is exact for evaluating the Bell-diagonal components of the end-to-end state, and in many scenarios, non-postselected swapping and Bell-diagonal twirling are equivalent. Moreover, for non-postselected swapping the Werner approximation is accurate in a high-fidelity regime compared to the number of initial states in the chain. The disadvantages of twirled approximations mostly arise when postselecting on the BSM measurement outcome. For postselected swapping, we have presented bounds on the end-to-end fidelity, given a general noisy form for the initial states when there are $N=2$ initial states in the chain. 
%We have seen that, for conditional swapping, states are more susceptible to a loss in performance after twirling if they have a noisy component with a high fidelity to $\ket{\Psi_{00}}$, or significant off-Bell-diagonal outcomes. 
%Finally, we have illustrated an example of the impact of twirled approximations 
With an example of evaluating the secret-key fraction when performing QKD, we demonstrated how the insights from our work may be used to determine whether the twirled approximation is accurate in a given scenario. 

\begin{acknowledgments}
We thank S{\'e}bastian de Bone, Filip Rozp\k{e}dek, Kenneth Goodenough and Hemant Sharma for helpful discussions. We also thank Álvaro G. Iñesta and Luca Marchese for proofreading the manuscript and providing feedback. BD acknowledges support from the KNAW Ammodo award (SW). GA and SW acknowledge support from the Quantum Internet Alliance (EU Horizon Europe grant agreement No. 101102140). GA acknowledges support from NSF CQN grant number 1941583, NSF grant number 2346089, and NSF grant number 2402861.
\end{acknowledgments}

\bibliography{refs}

\onecolumngrid
\appendix
\newpage
\counterwithin{definition}{section}
\counterwithin{proposition}{section}
\counterwithin{lemma}{section}
\counterwithin{theorem}{section}
\counterwithin{corollary}{section}

\section{Non-postselected swapping}
\label{app:unconditional_swapping}
\subsection{Repeater chains with $N=2$}
\label{app:unconditional_swapping_N=2}
\begin{proof}[Proof of Corollary \ref{cor:teleportation_pauli_channel}]
    Let the Bell-diagonal elements of $\rho$ be given by $\lambda_{ij}$, as in (\ref{eqn:twirled_approx_BD}). Then, by (\ref{eqn:twirl_inside_operators}), we have 
\begin{align}
     \Lambda_{\rho}^{\mathrm{tel}}(\sigma)\! &= 4 \bra{\Psi_{00}} \sigma_C \otimes \mathcal{B}(\rho_{AB}) \ket{\Psi_{00}}_{CA} \\ &= 4\sum_{i,j=0}^1 \lambda_{ij} \bra{\Psi_{00}} \sigma_C \otimes \ketbra{\Psi_{ij}}_{AB} \ket{\Psi_{00}} \\ &= \sum_{i,j=0}^1 \lambda_{ij} X^i Z^j \sigma  (X^i Z^j)^{\dag},
\end{align}
where in the final step we have used
\begin{align}
\bra{\Psi_{00}} \sigma_C \otimes \ketbra{\Psi_{ij}}_{AB} \ket{\Psi_{00}}_{CA}  &=  (X^i Z^j)_B \bra{\Psi_{00}} \sigma_C \otimes \ketbra{\Psi_{00}}_{AB}\! \ket{\Psi_{00}}_{CA}  (X^i Z^j)^{\dag }_B,
\end{align}
and the noticed that $$\bra{\Psi_{00}} \sigma_C \otimes \ketbra{\Psi_{00}}_{AB}\! \ket{\Psi_{00}}_{CA} = \frac{1}{4} X^i Z^j \sigma (X^i Z^j)^{\dag }$$
is the (non-normalised) result after the perfect teleportation of $\sigma_C$.
\end{proof}
\begin{proof}[Proof of Lemma \ref{lem:swap_bell_diagonal}]
In this proof, we make use of the flip-flop trick, which is that for any linear operator $M$, we have 
\begin{equation}
    M \otimes I \ket{\Psi_{00}} = I \otimes M^T \ket{\Psi_{00}}.
    \label{eqn:flip-flop}
\end{equation}
 This is also known as the flip-flop trick.

Suppose that in the entanglement swap, the BSM outcome is $mn$. The output state is then given by 
\begin{equation}
   \rho'_{mn}  = \frac{L_{mn}}{\Tr \left[ L_{mn} \right]},
   \label{eqn:sigma_L_mn}
\end{equation}
where 
\begin{equation}
    L_{mn} =  (Z^n X^m)_{B_2} \bra{\Psi_{mn}} \mathcal{B}(\rho_1) \otimes  \mathcal{B}(\rho_2) \ket{\Psi_{mn}}_{A_1 A_2} (X^m Z^n)_{B_2}.
    \label{eqn:L_mn}
\end{equation}
We now compute $L_{mn}$. We firstly consider the impact of each diagonal element:
\begin{align}
(Z^n X^m)_{B_2}\bra{\Psi_{mn}}_{A_1 A_2} & \left[ \ket{\Psi_{i_1 j_1}}_{B_1 A_1}\otimes \ket{\Psi_{i_2 j_2}}_{A_2 B_2} \right] \nonumber \\  &= \bra{\Psi_{00}}_{A_1 A_2} (Z^n X^m)_{A_2} \left[(X^{i_1} Z^{j_1})_{A_1} (X^{i_2} Z^{j_2})_{A_2} (Z^m X^m)_{B_2}\ket{\Psi_{00}}_{B_1 A_1}\otimes \ket{\Psi_{00}}_{A_2 B_2} \right] \nonumber \\ 
&\stackrel{a}{=} (Z^n X^m Z^{j_2}X^{i_2} X^m Z^n X^{i_1} Z^{j_1})_{B_2}  \bra{\Psi_{00}}_{A_1 A_2} \Big[ \ket{\Psi_{00}}_{B_1 A_1}\otimes \ket{\Psi_{00}}_{A_2 B_2} \Big] \nonumber \\
&\stackrel{b}{=} \pm 1 \cdot \left(X^{2m + i_1 + i_2} Z^{2n + j_1 + j_2}\right)_{B_2} \cdot \frac{1}{2} \ket{\Psi_{00}}_{B_1 B_2} \nonumber\\ &\stackrel{c}{=} (\pm 1)\cdot \frac{1}{2} \ket{\Psi_{i_1+i_2,j_1+j_2}}_{B_1 B_2},
\label{eqn:swap_pure_bell_states_indices}
\end{align}
where the addition in the subscript is modulo 2. In step ($a$), we have made use of the flip-flop trick multiple times to move all Pauli operators onto register $B_2$. In step ($b$), we have used the fact that
\begin{equation}
    \bra{\Psi_{00}}_{A_1 A_2} \Big[ \ket{\Psi_{00}}_{B_1 A_1}\otimes \ket{\Psi_{00}}_{A_2 B_2} \Big]  = \frac{1}{2}\ket{\Psi_{00}}_{B_1 B_2} \label{eqn:project_Psi_00_on_Psi_00x2}
\end{equation}
and that reordering Pauli operators may sometimes incur a factor of $-1$. In step $(c)$, we have used the definition (\ref{eqn:bell_basis}) of the Bell basis. Relabelling the eigenvalues as $\mathcal{B}(\rho_1)  = \sum_{i,j}\lambda_{ij} \ketbra{\Psi_{ij}}$ and  $\mathcal{B}(\rho_1)  = \sum_{i,j}\mu_{ij} \ketbra{\Psi_{ij}}$,  from (\ref{eqn:L_mn}) we see that 
\begin{align*}
    L_{mn} &= \sum_{i_1,j_1,i_2,j_2} \lambda_{i_1 j_1} \mu_{i_2 j_2} Z^m X^n \bra{\Psi_{mn}}_{A_1 A_2} \left[ \ketbra{\Psi_{i_1 j_1}}_{B_1 A_1} \otimes  \ketbra{\Psi_{i_2 j_2}}_{A_2 B_2} \right] \ket{\Psi_{mn}}_{A_2 B_2}X^n Z^m \\ &=  \sum_{i_1,j_1,i_2,j_2}  \lambda_{i_1 j_1} \mu_{i_2 j_2} \cdot (\pm 1)^2 \cdot \frac{1}{4} \ketbra{\Psi_{i_1+i_2,j_1+j_2}}_{B_1 A_1},
\end{align*}
and so 
\begin{align*}
    \Tr \left[ L_{mn} \right] = \sum_{i_1,j_1,i_2,j_2} \lambda_{i_1 j_1} \mu_{i_2 j_2} \cdot \frac{1}{4} = \frac{1}{4} = p'_{mn},
\end{align*}
due to normalisation of the initial states. In the above, $p'_{mn}$ is the probability of obtaining outcome $mn$ in the BSM. From (\ref{eqn:sigma_L_mn}), we therefore see that the full swap outcome after measuring $mn$ is 
\begin{equation*}
    \rho'_{B_1 B_1}  =  \sum_{i_1,j_1,i_2,j_2}  \lambda_{i_1 j_1} \mu_{i_2 j_2} \ketbra{\Psi_{i_1+i_2,j_1+j_2}}_{B_1 B_2}.
\end{equation*}
In particular, this is Bell-diagonal and independent of the measurement outcome. In the four-vector notation from (\ref{eqn:BD_as_4-vector}), for  $\mathcal{B}(\rho_1) \equiv (\lambda_0,\dots,\lambda_3)^T$ and $\mathcal{B}(\rho_2) \equiv (\mu_0,\dots,\mu_3)^T$, the end-to-end state is $\rho'_{mn} \equiv (\lambda_0',\dots,\lambda_3')^T$, where
    \begin{equation}
\left( \begin{array}{c}
    \lambda_0'   \\
    \lambda_1'   \\
    \lambda_2'   \\
    \lambda_3'   \\
\end{array} \right)  = 
\left( \begin{array}{c}
    \lambda_0 \mu_0 + \lambda_1 \mu_1 + \lambda_2 \mu_2 + \lambda_3 \mu_3   \\
    \lambda_0 \mu_1 + \lambda_1 \mu_0 + \lambda_2 \mu_3 + \lambda_3 \mu_2   \\
    \lambda_0 \mu_2 + \lambda_2 \mu_0 + \lambda_3 \mu_1 + \lambda_1 \mu_3   \\
    \lambda_0 \mu_3 + \lambda_3 \mu_0 + \lambda_1 \mu_2 + \lambda_2 \mu_1   
\end{array} \right) \equiv \rho'_{\mathcal{B}}.
\end{equation}
\end{proof}

\begin{proof}[Proof of Corollary \ref{cor:swap_inv_twirling}]
 Recalling Definition \ref{def:entanglement_swapping_protocol}, the average outcome state of a standard entanglement swap on $\rho_1\otimes \rho_2$ is given by
    \begin{equation}
         \rho' = I_2 \otimes \Lambda_{\rho_2}^{\mathrm{tel}}(\rho_1) = I_2 \otimes \Lambda_{\mathcal{B}(\rho_2)}^{\mathrm{tel}}(\rho_1).
    \end{equation}
    Now, letting $M \coloneqq  \rho_1 - \mathcal{B}(\rho_1)$ be the operator containing the off-diagonal components of $\rho_1$, by linearity it follows that
    \begin{equation}
        \rho' = I_2 \otimes \Lambda_{\mathcal{B}(\rho_2)}^{\mathrm{tel}}\left(\mathcal{B}(\rho_1) +M\right) = I_2 \otimes \Lambda_{\mathcal{B}(\rho_2)}^{\mathrm{tel}}\left(\mathcal{B}(\rho_1)\right) + I_2 \otimes \Lambda_{\mathcal{B}(\rho_2)}^{\mathrm{tel}}(M).
        \label{eqn:split_rho_1_orthog}
    \end{equation}
    We now claim that $\mathcal{B}\left( I_2 \otimes \Lambda_{\mathcal{B}(\rho_2)}^{\mathrm{tel}}(M) \right) = 0$. Denoting $\Psi_{ab} \equiv \ketbra{\Psi_{ab}}$, by linearity in $\rho_2$ and $M$ it suffices to show that the terms corresponding to basis elements vanish, 
    \begin{equation}
        \mathcal{B}\left( I_2 \otimes \Lambda_{\Psi_{ab}}^{\mathrm{tel}}(\ket{\Psi_{ij}}\bra{\Psi_{kl}}) \right) = 0
    \end{equation}
    for $\ket{\Psi_{ij}} \neq \ket{\Psi_{kl}}$. To show this, we recall from the expression (\ref{eqn:teleportation_channel}) for the standard teleportation channel that 
    \begin{equation}
        I_2 \otimes \Lambda_{\Psi_{ab}}^{\mathrm{tel}} (\ket{\Psi_{ij}}\bra{\Psi_{kl}}) = \sum_{m,n} K_{mn}
    \end{equation}
    where 
    \begin{equation}
        K_{mn} \coloneqq Z^m X^n \bra{\Psi_{mn}}_{A_1 A_2} \left[ \ketbra{\Psi_{ab}}_{B_1 A_1} \otimes  \ket{\Psi_{ij}}\bra{\Psi_{kl}}_{A_2 B_2} \right] \ket{\Psi_{mn}}_{A_1 A_2} X^n Z^m.
    \end{equation}
    Using (\ref{eqn:swap_pure_bell_states_indices}), we see that
    \begin{align}
       K_{mn} =  \pm  \frac{1}{4} \ket{\Psi_{a+i,b+j}}\bra{\Psi_{a+k,b+l}},
    \end{align}
    which is off-diagonal (because we have assumed $\ket{\Psi_{ij}} \neq \ket{\Psi_{kl}}$). Therefore,  $\mathcal{B}(K_{mn}) = 0$. It then follows that 
    \begin{equation}
        \mathcal{B}\left(I_2 \otimes \Lambda_{\Psi_{ab}}^{\mathrm{tel}}(\ket{\Psi_{ij}}\bra{\Psi_{kl}}) \right)= \mathcal{B}\left( \sum_{m,n} K_{mn}\right) =\sum_{m,n} \mathcal{B}\left(  K_{mn}\right) =  0,
    \end{equation}
    and therefore
    \begin{equation}
        \mathcal{B}\left( I_2 \otimes \Lambda_{\mathcal{B}(\rho_2)}^{\mathrm{tel}}(M) \right) = 0.
    \end{equation}
    From (\ref{eqn:split_rho_1_orthog}), applying $\mathcal{B}$ therefore results in
    \begin{align}
        \mathcal{B}(\rho') &= \mathcal{B}\left( I_2 \otimes \Lambda_{\mathcal{B}(\rho_2)}^{\mathrm{tel}}(\mathcal{B}(\rho_1)) \right) \\ &= \mathcal{B}\left( \rho'_{\mathcal{B}} \right) = \rho'_{\mathcal{B}},
    \end{align}
    where we have made use of (\ref{eqn:unconditional_swap_BD}), where we saw that the result of a non-postselected swap on Bell-diagonal states $\mathcal{B}(\rho_1)\otimes \mathcal{B}(\rho_2)$ is the same as a postselected swap. The result of the postselected swap is the state $\rho'_{\mathcal{B}}$ obtained in Lemma \ref{lem:swap_bell_diagonal}.
\end{proof}

\subsection{Non-postselected swapping on repeater chains with $N>2$}
\label{app:swap_and_correct}
\begin{definition}[Swap-and-correct protocol, technical]
    For a length-$N$ repeater chain, a \textit{swap-and-correct} protocol $\mathcal{P}$ dictates where to apply Pauli corrections. Given the $N-1$ BSM outcomes that form the syndrome $\vec{s}$, $\mathcal{P}$ is a map 
    \begin{equation}
        \mathcal{P} : \{I,X,Z,XZ\}^{N-1} \rightarrow \{I,X,Z,XZ\}^{N+1}
    \end{equation}
    %\begin{equation}
    %    \mathcal{P} : \{\otimes_{k=1}^{N-1} \ket{\Psi_{i_k j_k}} \} \rightarrow \{I,X,Z,ZX\}^{\otimes (N+1)}
    %\end{equation}
     such that $\mathcal{P}_k(\vec{s})$ is the correction applied to node $k$, given the syndrome $\vec{s}$. The syndrome is denoted such that $s_i \equiv X^m Z^n$ means that outcome $\ket{\Psi_{mn}}$ was measured on node $i$. Moreover, $\mathcal{P}$ satisfies the following two properties:
     \begin{enumerate}[(A)]
         \item \textbf{$\mathcal{P}$ is physically implementable.} For any swap-and-correct protocol $\mathcal{P}$, there exists an associated permutation $\alpha \in \mathrm{Sym}(N-1)$ in which the $N-1$ BSMs are carried out, where $\mathrm{Sym}$ denotes the symmetric group. For $\mathcal{P}$ to be physical, then before the $k$th BSM, the correction $\mathcal{P}_{\alpha(k)}$ must only depend on outcomes of BSMs that have already been carried out, which are given by $(\alpha(1),\dots,\alpha(k-1))$. 
         \item \textbf{$\mathcal{P}$ is correct.} For any syndrome $\vec{s}$, $\mathcal{P}$ transforms $\ketbra{\Psi_{00}}^{\otimes N}$ into $\ketbra{\Psi_{00}}$.
     \end{enumerate}
     \label{def:swap_and_correct_protocol_technical}
\end{definition}
By the assumption (A), we slightly abuse notation to write $ \mathcal{P}_{\alpha(k)}(s_{\alpha(1)},\dots,s_{\alpha(k-1)}) \equiv \mathcal{P}_{\alpha(k)}(\vec s)$. Note that the first correction, $\mathcal{P}_{\alpha(1)}$, is independent of $\vec s$.

In particular, given a swap-and-correct protocol $\mathcal{P}$, it may be executed as follows. Given syndrome $\vec s$,
\begin{enumerate}[(1)]
    \item Apply correction $\mathcal{P}_{\alpha(1)} \equiv \mathcal{P}_{\alpha(1)}(\vec s)$ to node $\alpha(1)$. Apply BSM at node $\alpha(1)$ to get outcome $s_{\alpha(1)}$.
    \item Apply correction $\mathcal{P}_{\alpha(2)}(s_{\alpha(1)})\equiv \mathcal{P}_{\alpha(2)}(\vec s)$ to node $\alpha(2)$. Apply BSM at node $\alpha(2)$ to get outcome $s_{\alpha(2)}$.
    \item[$\vdots$] \phantom{text} 
    \item[$(N\!-\!1)$] Apply correction $\mathcal{P}_{\alpha(N-1)}(s_{\alpha(1)},\dots,s_{\alpha(N)}) \equiv \mathcal{P}_{\alpha(N-1)}(\vec s)$ to node $\alpha(N-1)$. Apply BSM at node $\alpha(N-1)$ to get outcome $s_{\alpha(N-1)}$.
    \item[$(N)$] Apply corrections $\mathcal{P}_{0}(\vec{s})$ and $\mathcal{P}_{N}(\vec{s})$ to nodes $0$ and $N$.
\end{enumerate}
We note that in principle, corrections may be applied at any point in the protocol up to the BSM on that node. Similarly, corrections may be applied at the end nodes at any point in the protocol. However, both strategies are captured by the above formalism by simply combining all corrections and applying them just before the BSM (for the repeater nodes), and after all BSMs (for the end nodes).

% We firstly introduce some notation. Let $A = \{I,X,Z,XZ\}$ be the set of possible corrections. 

In the following, we consider swap-and-correct protocols $\mathcal{P}$. We also use $A= \{I,X,Z,XZ\}$ as shorthand for the set of Pauli matrices (up to a phase). 
As discussed above, such a protocol consists of Bell-state measurements and Pauli corrections, potentially to the repeater nodes as well as the end nodes
%PROPERTIES OF SINGLE PAULI CHANNELS: combine into proposition + list important  properties?
For the end nodes, which have indices $0$ and $N$, there is only one qubit that this can be applied to. However, the repeater nodes $1,\dots, N-1$ each hold two qubits, and so there is a choice of which qubit to apply the correction. We now show that applying the correction to either qubit will give the same result. This allows us to simplify notation later on.

Let the qubit registers in the $k$th node be denoted as $k_1$ and $k_2$. Given a swap-and-correct protocol with associated permutation $\alpha \in \mathrm{Sym}(N-1)$, the correction at the $k$th node is $\mathcal{P}_{\alpha^{-1}(k)}(\vec{s})$. Directly after the correction, the BSM on node $k$ will be applied. Suppose that the correction is applied to register $k_1$, and the BSM outcome is $s\in A$.
%(recalling from Definition \ref{def:swap_and_correct_protocol_technical} that we use of $A$ to denote the result of the BSM). 
Letting $c = \mathcal{P}_{\alpha^{-1}(k)}(\vec{s})\in A$ denote the correction, the resulting projection on the total state $\rho$ of the chain is then 
\begin{align}
    \Tr_{k_1 k_2} \left[ s_{k_2}\ketbra{\Psi_{00}}_{k_1 k_2} s_{k_2}^{\dag} c_{k_1} \rho c_{k_1}^{\dag} \right] &\stackrel{\mathrm{i}}{=} \Tr_{k_1 k_2} \left[ s_{k_2} c_{k_2}\ketbra{\Psi_{00}}_{k_1 k_2} c_{k_2}^{\dag} s_{k_2}^{\dag}  \rho  \right] \\ 
    &\stackrel{\mathrm{ii}}{=} \Tr_{k_1 k_2} \left[  c_{k_2} s_{k_2} \ketbra{\Psi_{00}}_{k_1 k_2} s_{k_2}^{\dag} c_{k_2}^{\dag}  \rho  \right] \\ 
    &\stackrel{\mathrm{iii}}{=} \Tr_{k_1 k_2} \left[   s_{k_2} \ketbra{\Psi_{00}}_{k_1 k_2} s_{k_2}^{\dag} c_{k_2}  \rho c_{k_2}^{\dag} \right] \label{eqn:Pauli_correction_k_2}
\end{align}
where we have (i) used cyclicity of the trace and the flip-flop trick (\ref{eqn:flip-flop}), (ii) used the fact that $sc = \pm cs$ for $c,s \in A$,  and (iii) used cyclicity of the trace and $s^{\dag} = \pm s$, for $s\in A$.
In particular, we notice that (\ref{eqn:Pauli_correction_k_2}) corresponds to applying the Pauli correction to register $k_2$. We therefore see that, as long as a BSM is applied after the correction, it does not matter to which qubit the correction is applied. The same holds for Pauli operators arising from Bell-diagonal twirling (see Definition \ref{def:bell_diag_twirl_s&c}).
\begin{definition}
    For a unitary $U$ we denote the corresponding channel where the unitary is applied as
\begin{equation}
    U(\rho) = U\rho U^{\dagger}.
    \label{eqn:single_pauli_channel}
\end{equation}
\end{definition}
We will use the above notation principally for $U \in A$. 
\begin{definition} 
For any map $\Lambda$ acting on two-qubit states, we denote $(\Lambda)_k$ to be this channel applied to the two qubits in node $k$. 
\end{definition}
In particular, $(s)_k$ is the Pauli correction $s$ applied to node $k$. 
\begin{definition}[Twirling of the $k$th state]
    For $k=1,\dots, N-1$ we denote the twirling map of the $k$th state in the chain (shared between nodes $k-1$ and $k$) as 
\begin{equation}
    \mathcal{B}_{k-1,k} = \frac{1}{4} \sum_{s\in A} (s)_{k-1} \circ (s)_k.
    \label{eqn:bell_diag_twirl_s&c}
\end{equation}
\label{def:bell_diag_twirl_s&c}
\end{definition}
In (\ref{eqn:bell_diag_twirl_s&c}), we have recalled the Bell-diagonal twirling map from Lemma \ref{lem:pauli_twirling}.
To simplify notation, in (\ref{eqn:bell_diag_twirl_s&c}) we have not specified the specific qubit of each repeater nodes to which twirling is applied, because we will always be interested in the case where the map $\mathcal{B}_{k-1,k}$ is applied \textit{before} the BSMs. By the same argument used to obtain (\ref{eqn:Pauli_correction_k_2}), applying the correction to either qubit of the repeater node is equivalent. 

\begin{lemma}[Properties of Pauli operators and Bell-diagonal twirling]
    The following properties hold:
    \begin{enumerate}[(i)]
        \item For $s_1, s_2 \in A$, 
        \begin{equation}
        (s_1)_k \circ (s_2)_k = (s_1 s_2)_k = (s_2 s_1)_k = (s_2)_k \circ (s_1)_k.
        \label{eqn:single_pauli_channels_commute}
    \end{equation}
        \item For $s\in A$, $(s^{\dagger})_k = (s)_k$.
        \item For $s\in A$, we have $(s)_{k-1} \circ \mathcal{B}_{k-1,k} = \mathcal{B}_{k-1,k} \circ (s)_k$.
    \end{enumerate}
    \label{lem:properties_paulis_and_twirl}
\end{lemma}
\begin{proof}
    \begin{enumerate}[($i$)]
        \item We use the fact that, although interchanging the order of the Pauli operators may incur a sign difference $ s_1 s_2 = \pm s_2 s_1$, this will not affect the channel (\ref{eqn:single_pauli_channel}) because the sign is global.
        \item We use that for any $s\in A$ we have $s^{\dagger } = \pm s$, and the incurred sign does not affect the channel.
        \item Recall the channel $\mathcal{B}_{k-1,k}$ from (\ref{eqn:bell_diag_twirl_s&c}). For $s\in A$, we have 
            \begin{align}
                (s)_{k-1} \circ \mathcal{B}_{k-1,k} &=  \frac{1}{4} \sum_{s'\in A} (ss')_{k-1} \circ (s')_k \\ &= \frac{1}{4} \sum_{r\in A} (r)_{k-1} \circ (s^{\dagger} r)_k \\ &= \mathcal{B}_{k-1,k} \circ (s)_k, 
            \end{align}
where in the second line we have made the change of variable $r = s s'$, and used properties $(i)$ and $(ii)$.
    \end{enumerate}
\end{proof}

\begin{definition}[Bell-state projection, single repeater node]
    For $k = 1,\dots , N-1$, suppose that BSM outcome $mn$ is obtained at node $k$. Let $s = X^m Z^n \in A$. We denote the unnormalised map corresponding to projection onto this Bell state as 
\begin{equation}
    \left(M_{s}\right)_k (\rho) = \Tr_{k} \left[ s\ketbra{\Psi_{00}}_k s^{\dag} \cdot \rho \right].
    \label{eqn:BSM_projection_map}
\end{equation}
\end{definition}
From (\ref{eqn:single_pauli_channel}), we note that (\ref{eqn:BSM_projection_map}) may be written as  
\begin{equation}
    \left(M_{s}\right)_k = (M_{I} \circ s)_k.
    \label{eqn:BSM_breakdown}
\end{equation}
\begin{definition}[Conditional projection, swap-and-correct protocol]
Let $\mathcal{P}$ be a swap-and-correct protocol for repeater chains with $N$ initial states and $N-1$ repeaters. 
    Let $\vec{s}\in A^{N-1}$ be the swap syndrome. We define
    \begin{equation}
    \Lambda_{\mathcal{P},\vec{s}} = \bigcirc_{k=1}^{N-1} (M_{s_k})_k   \circ \bigcirc_{j=0}^{N} (\mathcal{P}_{j}(\vec{s}))_j.
    \label{eqn:Lambda_P,S}
    \end{equation}
    to be the non-normalised map where syndrome $\vec s$ is measured and $\mathcal{P}$ is applied.
    \label{def:Lambda_P,s}
\end{definition}

\begin{lemma}[Correctness condition]
For any swap-and-correct protocol $\mathcal{P}$ and syndrome $\vec s$, the correctness condition (property B in Definition \ref{def:swap_and_correct_protocol_technical}) implies that
\begin{equation}
    \mathcal{P}_{0}(\vec{s}) \cdot \mathcal{P}_{N}(\vec{s}) \cdot \prod_{j = 1}^{N-1} s_j \mathcal{P}_{j}(\vec{s}) = \pm I.
    \label{eqn:condition_from_correctness}
\end{equation}
\label{lem:correctness_property}
\end{lemma}
\begin{proof}
%     We firstly note that, given $s\in A$ and the pure state $\ketbra{\Psi_{00}}_{k-1,k}$ shared between registers $k-1$ and $k$, we have
%     \begin{align}
%          (s)_{k-1} \left( \ketbra{\Psi_{00}}_{k-1,k} \right) = (s)_{k} \left( \ketbra{\Psi_{00}}_{k-1,k}\right),
%     \end{align}
% where we have used the identity \ref{eqn:flip-flop}. 
Recalling (\ref{eqn:Lambda_P,S}), (\ref{eqn:BSM_breakdown}) and (\ref{eqn:single_pauli_channels_commute}), we may rewrite 
\begin{equation}
    \Lambda_{\mathcal{P},\vec{s}} = \bigcirc_{k=1}^{N-1} (M_{I})_k  \circ \bigcirc_{j=1}^{N-1} (s_j \mathcal{P}_{j}(\vec{s}))_j \circ (\mathcal{P}_{0}(\vec{s}))_0 \circ (\mathcal{P}_{N}(\vec{s}))_N.
    \label{eqn:Lambda_P,S_BSM_breakdown}
\end{equation}
We firstly note that, given $s\in A$ and the pure state $\ketbra{\Psi_{00}}_{i,j}$ shared between registers $i$ and $j$, we have
\begin{align}
     (s)_{i} \left( \ketbra{\Psi_{00}}_{i,j} \right) = (s)_{j} \left( \ketbra{\Psi_{00}}_{i,j}\right),
     \label{eqn:move_paulis_rho_ideal}
\end{align}
where we have used the flip-flop trick (\ref{eqn:flip-flop}).

Letting $\rho_{\mathrm{ideal}} = \ketbra{\Psi_{00}}^{\otimes N}$, it follows from (\ref{eqn:Lambda_P,S_BSM_breakdown}) and (\ref{eqn:move_paulis_rho_ideal}) that one may move all Pauli operators to act on node $N$ when the channel acts on this state,
    \begin{align}
        \Lambda_{\mathcal{P},\vec{s}}  \left( \rho_{\mathrm{ideal}}\right) &= \bigcirc_{k=1}^{N-1} (M_I)_k \circ \bigcirc_{j = 1}^{N-1} (s_j \mathcal{P}_{j}(\vec{s}))_N \circ (\mathcal{P}_{0}(\vec{s}))_N \circ (\mathcal{P}_{N}(\vec{s}))_N (\rho_{\mathrm{ideal}}) \nonumber \\ &= \bigcirc_{k=1}^{N-1} (M_I)_k \circ \Big( \mathcal{P}_{0}(\vec{s}) \cdot \mathcal{P}_{N}(\vec{s}) \cdot \prod_{j = 1}^{N-1} s_j \mathcal{P}_{j}(\vec{s}) \Big)_N (\rho_{\mathrm{ideal}}). \label{eqn:rho_ideal_paulis_at_end}
    \end{align}
Now, by recursively applying (\ref{eqn:project_Psi_00_on_Psi_00x2}), we see that  
\begin{equation}
    \bigcirc_{k=1}^{N-1} (M_I)_k (\rho_{\mathrm{ideal}}) = \frac{1}{4^{N-1}} \ketbra{\Psi_{00}}_{0,N}.
    \label{eqn:project_Psi_00_on_Psi_00xN}
\end{equation}
Then, (\ref{eqn:rho_ideal_paulis_at_end}) simplifies to 
\begin{equation}
    \Lambda_{\mathcal{P},\vec{s}}  \left( \rho_{\mathrm{ideal}}\right) = \frac{1}{4^{N-1}}\Big( \mathcal{P}_{0}(\vec{s}) \cdot \mathcal{P}_{N}(\vec{s}) \cdot \prod_{j = 1}^{N-1} s_j \mathcal{P}_{j}(\vec{s}) \Big)_N \left( \ketbra{\Psi_{00}}_{0,N}\right),
\end{equation}
from which we see that $\Tr \left[\Lambda_{\mathcal{P},\vec{s}}  \left( \rho_{\mathrm{ideal}}\right)\right] = 1/4^{N-1}$. The output state after measuring syndrome $\vec{s}$ and applying protocol $\mathcal{P}$ is then
\begin{equation}
    \frac{\Lambda_{\mathcal{P},\vec{s}}  \left( \rho_{\mathrm{ideal}}\right)}{\Tr \left[\Lambda_{\mathcal{P},\vec{s}}  \left( \rho_{\mathrm{ideal}}\right)\right]} = \Big( \mathcal{P}_{0}(\vec{s}) \cdot \mathcal{P}_{N}(\vec{s}) \cdot \prod_{j = 1}^{N-1} s_j \mathcal{P}_{j}(\vec{s}) \Big)_N \left( \ketbra{\Psi_{00}}_{0,N}\right).
\end{equation}
For the correctness condition $\Lambda_{\mathcal{P},\vec{s}}  \left( \rho_{\mathrm{ideal}}\right) = \ketbra{\Psi_{00}}_{0,N}$ to hold, we therefore require
\begin{equation}
    \mathcal{P}_{0}(\vec{s}) \cdot \mathcal{P}_{N}(\vec{s}) \cdot \prod_{j = 1}^{N-1} s_j \mathcal{P}_{j}(\vec{s}) = \pm I.
\end{equation}
\end{proof}

\begin{lemma}
    For any swap-and-correct protocol $\mathcal{P}$ and syndrome $\vec{s}$, we have
    \begin{equation}
        \Lambda_{\mathcal{P},\vec{s}} \circ \bigcirc_{k=1}^N \mathcal{B}_{k-1,1} = \Lambda_{\mathrm{seq},\vec{I}}  \circ \bigcirc_{k=1}^N \mathcal{B}_{k-1,1},
        \label{eqn:syndrome_independence_twirled_inputs}
    \end{equation}
where $\vec{I} = (I,\dots,I)$ is the syndrome with outcome $00$ at every repeater node, and $\mathrm{seq}$ is the sequential swapping protocol. In particular, if Bell-diagonal twirling is firstly applied to all initial states, then the end-to-end state is independent of $\vec{s}$ and $\mathcal{P}$. 
\label{lem:s&c_bell_diag_P,s_indep}
\end{lemma}
\begin{proof}
 Recalling the identity (\ref{eqn:Lambda_P,S}) for $\Lambda_{\mathcal{P},\vec{s}}$ and property ($iii$) in Lemma \ref{lem:properties_paulis_and_twirl}, we see that 
\begin{align}
    \Lambda_{\mathcal{P},\vec{s}} \circ \bigcirc_{l=1}^N \mathcal{B}_{l-1,l} = \bigcirc_{k=1}^{N-1} (M_{s_k})_k \circ \bigcirc_{l=1}^N \mathcal{B}_{l-1,l} \circ \bigcirc_{j = 0}^N (\mathcal{P}_{j}(\vec{s}))_N . 
\end{align}
In particular, we have moved all Pauli corrections to act at register $N$. Recalling (\ref{eqn:BSM_breakdown}), we can do the same with all Pauli operators arising from the syndrome, to obtain
\begin{align}
    \Lambda_{\mathcal{P},\vec{s}} \circ \bigcirc_{l=1}^N \mathcal{B}_{l-1,l} &= \bigcirc_{k=1}^{N-1} (M_{I})_k \circ \bigcirc_{l=1}^N \mathcal{B}_{l-1,1} \circ \bigcirc_{j = 1}^{N-1} (s_j \mathcal{P}_{j}(\vec{s}))_N \circ (\mathcal{P}_{0}(\vec{s}))_N \circ (\mathcal{P}_{N}(\vec{s}))_N \nonumber \\ &=  \bigcirc_{k=1}^{N-1} (M_{I})_k \circ \bigcirc_{l=1}^N \mathcal{B}_{l-1,1} \circ \Big( \mathcal{P}_{0}(\vec{s}) \cdot \mathcal{P}_{N}(\vec{s}) \cdot \prod_{j = 1}^{N-1} s_j \mathcal{P}_{j}(\vec{s}) \Big)_N . \label{eqn:move_paulis_to_end_twirling}
\end{align}
By Lemma \ref{lem:correctness_property}, due to the correctness property of $\mathcal{P}$, the term inside the brackets is simply the identity.
Then,
\begin{align}
    \Lambda_{\mathcal{P},\vec{s}} \circ \bigcirc_{l=1}^N \mathcal{B}_{l-1,1} = \bigcirc_{k=1}^{N-1} (M_{I})_k \circ \bigcirc_{l=1}^N \mathcal{B}_{l-1,1} .
\end{align}
We now note that $\Lambda_{\mathrm{seq},\vec{I}} = \bigcirc_{k=1}^{N-1} (M_{I})_k $, where $\mathrm{seq}$ denotes the sequential swapping protocol. This is because, given syndrome $\vec{I}$, all corrections specified by this protocol are the identity. Therefore,
\begin{align}
    \Lambda_{\mathcal{P},\vec{s}} \circ \bigcirc_{l=1}^N \mathcal{B}_{l-1,1} =   \Lambda_{\mathrm{seq},\vec{I}} \circ \bigcirc_{l=1}^N \mathcal{B}_l.
    \label{eqn:discussion_Lambda_seq,I}
\end{align}
\end{proof}

\begin{proof}[Proof of Theorem \ref{thm:swap_and_correct_equivalence}]
Recalling Definition \ref{def:Lambda_P,s}, we may write down the map corresponding to \textit{non-postselected} swapping with swap-and-correct protocol $\mathcal{P}$ as
\begin{equation}
    \Lambda_{\mathcal{P}} \coloneqq \sum_{\vec{s}\in A^{N-1}} \Lambda_{\mathcal{P},\vec{s}},
    \label{eqn:unconditional_s&c}
\end{equation}
where $\Lambda_{\mathcal{P},\vec{s}}$ is given in (\ref{eqn:Lambda_P,S}). We now claim that 
\begin{equation}
    \mathcal{B}_{0,N} \circ \Lambda_{\mathcal{P}} = \Lambda_{\mathrm{seq},\vec{I}} \circ \bigcirc_{k=1}^N \mathcal{B}_{k-1,k},
    \label{eqn:T.Lambda=Lambda_seq.T}
\end{equation}
where 
\begin{equation}
    \mathcal{B}_{0,N} = \frac{1}{4}\sum_{s\in A} (s)_0 \circ (s)_N
    \label{eqn:twirling_e2e}
\end{equation}
denotes the Bell-diagonal twirling of the end-to-end state. We firstly make use of (\ref{eqn:unconditional_s&c}) to write
\begin{align}
    \mathcal{B}_{0,N} \circ \Lambda_{\mathcal{P}} = \sum_{\vec{s}}  \mathcal{B}_{0,N} \circ \Lambda_{\mathcal{P},\vec{s}},
    \label{eqn:T.Lambda_linearity}
\end{align}
Recalling (\ref{eqn:Lambda_P,S_BSM_breakdown}), we have 
\begin{align}
    \mathcal{B}_{0,N} \circ \Lambda_{\mathcal{P},\vec{s}} &= \mathcal{B}_{0,N} \circ \bigcirc_{k=1}^{N-1} (M_{I})_k  \circ \bigcirc_{j=1}^{N-1} (s_j \mathcal{P}_{j}(\vec{s}))_j  \circ (\mathcal{P}_{0}(\vec{s}) \mathcal{P}_{N}(\vec{s}))_N \\ &= \frac{1}{4}\sum_{t\in A}  \bigcirc_{k=1}^{N-1} (M_{I})_k \circ (t)_0  \circ \bigcirc_{j=1}^{N-1} (s_j \mathcal{P}_{j}(\vec{s}))_j  \circ ( t \mathcal{P}_{0}(\vec{s}) \mathcal{P}_{N}(\vec{s}))_N. \label{eqn:T.Lambda_P,S_breakdown}
\end{align}
In the first step, we have used (\ref{eqn:Lambda_P,S_BSM_breakdown}), and then moved the correction $(\mathcal{P}_{0}(\vec{s}))_0$ to act on register $N$. This is possible by applying property ($iii$) of Lemma \ref{lem:properties_paulis_and_twirl} together with the twirling operator $\mathcal{B}_{0,N}$. In the second step, we have expanded $\mathcal{B}_{0,N}$ according to (\ref{eqn:twirling_e2e}).

We now perform a change of variables: for $j=1,\dots, N-1$, we let
\begin{align}
    r_0 \coloneqq t, \;\;\; r_j \coloneqq s_{\alpha(j)} \mathcal{P}_{\alpha(j)}(\vec{s}) \text{  for  }j = 1,\dots,N-1, \label{eqn:mapping(t,s)-r}
\end{align}
where $\alpha$ is the permutation (associated with $\mathcal{P}$) in which BSMs are carried out -- see Definition \ref{def:swap_and_correct_protocol_technical}. 
We now claim that (\ref{eqn:mapping(t,s)-r}) defines a bijective map $(t,\vec{s}) \leftrightarrow \vec{r}$, for $\vec{r} = (r_0,\dots, r_{N-1})\in A^{N}$ and $(t,\vec{s}) \in A^{N}$. It suffices to show that the map (\ref{eqn:mapping(t,s)-r}) is invertible. We show invertibility by explicitly writing down the inverse of (\ref{eqn:mapping(t,s)-r}). Given $\vec{r}$, one may recursively deduce $(t,\vec{s})$ with the map
%Invertibility follows from the fact that $\mathcal{P}$ describes a physical swap-and-correct protocol (property (1) of Definition \ref{def:swap_and_correct_protocol_technical}). 
\begin{align*}
    t &= r_0 \\ 
    s_{\alpha(1)} &= \mathcal{P}_{\alpha(1)}(\vec{s})  r_1  \\ 
    s_{\alpha(2)} &= \mathcal{P}_{\alpha(2)}(\vec{s}) r_2 \\ 
    &\vdots \\ 
    s_{\alpha(N-1)} &= \mathcal{P}_{\alpha(N-1)}(\vec{s})  r_{N-1}, \;\;\;\;\;\;\;\;\; \text{up to a factor of }\pm 1.
\end{align*}
In the above, at each step, the multiplier $\mathcal{P}_{\alpha(k)}(\vec{s})$ may be computed using the values $(s_{\alpha(1)},\dots,s_{\alpha(k-1)})$ which have been found in the previous steps. This is due to the physicality property of the swap-and-correct protocol (Property A in Definition \ref{def:swap_and_correct_protocol_technical}). The property enforces that corrections at a given node must only depend on the BSM outcomes that have previously been obtained.
Recalling the identity (\ref{eqn:condition_from_correctness}) that followed from the correctness condition, we have 
\begin{align}
    t\mathcal{P}_{0}(\vec{s}) \mathcal{P}_{N}(\vec{s}) = \pm \prod_{j=0}^{N-1} r_j.
    \label{eqn:product_operator_node_N}
\end{align}
Given the bijective map $(t,\vec{s}) \leftrightarrow \vec{r}$, we may therefore combine (\ref{eqn:T.Lambda_linearity}) with (\ref{eqn:T.Lambda_P,S_breakdown}) and (\ref{eqn:product_operator_node_N}) to write
\begin{equation}
    \mathcal{B}_{0,N}\circ \Lambda_{\mathcal{P}} = \frac{1}{4}\sum_{\vec{r}\in A^{N-1}} \left[ \bigcirc_{k=1}^{N-1} (M_{I})_k \circ (r_0)_0  \circ \bigcirc_{j=1}^{N-1} (r_j)_j  \circ \left(\prod_{i=0}^{N-1}r_i \right)_N \right].
    \label{eqn:unconditional_swapping_in_r}
\end{equation}
We now claim that the above is equal to $ \Lambda_{\mathrm{seq},\vec{I}} \circ \bigcirc_{k=1}^N \mathcal{B}_{k-1,k}$. to see this, we perform another change of variables to $\vec{u} = (u_0,\dots,u_{N-1})\in A^{N-1}$, given by 
\begin{align}
    u_0 = r_0, \;\;\; u_k = \prod_{j=0}^k r_j,\;\;\; k = 1, \dots, N-1
\end{align}
The above map has inverse 
\begin{align}
    r_0 = u_0, \;\;\; r_k = u_{k-1} u_k, \;\;\; \text{up to a factor of }\pm 1,
\end{align}
and is therefore bijective.
Performing this change of variable on (\ref{eqn:unconditional_swapping_in_r}), we obtain
\begin{align}
        \mathcal{B}_{0,N}\circ \Lambda_{\mathcal{P}} &= \frac{1}{4}\sum_{\vec{u}\in A^{N-1}} \left[ \bigcirc_{k=1}^{N-1} (M_{I})_k \circ (u_0)_0  \circ \bigcirc_{j=1}^{N-1} (u_{j-1} u_j)_j  \circ \left(u_{N-1} \right)_N \right] \\ &= \frac{1}{4}\sum_{\vec{u}\in A^{N-1}} \left[ \bigcirc_{k=1}^{N-1} (M_{I})_k \circ \bigcirc_{j=0}^{N-1} \big( (u_j)_j \circ (u_j)_{j+1} \big) \right].
\end{align}
Recalling (\ref{eqn:bell_diag_twirl_s&c}), this may be rewritten as 
\begin{align}
    \mathcal{B}_{0,N}\circ \Lambda_{\mathcal{P}} &= \frac{4^N}{4} \bigcirc_{k=1}^{N-1} (M_{I})_k \circ \bigcirc_{j=1}^{N} \mathcal{B}_{j-1,j} \\ &= 4^{N-1} \Lambda_{\mathrm{seq},\vec{I}} \circ \bigcirc_{j=1}^{N} \mathcal{B}_{j-1,j}
    \label{eqn:B.Lambda_P,s=Lambda_seq.B}
\end{align}
where have have recalled that $\Lambda_{\mathrm{seq},\vec{I}} = \bigcirc_{k=1}^{N-1} (M_{I})_k$. In particular, from Lemma (\ref{lem:s&c_bell_diag_P,s_indep}) we recall that, for any syndrome $\vec s$, we have $$\Lambda_{\mathcal{P},\vec{s}} \circ \bigcirc_{k=1}^N \mathcal{B}_{k-1,1} = \Lambda_{\mathrm{seq},\vec{I}}  \circ \bigcirc_{k=1}^N \mathcal{B}_{k-1,1}.$$ Then, noting that $|A^{N-1}|=4^{N-1}$, from (\ref{eqn:B.Lambda_P,s=Lambda_seq.B}) we have
\begin{align}
    \mathcal{B}_{0,N}\circ \Lambda_{\mathcal{P}} &= \sum_{\vec s\in A^{N-1}} \Lambda_{\mathcal{P},\vec{s}} \circ \bigcirc_{k=1}^N \mathcal{B}_{k-1,1} \\ &= \Lambda_{\mathcal{P}} \circ \bigcirc_{k=1}^N \mathcal{B}_{k-1,1},
\end{align}
where we have used the identity (\ref{eqn:unconditional_s&c}) for $\Lambda_{\mathcal{P}}$.
\end{proof}

\begin{proof}[Proof of Lemma \ref{lem:bound_e2e_fidelity_BD}]
We write the eigenvalues of $\mathcal{B}(\rho_1)$ and $\mathcal{B}(\rho_2)$ as 
\begin{align*}
    \mathcal{B}(\rho_1) \equiv  (F_1, (1-F_1)\vec{u})^T, \\ 
    \mathcal{B}(\rho_2) \equiv (F_2, (1-F_2)\vec{v})^T,
    \label{eqn:form_bell_diag}
\end{align*}
where $\vec{u}$  and $\vec{v}$ are length-3 vectors such that $u_1+u_2+u_3 = v_1+v_2+v_3 = 1$ and $u_i,v_i \geq 0$. By Lemma \ref{lem:swap_bell_diagonal}, The fidelity of the end-to-end state is
\begin{equation*}
    F' = F_1 F_2 + (1-F_1)(1-F_2)(\vec{u}\cdot \vec{v}),
\end{equation*}
where $\vec{u}\cdot \vec{v} = u_1 v_1 + u_2 v_2 + u_3 v_3$. The end-to-end fidelity $F'$ satisfies
\begin{equation}
    F_1 F_2 \leq F' \leq F_1 F_2 + (1-F_1) (1-F_2).
    \label{eqn:bounds_fidelity_BD_N=2}
\end{equation}
For the lower bound, we have used $ \vec{u}\cdot \vec{v} \geq 0$. Note that this is saturated whenever $\vec{u}$ and $\vec{v}$ are orthogonal. For the upper bound, we have used 
\begin{equation*}
    \vec{u}\cdot \vec{v} \leq (u_1 + u_2 + u_3)(v_1 + v_2 + v_3) = 1,
\end{equation*}
which is saturated exactly when $\vec{u} = \vec{v}$ and $\vec{u}$ is one of $(1,0,0)^T$, $(0,1,0)^T$, or $(0,0,1)^T$. In particular, this is when the initial states $\mathcal(B)(\rho_k)$ are of rank two and have the same non-zero eigenvectors.

To show the bounds for a length-$N$ repeater chain, we apply the bounds (\ref{eqn:bounds_fidelity_BD_N=2}) inductively.  For the upper bound, we notice that 
\begin{equation}
    F_1 F_2 + (1-F_1) (1-F_2) = \frac{1}{2}(2F_1-1)(2F_2-1) + \frac{1}{2},
    \label{eqn:rank_two_fidelity_decay}
\end{equation}
and applying the upper and lower bounds from (\ref{eqn:bounds_fidelity_BD_N=2}) inductively gives
\begin{equation}
    \prod_{k=1}^N F_k \leq F' \leq \frac{1}{2}\prod_{k=1}^N(2F_k-1)+ \frac{1}{2}.
\end{equation}
By Lemma \ref{lem:swap_bell_diagonal}, the end-to-end state is the same for both postselected and non-postselected swapping.
\end{proof}

%BELOW: ORIGINAL PROOF WHEN WE CONSIDERED SWAPPING TWO WERNER STATES
\begin{comment}
\begin{proof}
    By Proposition \ref{prop:swap_bell_diagonal}, we see that the outcome state may be written as
   \begin{equation}
    \left( \begin{array}{c}
    \mu_0   \\
    \mu_1   \\
    \mu_2   \\
    \mu_3   \\
\end{array} \right)  = 
\left( \begin{array}{c}
    F^2 + (1-F)^2(v_1^2 + v_2^2+v_3^2) \\
    2F(1-F) v_1 + 2(1-F)^2 v_2 v_3   \\
    2F(1-F) v_2 + 2(1-F)^2 v_3 v_1  \\
    2F(1-F) v_3 + 2(1-F)^2 v_1 v_2  
\end{array} \right).
    \end{equation}
    of which the first component is the fidelity. Noting that
    \begin{equation*}
        v_1^2 + v_2^2+v_3^2 \leq (v_1 + v_2 +v_3)^2 =1,
    \end{equation*}
    this is maximised exactly when $v_i = 1$ for one of the components, and the other two are equal to zero. Moreover, this is minimised when $v_1=v_2 = v_3 = 1/3$, i.e. when the input states are Werner. \bd{more explanation?}
\end{proof}
\end{comment}

\begin{proof}[Proof of Theorem \ref{thm:werner_approximation_accuracy}]
% We firstly note that the trace norm between two Werner states with fidelities $F$ and $G$ is 
% \begin{equation}
% \frac{1}{2}\lVert (F-G)\ketbra{\Psi_{00}} + \frac{G-F}{3} \left(I_4 - \ketbra{\Psi_{00}} \right) \rVert_1  = |F-G|.
% \end{equation}
Let $\mathcal{W}_{[N]}$ denote the Bell-diagonal twirling of all initial states $\rho_{\mathrm{in}} = \bigotimes_{k=1}^N \rho_k$ such that
 \begin{equation}
     \mathcal{W}_{[N]} \left(\rho_{\mathrm{in}}\right) = \otimes_{k=1}^N \mathcal{W}(\rho_k).
 \end{equation}
Then,
\begin{align}
    F' &= \bra{\Psi_{00}} \Lambda_{\mathcal{P}}(\rho_{\mathrm{in}})\ket{\Psi_{00}} \\ 
    F'_{\mathcal{W}} &= \bra{\Psi_{00}}\Lambda_{\mathcal{P}}\left( \mathcal{W}_{[N]} \left(\rho_{\mathrm{in}}\right)\right) \ket{\Psi_{00}}
\end{align}
are the true end-to-end fidelity and the end-to-end fidelity with the Werner approximation.
\begin{align}
    F' = \bra{\Psi_{00}} \Lambda_{\mathcal{P}}(\rho_{\mathrm{in}})\ket{\Psi_{00}} &= \bra{\Psi_{00}} \mathcal{B}\left(\Lambda_{\mathcal{P}}(\rho_{\mathrm{in}})\right) \ket{\Psi_{00}} \\ &= \bra{\Psi_{00}} \Lambda_{\mathcal{\mathrm{seq}}}( \mathcal{B}_{[N]}\left(\rho_{\mathrm{in}})\right) \ket{\Psi_{00}},  \label{eqn:F'_sequential_swapping}
\end{align}
where in the first step we have used the fact that the fidelity is invariant under Bell-diagonal twirling, and in the second step we have applied Theorem \ref{thm:swap_and_correct_equivalence}.
Now, since Werner states are Bell-diagonal, we have $\mathcal{B}\left(\mathcal{W}(\rho_k)\right) = \mathcal{W}(\rho_k)$ and $\mathcal{B}_{[N]}\left(\mathcal{W}_{[N]}(\rho_{\mathrm{in}})\right) =\mathcal{W}_{[N]}(\rho_{\mathrm{in}}) $. In particular, we see that
\begin{align}
    \Lambda_{\mathcal{P}}\left( \mathcal{W}_{[N]} \left(\rho_{\mathrm{in}}\right)\right) = \Lambda_{\mathcal{P}}\left( \mathcal{B}_{[N]}\left( \mathcal{W}_{[N]} \left(\rho_{\mathrm{in}}\right)\right)\right)  &= \mathcal{B}\left(\Lambda_{\mathrm{seq}}\left( \mathcal{W}_{[N]} \left(\rho_{\mathrm{in}}\right)\right)\right) \\ &= \Lambda_{\mathrm{seq}}\left( \mathcal{W}_{[N]} \left(\rho_{\mathrm{in}}\right)\right),
\end{align}
where in the first step we have applied Theorem \ref{thm:swap_and_correct_equivalence}, and in the second step we have used the fact that the sequential swapping of Bell-diagonal states results in a Bell-diagonal state (see Lemma \ref{lem:swap_bell_diagonal}). Then,
\begin{equation}
    F'_{\mathcal{W}} = \bra{\Psi_{00}}\Lambda_{\mathcal{P}}\left( \mathcal{W}_{[N]} \left(\rho_{\mathrm{in}}\right)\right) \ket{\Psi_{00}} = \bra{\Psi_{00}} \Lambda_{\mathrm{seq}}\left( \mathcal{W}_{[N]} \left(\rho_{\mathrm{in}}\right)\right) \ket{\Psi_{00}}.
    \label{eqn:F'_W_sequential_swapping}
\end{equation}
In particular, both (\ref{eqn:F'_sequential_swapping}) and (\ref{eqn:F'_W_sequential_swapping}) are the end-to-end fidelity after swapping $N$ Bell-diagonal states. For (\ref{eqn:F'_sequential_swapping}), the fidelity of the $k$th state is $\bra{\Psi_{00}} \mathcal{B}(\rho_k) \ket{\Psi_{00}} = F_k$. For (\ref{eqn:F'_W_sequential_swapping}), the fidelity of the $k$th state is again $\bra{\Psi_{00}} \mathcal{W}(\rho_k) \ket{\Psi_{00}} = F_k$. In particular, the bounds from Lemma \ref{lem:bound_e2e_fidelity_BD} apply to both $F'$ and $F'_{\mathcal{W}}$. Letting $\epsilon_k = 1-F_k$, we then have
\begin{align}
    |F'-F'_{\mathcal{W}}| &\leq \frac{1}{2}\prod_{k=1}^N(2F_k-1)+ \frac{1}{2} - \prod_{k=1}^N F_k \\ &= \frac{1}{2}\prod_{k=1}^N(1-2\epsilon_k)+ \frac{1}{2} - \prod_{k=1}^N (1-\epsilon_k) \\ &\stackrel{\mathrm{i}}{=} \frac{1}{2} \left(1-2\sum_{k}\epsilon_k + 4 \sum_{k,l}\epsilon_k \epsilon_l + \mathcal{O}(N^3 \epsilon^3) \right) + \frac{1}{2}- \left( 1 - \sum_{k}\epsilon_k + \sum_{k,l}\epsilon_k \epsilon_l + \mathcal{O}(N^3 \epsilon^3) \right)
    \\ &\stackrel{\mathrm{ii}}{=} \sum_{k,l}\epsilon_k \epsilon_l + \mathcal{O}(N^3 \epsilon^3)
    \\ &\stackrel{\mathrm{iii}}{\leq} {N \choose 2} \epsilon^2  + \mathcal{O}(N^3 \epsilon^3),
\end{align}
where in (i) we have performed a series expansion to second order in the infidelity, in (ii) we have noticed that the first- and second-order terms cancel, and in (iii) we have used $\epsilon_k = 1-F_k < \epsilon$ for all $k=1,\dots, N$ and the fact that there are ${N \choose 2}$ second-order terms in the sum.
\end{proof}

\newpage

\section{Postselected swapping}
\label{app:conditional_swapping}
\begin{proof}[Proof of Proposition \ref{prop:p_1>p_2}]
Consider $\rho \in S_{p_1}$. By definition, this can be written as 
\begin{equation}
    \rho = p_1 \ketbra{\Psi_{00}} + (1-p_1)  \sigma,
    \label{eqn:noisy_state_p_1}
\end{equation}
where $\sigma$ is a valid density matrix. 
Letting $p_1 = p_2 + (p_1 - p_2)$, we may rewrite $\rho$ as
    \begin{equation*}
        \rho_1 = p_2 \ketbra{\Psi_{00}}  + (1-p_2) \Big[ \frac{p_1 - p_2}{1-p_2} \ketbra{\Psi_{00}} + \frac{1-p_1}{1-p_2} \sigma \Big].
    \end{equation*}
    Since $p_1 > p_2$, the term in the brackets is a valid ensemble of states and therefore also a density matrix. For the converse, consider for example the states 
    \begin{equation*}
        \rho = p_2 \ketbra{\Psi_{00}} + (1-p_2) \ketbra{\Psi_{11}}.
    \end{equation*}
    and suppose that this can be written in the form (\ref{eqn:noisy_state_p_1}). Then, one may show that 
    \begin{equation*}
        \ketbra{\Psi_{11}} = \frac{p_1 - p_2}{1-p_2} \ketbra{\Psi_{00}} + \frac{1-p_1}{1-p_2} \sigma
    \end{equation*}
    computing the fidelity to $\ket{\Psi_{00}}$ then yields
    \begin{equation*}
        0 = \frac{p_1 - p_2}{1-p_2}  + \frac{1-p_1}{1-p_2} \bra{\Psi_{00}} \sigma \ket{\Psi_{00}} 
    \end{equation*}
    which results in a contradiction as the RHS is positive ($\sigma$ is a density matrix).
\end{proof}

\begin{proposition}[Proposition \ref{prop:p_1>p_2}, general version]
    For fixed $p_1$, $p_2$, $F_1$, $F_2$, let 
    \begin{align*}
        F'_{ij,\max} &\coloneqq \max \left\{ F'_{ij}(\rho_1, \rho_2) \text{ s.t. } \rho_k \in S_{p_k,F_k}\right\} \\ F'_{ij,\min} &\coloneqq \min \left\{ F'_{ij}(\rho_1, \rho_2) \text{ s.t. } \rho_k \in S_{p_k,F_k}\right\},
    \end{align*}
    where $F'_{ij}$ is given in Lemma \ref{lem:formulae_swapping_noisy_states_id}.
    Then, the above quantities are independent of $(i,j)$, or alternatively
    \begin{align*}
        F'_{ij,\max} &=  F'_{00,\max} \equiv F'_{\max} \\
        F'_{ij,\min} &=  F'_{00,\min} \equiv F'_{\min}.
    \end{align*}
    \label{prop:mmt_outcome_independent_non_id}
\end{proposition}

\begin{proof}[Proof of Proposition \ref{prop:mmt_outcome_independent_non_id}]
    Consider $\rho_k \in S_{p_k,F_k}$. We now show that for any $(i,j)$ there is $\omega_k \in S_{p_k,F_k}$ such that 
    \begin{align}
        F'_{ij}(\rho_1,\rho_2) = F'_{00}(\omega_1,\omega_2). 
        \label{eqn:fidelity_outcomes_rho_sigma}
    \end{align}
    We claim that this is the case for $\omega_1 =\rho_1$ and $\omega_2 =(Z^j X^i  \otimes  Z^j X^i )\rho_2 ( X^i Z^j \otimes X^i Z^j)$. Firstly, it is clear that $\omega_1 \in S_{p_1,F_1}$. Also,
    \begin{align}
        \omega_2 &= ( Z^j X^i  \otimes Z^j X^i  )\rho_2 (  X^i  Z^j \otimes X^i Z^j) \nonumber \\
        &= ( Z^j X^i  \otimes Z^j X^i  )\left( p_2 \ketbra{\Psi_{00}} + (1-p_2)\sigma_2 \right) (  X^i  Z^j \otimes X^i Z^j) \nonumber
        \\ &= p_2 \ketbra{\Psi_{00}} + (1-p_2) \eta_2 \label{eqn:sigma_2_noisy_comp}
    \end{align}
    where 
    \begin{equation*}
       \eta_2  = (X^i Z^j \otimes X^i Z^j )\sigma_2 ( Z^j X^i \otimes Z^j X^i)
    \end{equation*}
    is the noisy component of $\omega_2$. In (\ref{eqn:sigma_2_noisy_comp}), we have used the flip-flop trick, which means that the Paulis applied to both registers have no effect on the pure $\ket{\Psi_{00}}$ component. One may use the same trick to show that 
    \begin{align}
        \bra{\Psi_{00}} \omega_2 \ket{\Psi_{00}} &= \bra{\Psi_{00}} (Z^j X^i  \otimes  Z^j X^i )\rho_2 ( X^i Z^j \otimes X^i Z^j) \ket{\Psi_{00}} \nonumber \\ &=  \bra{\Psi_{00}} \rho_2 \ket{\Psi_{00}} = F_2,
        \label{eqn:sigma_2_fidelity}
    \end{align}
    and that therefore by (\ref{eqn:sigma_2_noisy_comp}) and (\ref{eqn:sigma_2_fidelity}), $\omega_2 \in S_{p_2,F_2}$. We now show (\ref{eqn:fidelity_outcomes_rho_sigma}). Recalling Definition \ref{def:post_swap_fid_and_prob}, we have
    \begin{align}
        p'_{00}(\omega_1\otimes \omega_2)  &=\Tr \left[ \ketbra{\Psi_{00}}_{A_1 A_2} \rho_1 \otimes \omega_2 \right] \\ &= \Tr \left[ \ketbra{\Psi_{00}}_{A_1 A_2} \omega_1 \otimes (Z^j X^i)_{A_2} (Z^j X^i)_{B_2}\rho_2 ( X^i Z^j)_{A_2} (X^i Z^j)_{B_2}  \right] \\ &= \Tr \left[ \ketbra{\Psi_{ij}}_{A_1 A_2} \rho_1 \otimes \rho_2   \right] \\ &= p'_{ij}(\rho_1\otimes \rho_2).
    \end{align}
    One may similarly show that 
    \begin{equation}
       \Tr \left[ \ketbra{\Psi_{00}}_{B_1 B_2} \ketbra{\Psi_{00}}_{A_1 A_2} \omega_1 \otimes \omega_2 \right] = \Tr \left[ \ketbra{\Psi_{ij}}_{B_1 B_2} \ketbra{\Psi_{ij}}_{A_1 A_2} \rho_1 \otimes \rho_2 \right],
    \end{equation}
    from which we see that 
    \begin{align}
        F'_{00}(\omega_1 \otimes \omega_2) &= \frac{1}{p'_{00}(\omega_1 \otimes \omega_2)} \Tr \left[ \ketbra{\Psi_{00}}_{B_1 B_2} \ketbra{\Psi_{00}}_{A_1 A_2} \omega_1 \otimes \omega_2 \right] \\ &= \frac{1}{p'_{ij}(\rho_1 \otimes \rho_2)} \Tr \left[ \ketbra{\Psi_{ij}}_{B_1 B_2} \ketbra{\Psi_{ij}}_{A_1 A_2} \rho_1 \otimes \rho_2 \right] \\ &= F'_{ij}(\rho_1 \otimes \rho_2).
    \end{align}
\end{proof}

We now derive the formula for the end-to-end fidelity as contained in Lemma \ref{lem:formulae_swapping_noisy_states_id}. In the main text, the result is stated for the swapping of identical states for simplicity. Here, for the completeness we state and prove the same Lemma for non-identical states.
\begin{lemma}[Lemma \ref{lem:formulae_swapping_noisy_states_id}, general version]
Consider performing a postselected swap on a pair of two-qubit states $\rho_1 \otimes \rho_2$ such that $\rho_k \in S_{p_k,F_k}$ for $k=1,2$. Let $F'_{ij} \equiv F'_{ij}(\rho_1\otimes \rho_2)$ denote the end-to-end fidelity after measuring outcome $ij$ in the BSM, and $p'_{ij}$ the probability of measuring that outcome (Definition \ref{def:post_swap_fid_and_prob}). Then, 
\begin{equation}
    p'_{ij} = \frac{p_1 + p_2 - p_1 p_2}{4} +  (1-p_1)(1-p_2) \tilde{p}'_{ij}
\label{eqn:formula_swap_prob_noisy_states_non_id}
\end{equation}
and 
\begin{equation}
    F'_{ij} = \frac{ p_1 F_2 + p_2 F_1 - p_1 p_2 + 4(1-p_1)(1-p_2)\tilde{p}'_{ij}\tilde{F}'_{ij} }{p_1 + p_2 - p_1 p_2 + 4(1-p_1)(1-p_2)\tilde{p}'_{ij}},
    \label{eqn:formula_post_swap_fidelity_noisy_states_non_id}
\end{equation}
where $\tilde{F}'_{ij} \coloneqq F_{ij}'(\sigma_1\! \otimes \! \sigma_2)$ and $ \tilde{p}'_{ij} \coloneqq p_{ij}'(\sigma_1\! \otimes \! \sigma_2)$ are the corresponding swap statistics of the noisy components.
\label{lem:formulae_swapping_noisy_states_non_id}
\end{lemma}

\begin{proof}[Proof of Lemma \ref{lem:formulae_swapping_noisy_states_non_id}]
We start with two states of the form
\begin{align*}
    \rho_1 = p_1 \ketbra{\Psi_{00}} + (1 - p_1)\sigma_1 \\ 
    \rho_2 = p_2 \ketbra{\Psi_{00}} + (1 - p_2)\sigma_2,
\end{align*}
where $F_1  = \bra{\Psi_{00}} \rho_1 \ket{\Psi_{00}}$ and $F_2 = \bra{\Psi_{00}} \rho_2 \ket{\Psi_{00}}$. We carry out the usual swapping protocol: we start with the state $\rho_{1} \otimes \rho_{2}$, which may be expanded as 
\begin{multline}
    \rho_{1} \otimes \rho_{2} = p_1 p_2 \cdot \ket{\Psi_{00}}\bra{\Psi_{00}} \otimes \ket{\Psi_{00}}\bra{\Psi_{00}} + p_1(1-p_2)\cdot \ket{\Psi_{00}}\bra{\Psi_{00}} \otimes \sigma_2 \\ + (1-p_1)p_2 \cdot \sigma_1 \otimes \ket{\Psi_{00}}\bra{\Psi_{00}} + (1-p_1)(1-p_2) \cdot \sigma_1 \otimes \sigma_2.
\end{multline}
%continue from here
Supposing that the middle station measures the BSM outcome $ij$, the output state is
\begin{equation}
   \rho'_{ij} = \frac{1}{p'_{ij}}\bra{\Psi_{ij}} \rho_{1} \otimes \rho_{2} \ket{\Psi_{ij}}_{A_1 A_2},
\end{equation}
where $p'_{ij}$ is the probability of obtaining the BSM outcome $ij$, given by
\begin{equation}
    p'_{ij} = \text{Tr}_{B_1 B_2} \left[ \bra{\Psi_{ij}} \rho_1 \otimes \rho_2 \ket{\Psi_{ij}}_{A_1 A_2} \right].
    \label{eqn:p_ij_trace}
\end{equation}
We now obtain expressions for $p'_{ij}$ and $\rho'_{ij}$ in terms of $p_i$, $F_i$ and $\sigma_i$. From (\ref{eqn:bell_basis}), we have 
\begin{multline}
    \bra{\Psi_{ij}} \rho_1\otimes \rho_2 \ket{\Psi_{ij}}_{A_1 A_2} = \frac{p_1 p_2}{4} \ket{\Psi_{ij}}\bra{\Psi_{ij}} +\frac{p_1(1-p_2)}{4} (X^i Z^j \otimes I_2) \sigma_2 (Z^j X^i \otimes I_2) \\ +\frac{p_2(1-p_1)}{4} (I_2 \otimes X^i Z^j) \sigma_1 (I_2 \otimes Z^j X^i) + (1-p_1)(1-p_2) \bra{\Psi_{ij}} \sigma_1 \otimes \sigma_2 \ket{\Psi_{ij}}_{A_1 A_2},
    \label{eqn:rho_1_times_rho_2_expansion}
\end{multline}
where the first three terms correspond to perfect teleportation (without Pauli corrections). From (\ref{eqn:p_ij_trace}), we now calculate $p'_{ij}$ by taking the trace of the above. Notice that the first three terms are proportional to valid density matrices. Therefore,
\begin{equation}
        p'_{ij} = \frac{p_1 p_2}{4} + \frac{p_1(1-p_2)}{4} + \frac{p_2(1-p_1)}{4} + (1-p_1)(1-p_2) \tilde{p}'_{ij},
        \label{eqn:p_ij_expansion}
\end{equation}
where $\tilde{p}'_{ij} = p'_{ij}(\sigma_1 \otimes \sigma_2)$. Then,
\begin{comment}
Now, applying Pauli corrections to (\ref{eqn:rho_1_times_rho_2_expansion}) and taking the fidelity to $\ket{\Psi_{00}}$ is equivalent to taking the fidelity to $\ket{\Psi_{ij}}$, due to the following reasoning:
\begin{equation*}
        \bra{\Psi_{00}}(\id_2 \otimes Z^j X^i) \rho'_{ij} (\id_2 \otimes X^i Z^j) \ket{\Psi_{00}}_{A_1 B_1}  = \bra{\Psi_{ij}}\rho'_{ij}  \ket{\Psi_{ij}}_{A_1 B_1}.
\end{equation*}
We may therefore compute the post-swap fidelity (to $\ket{\Psi_{00}}$) after the Pauli corrections as follows. Recalling the expansion (\ref{eqn:rho_1_times_rho_2_expansion}), we have
\end{comment}
\begin{align*}
   F'_{ij} &= \bra{\Psi_{ij}}\rho'_{ij}  \ket{\Psi_{ij}}_{B_1 B_2} \\  &= \frac{1}{p_{ij}} \left(\frac{p_1 p_2}{4} + \frac{p_1 (1-p_2)}{4} \tilde{F}_2 + \frac{p_2 (1-p_1)}{4} \tilde{F}_1  + (1-p_1)(1-p_2) \tilde{F}'_{ij} \tilde{p}'_{ij} \right)
\end{align*}
where $\tilde{F}_k = \bra{\Psi_{00}} \sigma_k \ket{\Psi_{00}}$ is the fidelity of the noisy components, and $\tilde{F}'_{ij} = F'_{ij}(\sigma_1 \otimes \sigma_2)$.
%\begin{align*}
%    \tilde{F}_i &= \bra{\Psi_{00}} \sigma_i \ket{\Psi_{00}}, \\ \tilde{F}'_{ij} \tilde{p}'_{ij} = \bra{\Psi_{ij}} & \bra{\Psi_{ij}} \sigma_1 \otimes \sigma_2 \ket{\Psi_{ij}}_{A_2 B_2} \ket{\Psi_{ij}}_{A_1 B_1}.
%\end{align*}
Recalling the fidelity constraint on our initial states $\rho_k$, we may rewrite
\begin{equation*}
    F_k = p_k + (1-p_k)\tilde{F}_k
\end{equation*}
and simplify our formula for the end-to-end fidelity to 
\begin{equation}
    F'_{ij} = \frac{1}{p'_{ij}}\left(\frac{p_1 p_2}{4} + \frac{p_1 (F_2-p_2)}{4}  + \frac{p_2 (F_1-p_1)}{4}   + (1-p_1)(1-p_2) \tilde{F}'_{ij} \tilde{p}'_{ij}\right).
\end{equation}
%Finally, substituting in $p_{ij}$ from (\ref{eqn:p_ij_expansion}) gives us 
%\begin{equation}
%    F'_{ij} = \frac{p_1 p_2 + p_1 (F_2-p_2) + p_2 (F_1-p_1) + 4(1-p_1)(1-p_2) \tilde{F}'_{ij} \tilde{p}'_{ij}}{        p_1 p_2 + p_1(1-p_2) + p_2(1-p_1) + 4(1-p_1)(1-p_2) \tilde{p}'_{ij}}.
%\end{equation}
\end{proof}

\begin{theorem}[Theorem \ref{thm:tight_upper_bound}, general version]
    Consider performing a postselected swap on a pair of two-qubit states $\rho_1 \otimes \rho_2$ such that $\rho_k \in S_{p_k,F_k}$ for $k=1,2$. Then, 
    \begin{equation}
        F'_{00}(\rho_1\otimes\rho_2) \leq 1 - p_1 (1-F_2) - p_2(1-F_1).
        \label{eqn:upper_bound_non-id}
    \end{equation}
In particular, for the case $p_1 = p_2$, $F_1 = F_2$, the above bound is tight and therefore
    \begin{equation}
        F'_{\max}(p,F) = 1 - 2p(1-F),
        % \label{eqn:upper_bound_non-id}
    \end{equation}
where $F'_{\max}(p,F)$ is defined in (\ref{eqn:def_F'max}).
%i.e. this state provides the highest post-swap fidelity for all states in $S_{p,F}$.
\label{thm:tight_upper_bound_non_id}
\end{theorem}

\begin{proof}[Proof of Theorem \ref{thm:tight_upper_bound_non_id}]
    From Lemma \ref{lem:formulae_swapping_noisy_states_non_id}, after measuring  we obtain an outcome fidelity of 
    \begin{equation}
        F'_{00} = \frac{p_1 F_2 + p_2 F_1 - p_1 p_2 + 4(1-p_1)(1-p_2)\tilde{F}_{00}'\tilde{p}_{00}'}{p_1 + p_2 - p_1 p_2 + 4(1-p_1)(1-p_2) \tilde{p}_{00}'}.
        \label{eqn:formula_F'_noisy_state_appendix}
    \end{equation}
    Now, since $\tilde{F}_{00} \leq 1$, we have 
    \begin{equation}
         F'_{00} \leq \frac{p_1 F_2 + p_2 F_1 - p_1 p_2 + 4(1-p_1)(1-p_2)\tilde{F}_{00}'\tilde{p}_{00}'}{p_1 + p_2 - p_1 p_2 + 4(1-p_1)(1-p_2) \tilde{F}_{00}'\tilde{p}_{00}'}.
        \label{eqn:bound_1}
    \end{equation}
    We notice that the RHS of (\ref{eqn:bound_1}) is of the form 
    \begin{equation}
        \frac{a+x}{b+x} = 1 - \frac{b-a}{b+x},
        \label{eqn:rational_fn_x}
    \end{equation}
    for $b-a = p_1(1-F_2) + p_2 (1-F_1)\geq 0$. Then, (\ref{eqn:rational_fn_x}) is a non-decreasing function of $x = \tilde{F}_{00}'\tilde{p}_{00}'$. 
    We recall  that
    \begin{equation}
       \tilde{F}_{00}'\tilde{p}_{00}' = \bra{\phi} \sigma_1 \otimes \sigma_2 \ket{\phi},
    \end{equation}
    where
    \begin{align}
        \ket{\phi} &=  \ket{\Psi_{00}}_{A_1 A_2} \otimes \ket{\Psi_{00}}_{B_1 B_2}.
    \end{align}
Since the state $\ket{\phi}$ is a product of two maximally entangled qubit states, this is a maximally entangled state between two registers of dimension $d = 4$, shared between the four-dimensional registers $B_1 A_1$ and $A_2 B_2$. Therefore, since $\sigma_1 \otimes \sigma_2$ is unentangled with repect to these registers, from \cite{horodecki1999general}, the fidelity to $\ket{\phi}$ is bounded as
\begin{equation}
    \tilde{F}_{00}'\tilde{p}_{00}' =  \bra{\phi} \sigma_1 \otimes \sigma_2 \ket{\phi} \leq \frac{1}{d} = \frac{1}{4}.
\end{equation}
Combining the above results, it follows that 
    \begin{align}
         F'_{00} &\leq \frac{p_1 F_2 + p_2 F_1 - p_1 p_2 + 4(1-p_1)(1-p_2)\cdot \frac{1}{4}}{p_1 + p_2 - p_1 p_2 + 4(1-p_1)(1-p_2)\cdot \frac{1}{4}}\nonumber  \\ &= 1 - p_1 (1-F_2) - p_2 (1-F_1).
        % \label{eqn:bound_1}
    \end{align}
For the case $p_1 = p_2 = p$, $F_1 = F_2 = F$, we now show that the above bound is tight. Letting $\rho_k \in S_{p,F}$ such that 
\begin{align}
    \rho_1 = \rho_2 = p\ketbra{\Psi_{00}} + (1-p)\ketbra{\psi},
    \label{eqn:optimal_state}
\end{align}
with 
\begin{equation}
    \ket{\psi} = \sqrt{\tilde{F}} \ket{\Psi_{00}} + \sqrt{1-\tilde{F}} \ket{\Psi_{11}}
    \label{eqn:psi_best_state}
\end{equation}
where $\tilde{F} = (F-p)/(1-p)$, we now compute the values of $\tilde{F}'_{00}$ and $\tilde{p}'_{00}$ for $\sigma_1 = \sigma_2 = \ket{\psi}$.

We firstly expand the initial state as
\begin{equation}
    \ket{\psi}^{\otimes 2} = \tilde{F} \ket{\Psi_{00}}^{\otimes 2} +(1- \tilde{F} ) \ket{\Psi_{11}}^{\otimes 2} + \sqrt{\tilde{F}(1-\tilde{F})} \ket{\Psi_{00}} \otimes \ket{\Psi_{11}} + \sqrt{\tilde{F}(1-\tilde{F})} \ket{\Psi_{11}} \otimes \ket{\Psi_{00}}.
    \label{eqn:projection_optimal_state}
\end{equation}
We compute the action of a Bell-state measurement on each component as the following. Recalling that $\ket{\Psi_{11}}_{AB} = I \otimes X Z \ket{\Psi_{00}}_{AB}$ and that 
\begin{equation}
    \bra{\Psi_{00}}_{A_1 A_2}\Big[ \ket{\Psi_{00}} \otimes \ket{\Psi_{00}} \Big] = \frac{1}{2} \ket{\Psi_{00}}_{B_1 B_2},
\end{equation}
we have \Guus{I think there should be a minus in the below equation, XZ picks up a minus when pingponging; equivalently, $\ket{\Psi_{00}}$ is stabilized by $-YY$, not by $YY$.} \bd{I don't think this matters because $-Y\otimes Y = - i XZ \otimes i XZ =  XZ \otimes XZ$. So $\ket{\Psi_{00}}$ is also stabilised by $XZ$. The code did not agree because of the definition of the Bell basis I think: if we swap $\ket{\psi}^{\otimes 2}$ in the code, different BSM outcomes give unit end-to-end fidelity. However, if in the code we instead swap $\ket{\psi}\otimes (\sqrt{\tilde{F}} \ket{\Psi_{00}} - \sqrt{1-\tilde{F}} \ket{\Psi_{11}}) $ (essentially defining the Bell basis such that the Pauli operators are applied to the qubits in the middle) then it gives the same outcome.}
\begin{align}
    \bra{\Psi_{00}}_{A_1 A_2}  \Big[ \ket{\Psi_{11}} \otimes \ket{\Psi_{11}} \Big]  &= (XZ)_{B_1} (XZ)_{B_2} \bra{\Psi_{00}}_{A_1 A_2}  \Big[ \ket{\Psi_{00}} \otimes \ket{\Psi_{00}} \Big] \\ &= \frac{1}{2} (XZ)_{B_1} (XZ)_{B_2} \ket{\Psi_{00}}_{B_1 B_2} \\ &= \frac{1}{2} \ket{\Psi_{00}}_{B_1 B_2},
\end{align}
where we have used the flip-flop trick (\ref{eqn:flip-flop}) to move the Pauli gates from one register to the other. Using the same method, it may be shown that
\begin{align}
    \bra{\Psi_{00}}_{A_1 A_2}  \Big[ \ket{\Psi_{00}} \otimes \ket{\Psi_{11}} \Big] &= (XZ)_{B_2}  \frac{1}{2} \ket{\Psi_{00}}_{B_1 B_2} \nonumber \\ &= \frac{1}{2} \ket{\Psi_{11}}_{B_1 B_2} 
\end{align}
and 
\begin{align}
    \bra{\Psi_{00}}_{A_1 A_2}  \Big[ \ket{\Psi_{11}} \otimes \ket{\Psi_{00}} \Big] &= (XZ)_{B_1} \cdot \frac{1}{2} \ket{\Psi_{00}}_{B_1 B_2} \\ &= (ZX)_{B_2}  \frac{1}{2} \ket{\Psi_{00}}_{B_1 B_2} \\ &= - \frac{1}{2} \ket{\Psi_{11}}_{B_1 B_2}.
\end{align}
Then, recalling (\ref{eqn:projection_optimal_state}), we see that
\begin{align}
    \bra{\Psi_{00}}_{A_1 A_2} \left[ \ket{\psi}^{\otimes 2}\right] &= \left(\frac{\tilde{F}}{4} + \frac{1- \tilde{F}}{4}\right)\ket{\Psi_{00}}_{B_1 B_2} + \sqrt{\tilde{F}(1-\tilde{F})}\left( \ket{\Psi_{11}} - \ket{\Psi_{11}} \right) \\ &= \frac{1}{4}\ket{\Psi_{00}}_{B_1 B_2}.
\end{align}
We therefore see that 
\begin{equation}
    \tilde{p}'_{00} = \frac{1}{4}, \;\;\;\;\;\; \tilde{F}'_{00} = 1,
\end{equation}
and substituting these values into (\ref{eqn:formula_F'_noisy_state_appendix}), we see that the bound is saturated.
\Guus{Is this the right equation? It's in the future!}
\end{proof}
\begin{proposition}[Proposition \ref{prop:analytical_lower_bound}, general version]
Consider performing a postselected swap on a pair of two-qubit states $\rho_1 \otimes \rho_2$ such that $\rho_k \in S_{p_k,F_k}$ for $k=1,2$. Then, 
    \begin{equation}
        F'_{00}(\rho_1 \otimes \rho_2) \geq  \frac{p_1 F_2 + p_2 F_1 - p_1 p_2}{1+(1-p_1)(1-p_2) }.
    \end{equation}
\label{prop:analytical_lower_bound_non_id}
\end{proposition}
\begin{proof}[Proof of Proposition \ref{prop:analytical_lower_bound_non_id}]
    From Lemma \ref{lem:formulae_swapping_noisy_states_non_id}, after measuring  we obtain an outcome fidelity of 
    \begin{equation}
        F'_{00} = \frac{p_1 F_2 + p_2 F_1 - p_1 p_2 + 4(1-p_1)(1-p_2)\tilde{F}_{00}'\tilde{p}_{00}'}{p_1 + p_2 - p_1 p_2 + 4(1-p_1)(1-p_2) \tilde{p}_{00}'}.
    \end{equation}
Since $\tilde{F}_{00}'\tilde{p}_{00}'\geq 0$, we have
    \begin{equation}
        F'_{00} \geq \frac{p_1 F_2 + p_2 F_1 - p_1 p_2}{p_1 + p_2 - p_1 p_2 + 4(1-p_1)(1-p_2) \tilde{p}_{00}'}.
        %\label{eqn:formula_F'_noisy_state_appendix}
    \end{equation}
Now, recalling that 
\begin{align}
    \tilde{p}'_{00} = \Tr_{B_1 B_2} \left[\ketbra{\Psi_{00}}_{A_1 A_2} \sigma_1 \otimes \sigma_2  \right],
\end{align}
we see that $\tilde{p}'_{00} \leq 1/2$ as this is the fidelity between a separable state and a maximally entangled state with $d=2$ (see proof of Theorem \ref{thm:tight_upper_bound_non_id}). Therefore, 
    \begin{align}
        F'_{00} &\geq \frac{p_1 F_2 + p_2 F_1 - p_1 p_2}{p_1 + p_2 - p_1 p_2 + 2(1-p_1)(1-p_2) } \\ &=  \frac{p_1 F_2 + p_2 F_1 - p_1 p_2}{1+(1-p_1)(1-p_2) }.
    \end{align}
Recalling (\ref{eqn:def_F'min}), for the case $p_1 = p_2 = p$ and $F_1 = F_2 = F$, the result 
\begin{equation}
    F'_{\min}(p,F) \geq \frac{2pF - p^2}{2p - p^2 + 2(1-p)^2 }
\end{equation}
follows directly.
\end{proof}
%OLD PROOF WITH IDENTICAL STATES
\begin{comment} 
     When swapping two identical states, from Lemma \ref{lem:formulae_swapping_noisy_states}, after measuring we again obtain an outcome fidelity of 
    \begin{equation}
        F'_{ij} = \frac{2pF - p^2 + 4(1-p)^2\tilde{F}_{ij}\tilde{p}_{ij}}{2p - p^2 + 4(1-p)^2 \tilde{p}_{ij}}.
    \end{equation}
    Then, since $\tilde{F}_{ij}\tilde{p}_{ij} \geq 0$,
    \begin{equation}
        F'_{ij} \geq \frac{2pF - p^2}{2p - p^2 + 4(1-p)^2 \tilde{p}_{ij}}.
        \label{eqn:ana_lower_bound_S1}
    \end{equation}
    Now, from (\ref{eqn:swap_prob_noisy_component}), we have
    \begin{equation*}
        \tilde{p}_{ij} = \bra{\Psi_{ij}} \rho_{\text{noise},A_2} \otimes \rho_{\text{noise},B_2} \ket{\Psi_{ij}}_{A_2 B_2},
    \end{equation*}
    where $\rho_{\text{noise},A_2} = \Tr_{A_1} [\rho]$, and likewise for $ \rho_{\text{noise},B_2}$. Then, $\tilde{p}_{ij}$ is the fidelity of a product (unentangled) state to the maximally entangled state $\ket{\Psi_{ij}}$. Since this time the entangled state is between two systems of dimension $d=2$, by \cite{} this quantity is bounded above by 
    \begin{equation*}
        \tilde{p}_{ij} \leq \frac{1}{2}.
    \end{equation*}
    Combining this with (\ref{eqn:ana_lower_bound_S1}), it follows that 
    \begin{equation}
        F'_{ij} \geq \frac{2pF - p^2}{2p - p^2 + 2(1-p)^2 }.
    \end{equation}
    Rearranging the above then yields (\ref{eqn:lower_bound_fidelity}).
\end{comment}

\newpage 

\section{Further details of SDP formulation}
\subsection{SDP symmetrisation}
\label{app:sdp}
% \subsection{Introduction to SDP}
% \todo{Introduce standard form as given by Watrous. Our problem can be written in this form. Strong duality guaranteed.}
% The standard form of a semi-definite program is the minimisation of a linear function of a length-$m$ vector $x$, 
% \begin{mini}|s|
% {}{c^T x \;\;\;\;\;\;\;\;\;\;\;\;\;\;\;\;\;\;\;\;\;\;}
% {}{}
% \addConstraint{G_0 + \sum_{i=1}^m x_i G_i \geq 0,}
% \label{opt:standard_SDP}
% \end{mini}
% where $c\in \mathbb{R}^m$ and $G_0,\dots,G_m \in \mathbb{R}^{n\times n}$ are Hermitian matrices \cite{find_approprite_citation}.

% \subsection{Symmetry under correlated unitaries}
% \label{app:symmetrisation_unitary}
Here, we perform a symmetry reduction of the optimisation problem (\ref{opt:PPT1}), which we restate below for convenience:
\begin{mini}|s|
{\sigma}{ \mathrm{Tr}\big[ \ket{\Psi_{00}}\bra{\Psi_{00}}_{B_1 B_2} \ket{\Psi_{00}}\bra{\Psi_{00}}_{A_1 A_2} \sigma \big]}
{}{}
\addConstraint{ \mathrm{Tr}\big[\ket{\Psi_{00}}\bra{\Psi_{00}}_{A_1 A_2} \sigma \big] = \tilde{\delta}}
\addConstraint{\mathrm{Tr}\big[\ket{\Psi_{00}}\bra{\Psi_{00}}_{A_1 B_1} \sigma \big] = \tilde{F}}
\addConstraint{\mathrm{Tr}\big[\ket{\Psi_{00}}\bra{\Psi_{00}}_{A_2 B_2} \sigma \big] = \tilde{F}}
\addConstraint{\mathrm{Tr}[\sigma]= 1}
\addConstraint{\sigma\geq 0, \;\;\;\sigma^{\Gamma}\geq 0}.\label{opt:PPT1_app}
\end{mini}
In the above, the number of parameters involved in the optimisation is the number required to parameterise a quantum state over four qubits, which is of the order of $ 16^2 = 256$. In the following, we will reduce this number by identifying symmetries of the above optimisation problem, which will then enable us to find the solution more efficiently.

Firstly, we rewrite the above in terms of a rotated target state. In particular, we notice that since the set of PPT states is invariant under the application of local unitaries, it follows that the above is equivalent to the following. Applying the $ZX$ operator to registers $A_1$ and $B_2$, the objective function of the above transforms to
\begin{multline}
    \mathrm{Tr}\big[ \ketbra{\Psi_{00}}_{B_1 B_2} \ketbra{\Psi_{00}}_{A_1 A_2} (ZX)_{B_1} (ZX)_{A_2} \sigma (XZ)_{B_1} (XZ)_{A_2} \big] \\ = \mathrm{Tr}\big[ \ketbra{\Psi_{11}}_{B_1 B_2} \ket{\Psi_{11}}\bra{\Psi_{11}}_{A_1 A_2}  \sigma  \big],
\end{multline}
and the other constraints transform similarly, and so (\ref{opt:PPT1_app}) is equivalent to
\begin{mini}|s|
{\sigma}{ \mathrm{Tr}\big[ \ket{\Psi_{11}}\bra{\Psi_{11}}_{B_1 B_2} \ket{\Psi_{11}}\bra{\Psi_{11}}_{A_1 A_2}  \sigma \big]}
{}{}
\addConstraint{ \mathrm{Tr}\big[\ket{\Psi_{11}}\bra{\Psi_{11}}_{A_1 A_2} \sigma \big] = \tilde{\delta}}
\addConstraint{\mathrm{Tr}\big[\ket{\Psi_{11}}\bra{\Psi_{11}}_{A_1 B_1} \sigma \big] = \tilde{F}}
\addConstraint{\mathrm{Tr}\big[\ket{\Psi_{11}}\bra{\Psi_{11}}_{A_2 B_2} \sigma \big] = \tilde{F}}
\addConstraint{\mathrm{Tr}[\sigma]= 1} \addConstraint{\sigma\geq 0, \;\;\;\sigma^{\Gamma}\geq 0},\label{opt:PPT1_rotation}
\end{mini}
The reason for studying (\ref{opt:PPT1_rotation}) instead of (\ref{opt:PPT1_app}) is due to the symmetry properties of the $\ket{\Psi_{11}}$ state. In particular, for any one-qubit unitary $U$, the state $\ket{\Psi_{11}}$ satisfies
\begin{equation}
    (U\otimes U) \ketbra{\Psi_{11}} (U\otimes U)^{\dag} = \ketbra{\Psi_{11}}.
\end{equation} 
Therefore, under the transformation
%\Guus{I guess this transformation is also a twirl, is it useful to call it that explicitly?}
\begin{equation}
    \sigma \mapsto (U^{\otimes 4}) \sigma (U^{\otimes 4})^{\dag},
    \label{eqn:map_four_unitary}
\end{equation}
it can be seen that the objective function and constraints of (\ref{opt:PPT1_rotation}) are invariant. If $\sigma_{\mathrm{opt}}$ is an optimal solution of (\ref{opt:PPT1_rotation}), then 
\begin{equation}
   \overline{\sigma}_{\mathrm{opt}} = \int (U^{\otimes 4}) \sigma_{\mathrm{opt}} (U^{\otimes 4})^{\dag} \dd U
\end{equation}
is also optimal, where the integration is over the Haar measure. The state $\overline{\sigma}_{\mathrm{opt}}$ is invariant under the map (\ref{eqn:map_four_unitary}). In order to solve (\ref{opt:PPT1_rotation}) it therefore suffices to optimise over the set of operators that are invariant under the symmetry (\ref{eqn:map_four_unitary}). These states are given by %\cite{luca}
\begin{equation}
   \sum_{\tau\in S_4} r_{\tau} M_{\tau} : r_{\tau} \in \mathbb{C},
   \label{eqn:invariant_states}
\end{equation}
where $S_4$ is the symmetric group. In the above, 
\begin{equation}
    M_{\tau} = \sum_{\vec{k}\in \{0,1\}^4} \ket{\tau (\vec{k})} \bra{\vec{k}}
\end{equation}
where $\tau(\vec{k}) = (k_{\tau^{-1}(1)},k_{\tau^{-1}(2)},k_{\tau^{-1}(3)},k_{\tau^{-1}(4)})$. Then, $M_{\tau}$ is the operator that permutes the four registers according to the permutation $\tau$.
The unitary operators $M_{\tau}$ form a representation of the symmetric group $S_4$, and in particular satisfy
\begin{equation}
    M_{\tau_1} M_{\tau_2} = M_{\tau_1 \tau_2}.
\end{equation}
We therefore see from (\ref{eqn:invariant_states}) that in order to parameterise a state that is invariant under the map (\ref{eqn:map_four_unitary}), one requires a maximum of $2\cdot|S_4|=48$ parameters. Now, we make use of the identity 
\begin{equation}
    \ketbra{\Psi_{11}} = \frac{1}{2}\left(I_2 - \mathrm{SWAP}\right),
\end{equation}
where $\mathrm{SWAP}$ is the operation that swaps the two qubits, and we will rewrite the constraints and objective function of (\ref{opt:PPT1_rotation}).
In the following, we will denote a permutation by its decomposition into cycles \cite{sagan2013symmetric}. For example, the permutation that swaps registers $A_1$ and $A_2$ is denoted by $(A_1 A_2)$. Then, given the symmetrised form (\ref{eqn:invariant_states}), the first constraint of (\ref{opt:PPT1_rotation}) becomes
\begin{align}
    \tilde{\delta} &= \Tr\left[\ketbra{\Psi_{11}}_{A_1 A_2} \sigma \right] \\ &= \Tr\left[\frac{1}{2}\left(M_e - M_{(A_1 A_2)}\right) \sum_{\tau \in S_4} r_{\tau} M_{\tau} \right] \\ &= \frac{1}{2} \Tr \left[\sum_{\tau \in S_4} r_{\tau}\left(M_{\tau} -  M_{(A_1 A_2)\tau}\right) \right] \\ &= \frac{1}{2} \sum_{\tau \in S_4} r_{\tau}\left( \Tr \left[M_{\tau} \right]- \Tr \left[M_{(A_1 A_2)\tau} \right]\right) \\ &= v^T r,
\end{align}
where $v$ is a vector indexed by elements of $S_4$, with $v_{\tau}\coloneqq \left( \Tr \left[M_{\tau} \right]- \Tr \left[M_{(A_1 A_2)\tau} \right]\right)/2$.
This is a linear constraint on the vector $r$. We may perform the same procedure for all constraints and the objective function in (\ref{opt:PPT1_rotation}), transforming this into
\begin{mini}|s|
{r\in \mathbb{C}^{|S_4|}}{  u^T r }
{}{}
\addConstraint{ v^T r = \tilde{\delta}}
\addConstraint{w_1^T r  = \tilde{F}, \;\;\; w_2^T r = \tilde{F}}
\addConstraint{x^T r= 1}
\addConstraint{\sum_{\tau \in S_4} r_{\tau}M_{\tau} = \Big(\sum_{\tau \in S_4} r_{\tau}M_{\tau}\Big)^{\dagger} }
\addConstraint{\sum_{\tau \in S_4} r_{\tau}M_{\tau}\geq 0, \;\;\;\sum_{\tau \in S_4} r_{\tau}(M_{\tau})^{\Gamma}\geq 0 }.\label{opt:PPT_symmetrised}
\end{mini}
In the above, $u$, $v$ $w$ and $x$ are all vectors indexed by elements of $S_4$, with  
\begin{align}
u_{\tau}&\coloneqq \frac{1}{4}\left( \Tr \left[M_{\tau} \right]- \Tr \left[M_{(A_1 A_2)\tau} \right] - \Tr \left[M_{(A_1 A_2)\tau} \right] + \Tr \left[M_{(A_1 A_2)(B_1 B_2)\tau} \right]\right) \\
    v_{\tau}&\coloneqq \frac{1}{2}\left( \Tr \left[M_{\tau} \right]- \Tr \left[M_{(A_1 A_2)\tau} \right]\right) \\ (w_k)_{\tau} &\coloneqq \frac{1}{2}\left( \Tr \left[M_{\tau} \right]- \Tr \left[M_{(A_k B_k)\tau} \right]\right) \\ x_{\tau}&\coloneqq \Tr \left[M_{\tau} \right] .
\end{align}
The values $\Tr\left[M_{\tau} \right]$, and therefore the vectors $u$, $v$, $w$ and $x$, may be computed using the following proposition.
\begin{proposition}
    Suppose that $\tau, \tau' \in S_4$ are conjugate, i.e. there is a $ \nu$ such that $\tau' = \nu^{-1} \tau \nu$. Then, 
    \begin{equation*}
        \mathrm{Tr}(M_{\tau'}) = \mathrm{Tr}(M_{\tau}).
    \end{equation*}
    In particular, $\mathrm{Tr}[M_{\tau}]$ is determined by the conjugacy class (cycle type) of $\tau$, of which there are the following five possibilities:
    \begin{align}
        \mathrm{Tr}[M_{\tau}] = \begin{cases}
            16 \:\text{ if cycle type 1+1+1+1} \\ 
            8 \:\text{ if 2+1+1}
            \\ 
            4 \:\text{ if 2+2}
            \\ 
            4 \:\text{ if 3+1}
            \\ 
            2 \:\text{ if 4}.
        \end{cases}
        \label{eqn:conj_classes_trace}
    \end{align}
    \label{lem:trace_permutation_operators}
\end{proposition}
\begin{proof}
Since $\tau' = \nu^{-1} \tau \nu$, we have $M_{\tau'} = M_{\nu^{-1} \tau \nu} = M_{\nu^{-1}} M_{\tau} M_{\nu} = M_{\nu}^{-1} M_{\tau} M_{\nu} $. Then, $$ \Tr[M_{\tau'}] = \Tr[M_{\nu}^{-1} M_{\tau} M_{\nu}] = \Tr[M_{\tau}].$$ In particular, $\Tr[M_{\tau'}] = \Tr[M_{\tau}]$ for all $\tau' \in \mathrm{Cl}(\tau)$, where $$ \mathrm{Cl}(\tau) = \{ \nu^{-1} \tau \nu : \nu \in \mathrm{Sym}(4)\}$$ is the conjugacy class of $\tau'$. Now, for the symmetric group, the conjugacy class is determined by the cycle type \cite{sagan2013symmetric}. The group $\mathrm{Sym}(4)$ has five cycle types. One may then compute the values (\ref{eqn:conj_classes_trace}) by computing $\Tr[M_{\tau}]$ for a given example $\tau$ of each cycle type.
\end{proof}

\subsection{Feasible region of $\delta$}
\label{sec:delta_feasible_region}
Here, we explain how to compute the feasible region for the numerical optimisation over $\delta$. This is used to compute the tightened lower bound. For clarity, we restate here the expression of the lower bound presented in Section \ref{sec:sdp} of the main text,
\begin{equation}
    F'_{\min}(p,F) \geq \min_{\delta} \frac{1}{\delta}\left( \frac{Fp}{2} - \frac{p^2}{4} +(1-p)^2 H_{\mathrm{rel}}^*(p,F,\delta)\right),
    \label{eqn:lower_bound_with_SDP_appendix}
\end{equation}
where $H_{\mathrm{rel}}^*(p,F,\delta)$ is the optimal value of (\ref{opt:PPT1}), which is the lower bound on the post-swap fidelity, with $\delta$ fixed as the probability of obtaining the swap outcome $\ketbra{\Psi_{00}}$. We now show how to find the feasible region $[\delta_{\min},\delta_{\max}]$ within which it suffices to optimise in order to solve (\ref{eqn:lower_bound_with_SDP_appendix}). To find $\delta_{\min}$, we solve the following optimisation problem with SDP,
\begin{mini}|s|
{\sigma \in \mathrm{PPT}}{\mathrm{Tr}\big[\ketbra{\Psi_{00}}_{A_1 A_2} \sigma \big]}
{}{}
\addConstraint{\mathrm{Tr}\big[\ketbra{\Psi_{00}}_{B_1 A_1} \sigma \big] = \tilde{F}(p,F)}
\addConstraint{\mathrm{Tr}\big[\ketbra{\Psi_{00}}_{B_2 A_2} \sigma \big] = \tilde{F}(p,F)}
\addConstraint{\mathrm{Tr}[\sigma]= 1}
\addConstraint{\sigma\geq 0, \;\;\; \sigma^{\Gamma}\geq 0.}
\label{opt:PPT_for_delta}
\end{mini}
Letting $\tilde{\delta}_{\min}$ be the solution to the above, the minimum value of $\delta$ is $$\delta_{\min} = \frac{p^2}{2} - \frac{p^2}{4} + (1-p)^2 \tilde{\delta}_{\min}. $$ To find $\delta_{\max}$, we go through the same procedure, except the objective function of (\ref{opt:PPT_for_delta}) is instead maximised. Moreover, the problem (\ref{opt:PPT_for_delta}) may be symmetrised using the same proceduce as described in the previous subsection, which reduces the number of free parameters. This is how the solutions are computed in \cite{github_repo}.

\newpage

\section{Further analysis of SDP lower bound}
\label{app:more_SDP_lower_bound}
In this appendix, we further investigate the behaviour of the SDP lower bound for $F'_{\min}(p,F)$, which was presented in Section \ref{sec:sdp}. We firstly note that one may formulate an upper bound for $F'_{\max}(p,F)$ with the same method (i.e. performing a maximisation of the objective functions of (\ref{opt:PPT1}) and (\ref{eqn:lower_bound_with_SDP}) instead of a minimisaton). Although this is not necessary, because in Theorem \ref{thm:tight_upper_bound} we have an explicit solution for $F'_{\max}(p,F)$, we computed this solution in order to better understand the range of the end-to-end fidelity after the PPT relaxation. Interestingly, the result of this was always $F'_{\max}(p,F)$: in all cases tested, the corresponding SDP upper bound was tight. More specifically, it has a simple linear form in terms of $p$ and $F$, and the corresponding probability of measuring the outcome ($\delta^*$) takes a constant value of $1/4$.  This matches the example of the optimal state given in Theorem \ref{thm:tight_upper_bound}.

In the following, we will see that the SDP lower bound does not have these characteristics. 

\begin{figure}[t!]
    \centering
    \begin{subfigure}[b]{0.5\textwidth}
        \centering
        \includegraphics[width=80mm]{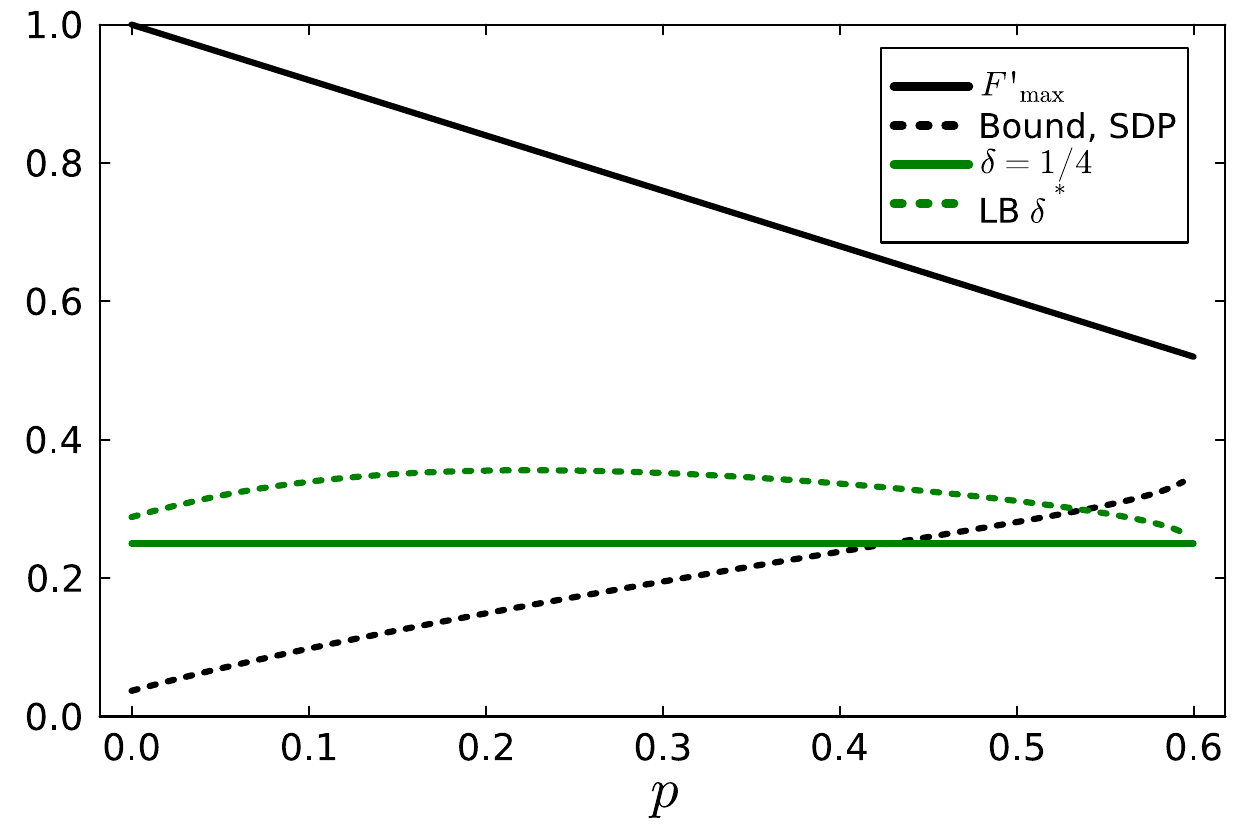}
        \caption{}
    \end{subfigure}%
    ~ 
    \begin{subfigure}[b]{0.5\textwidth}
        \centering
        \includegraphics[width=80mm]{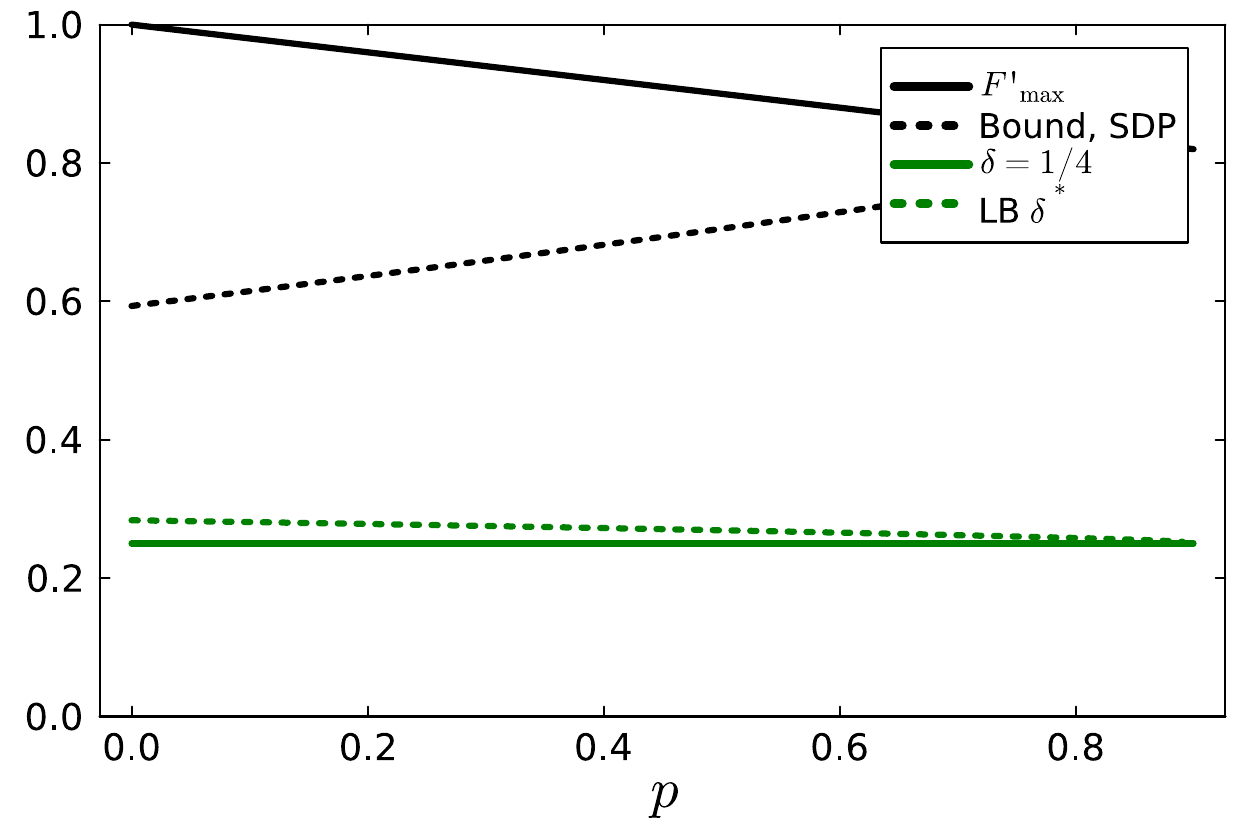}
        \caption{}
    \end{subfigure}
            \caption{\textbf{Probability of the optimal postselected swap outcome plotted against $p$}, for (a) $F=0.6$ and (b) $F=0.9$. Each is plotted for 100 values of $p$, uniformly spaced in the interval $[0,F]$. The black solid line is $F'_{\max}(p,F)$, and the green solid line is the postselected swap probability for the states saturating this value ($\delta = 1/4$). The black dotted line is the lower bound found with SDP, and the green dotted line is the postselected swap probability of the optimal state for the lower bound.
            \Guus{I think it's a little confusing that there is both F and $\delta$ in one figure, especially since the Y axis says fidelity. Maybe have the right-hand-side Y axis say post-swap fidelity or something?}
            }
            \label{fig:bounds_and_optimal_delta}
\end{figure}

\subsection{Optimal value of $\delta$}
Despite the fact that the SDP upper bound is tight and has a simple analytical form, in the case of no permutation symmetry we did not observe the same for the SDP lower bound. In order to further understand its behaviour, one may analyse the value of the postselected swap probability $\delta^*$ that minimises the expression of the lower bound (\ref{eqn:lower_bound_with_SDP}), i.e. 
\begin{align}
    F'_{\min}(p,F) &\geq \min_{\delta} \frac{1}{\delta}\left( \frac{Fp}{2} - \frac{p^2}{4} +(1-p)^2 H_{\mathrm{rel}}^*(p,F,\delta)\right) \label{eqn:lower_bound_with_SDP_appendix2} \\ &= \frac{1}{\delta^*}\left( \frac{Fp}{2} - \frac{p^2}{4} +(1-p)^2 H_{\mathrm{rel}}^*(p,F,\delta^*)\right). \label{eqn:delta_star}
\end{align}
Recalling that the states saturating the upper bound have constant postselected swap probability of $ 1/4$, we see from Figure \ref{fig:bounds_and_optimal_delta} that $\delta^*$ usually lies above this value, and is not constant. We were not able to find a good functional fit in terms of $p$ and $F$ for $\delta^*$. Similarly, we were not able to find a functional fit for the SDP lower bound: although for large values of $F$ this appears to be linear in some range of $p$ (see Figures \ref{fig:p_vs_F'_1} and \ref{fig:p_vs_F'_2} in the main text), we see from Figure \ref{fig:bounds_and_optimal_delta} that for $F=0.6$ this is not the case. We leave further analysis of the behaviour of $\delta^*$ to future work.

% \subsection{Full permutation symmetry}
% \todo{Decide whether to include}

\subsection{Dependence on $\delta$}
\begin{figure}[t]
\includegraphics[width=90mm]{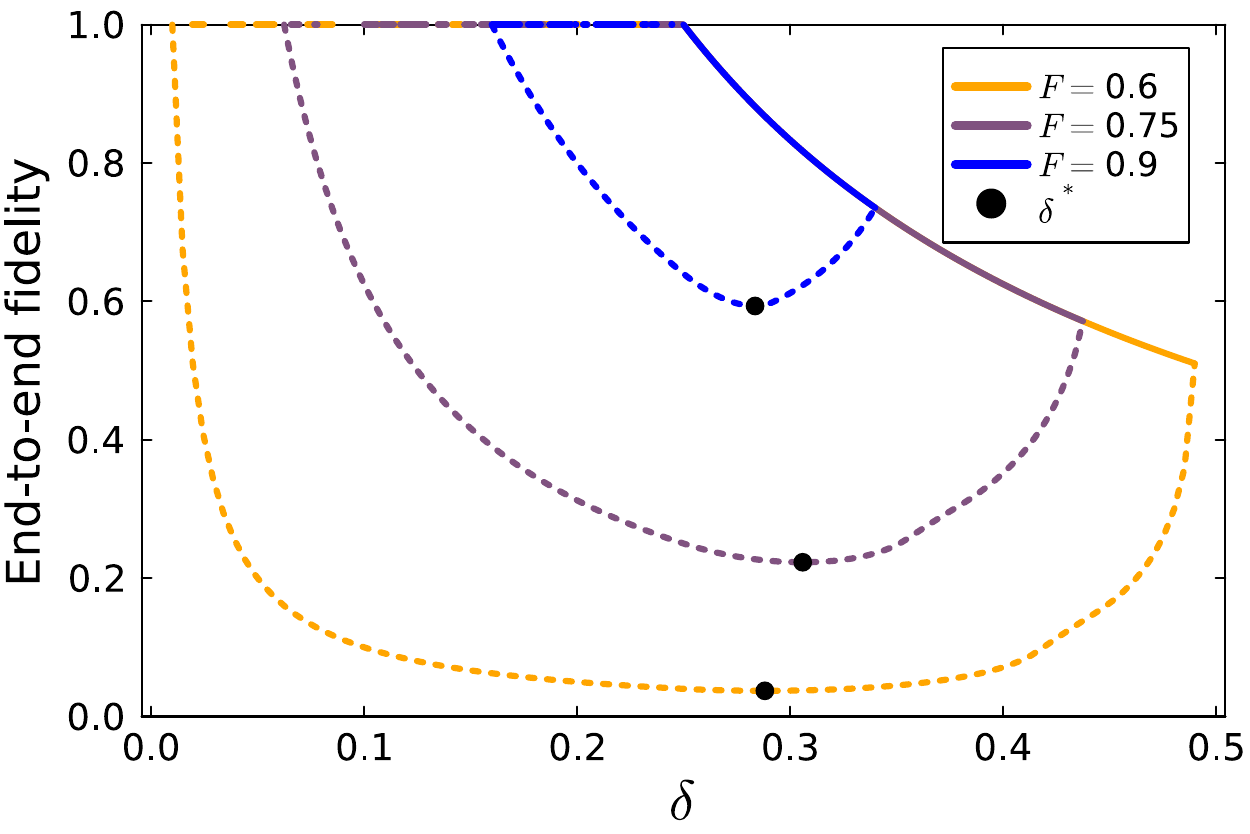}
\caption{\label{fig:show_delta_variation} \textbf{Bounds on the end-to-end fidelity, given a postselected swap probability $\delta$.} For three different values of $F$ and $p=0$, each is plotted for 100 values of $\delta$, uniformly spaced in the feasible region $[\delta_{\min},\delta_{\max}]$ for each value of $F$. The dotted lines are the lower bounds, bound by solving (\ref{opt:PPT1_appendix}). The solid lines are the upper bounds, found by solving the corresponding maximisation problem of (\ref{opt:PPT1_appendix}).
Also shown are the values $\delta^*$ that form the SDP lower bound for $F'_{\min}(p,F)$ from (\ref{eqn:delta_star}).}
\end{figure}

It was mentioned in the main text that fixing the parameter $\delta$ while optimising the end-to-end fidelity can aid to further understand the trade-off between rate and fidelity inherent to the entanglement swapping process. In particular, solving the problem (\ref{opt:PPT1}) provides an answer to the question `given that the probability of the postselected swap is $\delta$, how small (large) can my fidelity become after swapping'? In order to answer how small it can become, one may solve the following semi-definite program
\begin{mini}|s|
{\sigma }{ \mathrm{Tr}\big[ \ketbra{\Psi_{00}}_{B_1 B_2} \ketbra{\Psi_{00}}_{A_1 A_2} \sigma \big]}
{}{}
\addConstraint{ \mathrm{Tr}\big[\ketbra{\Psi_{00}}_{A_1 A_2} \sigma \big] = \tilde{\delta}(p,\delta)}
\addConstraint{\mathrm{Tr}\big[\ketbra{\Psi_{00}}_{B_1 A_1} \sigma \big] = \tilde{F}(p,F)}
\addConstraint{\mathrm{Tr}\big[\ketbra{\Psi_{00}}_{B_2 A_2} \sigma \big] = \tilde{F}(p,F)}
\addConstraint{\mathrm{Tr}[\sigma]= 1}
\addConstraint{\sigma\geq 0, \;\;\; \sigma^{\Gamma}\geq 0.}
\label{opt:PPT1_appendix}
\end{mini}
This was also given in (\ref{opt:PPT1}) in the main text. To find bounds on how large the fidelity can become, one may perform instead a maximisation of (\ref{opt:PPT1_appendix}), or simply replace the objective function by a minus sign. These bounds are shown in Figure \ref{fig:show_delta_variation}. These are plotted for $p=0$ and three different values of $F$, in the feasible range of $\delta$. Interestingly, the upper bounds for each of the three fidelity values always coincide in their corresponding feasible region.
This extends the observation from the beginning of Section \ref{sec:advantages_full}, where we saw that no matter the value of the initial fidelity, it is always possible to obtain a unit end-to-end fidelity (from Figure \ref{fig:show_delta_variation}, we observe this for values of $\delta$ below $1/4$, where all upper bounds are equal to one). Recalling the state $\ket{\psi} = \sqrt{F} \ket{\Psi_{00}} + \sqrt{1-F} \ket{\Psi_{11}}$ that swaps to unit fidelity with probability $1/4$, we therefore conclude that at the point $\delta=1/4$ the upper bound is tight.
Then, $\ket{\psi}$ is optimal in the sense that it has the maximum probability of swapping to perfect fidelity. Indeed, beyond $\delta = 1/4$, we see there is necessarily a decrease in the end-to-end fidelity if we demand that $\delta$ is larger than this value. It can also be seen from the figure that at the extremal values of the feasible region of $\delta$, the upper and lower bounds meet. In particular, when $\delta$ becomes close to $\delta_{\min}$, the lower bound goes to one. We conclude that if the postselected swap probability is made as small as possible, the end-to-end fidelity will necessarily increase to one. Despite this, the lower bound is not monotonic in $\delta$: for large values of $\delta$, we see from the figure that it increases again before joining up with the upper bound at $\delta_{\max}$. From this behaviour we may conclude that, for the values of $F$ tested, the numerical optimiser over $\delta$ that is employed in (\ref{eqn:lower_bound_with_SDP_appendix2}) to find the SDP lower bound is indeed finding the global minimum. This is highlighted by the black circles in Figure \ref{fig:show_delta_variation}, which are the values $\delta^*$ of the postselected swap probability that minimise the lower bound. The value of the objective function at each point is the SDP lower bound for $F_{\min}(0,F)$.

\newpage
\section{Invariance of secret-key fraction under Bell-diagonal twirling}
\label{app:SKF}
As discussed in Section \ref{sec:qkd}, the secret-key fraction depends on the QBER according to (\ref{eqn:skf}). The QBER is defined as the probability that, when both nodes measure their state in the X (Z) basis, they obtain different outcomes. Letting $\sigma$ be the entangled state shared between the two nodes, the probability of obtaining an error when measuring in the $Z$ basis is given by
\begin{align}
   Q_Z &=  \bra{01}\sigma\ket{01} + \bra{10}\sigma\ket{10} %\\  &= 1-(\bra{00}\rho\ket{00} + \bra{11}\rho\ket{11})
\end{align}
Noting that $$\ketbra{01} + \ketbra{10} = \frac{1}{2}\left(I_4 - Z\otimes Z\right)\sigma ,$$ we have
\begin{align}
    Q_Z(\sigma) &= \Tr\left[(\ketbra{01} + \ketbra{10}) \sigma \right] \\ 
    &=\frac{1}{2} \Tr\left[ \left(I_4 - Z\otimes Z\right)\sigma \right] \\
    &= \frac{1}{2} (1-\Tr\left[ (Z\otimes Z) \sigma \right]).
\end{align}
Now, since we have $(Z\otimes Z)\ket{\Psi_{ij}} = \pm \ket{\Psi_{ij}}$ for all Bell states $\ket{\Psi_{ij}}$, we see that $\Tr\left[ (Z\otimes Z) \sigma \right]$ only depends on the Bell-diagonal elements of $\sigma$, and therefore $Q_Z(\sigma) = Q_Z(\mathcal{B}(\sigma))$. The same holds for measuring in the $X$-basis:
\begin{align}
    Q_X(\sigma) &= \frac{1}{2} (1-\Tr\left[ (X\otimes X) \sigma \right]) \\ &= \frac{1}{2} (1-\Tr\left[ (X\otimes X) \mathcal{B}(\sigma) \right]) \\ &= Q_X(\mathcal{B}(\sigma)).
\end{align}
By (\ref{eqn:skf}), we therefore have $\text{SKF}(\sigma) = \text{SKF}(\mathcal{B}(\sigma))$.

\end{document}